\documentclass[11pt]{article}

\usepackage[margin=1in]{geometry} 

\usepackage{graphicx}             
\usepackage{amsmath}              
\usepackage{booktabs}             
\usepackage[round]{natbib}
\usepackage[linesnumbered,ruled,vlined]{algorithm2e}
\usepackage{amsmath,amssymb,amsthm}
\usepackage{xcolor}
\usepackage{colortbl}
\usepackage{float}
\usepackage{stfloats} 
\usepackage{booktabs}
\usepackage{graphicx}
\usepackage{multirow}
\usepackage{tikz}
\usepackage{hyperref}
\usepackage{bbm}
\usepackage{mathtools}
\usepackage{graphicx}
\usepackage{subcaption}
\usepackage{siunitx}
\usepackage{enumitem}
\usepackage[
    labelfont=bf,                 
    justification=justified,      
    singlelinecheck=false         
]{caption}                        
\usepackage{float}                
\usepackage[utf8]{inputenc}       
\usepackage{times}      

\newtheorem{assumption}{Assumption}

\newtheorem{lemma}{Lemma}
\newtheorem{theorem}{Theorem}
\newtheorem*{definition*}{Definition}
\newtheorem*{example}{Example}

\usetikzlibrary{
    shapes.geometric,
    arrows.meta,
    positioning,
    shadows,
    calc,
    backgrounds,
    fit
}

\title{Conformal Prediction Assessment:\\ A Framework for Conditional Coverage Evaluation and Selection}
\author{Zheng Zhou\thanks{These authors are co-first authors and  contributed equally to this work.}\hspace{.2cm}\\
    Qiuzhen College, Tsinghua University\\
    and \\
    Xiangfei Zhang\footnotemark[1] \hspace{.2cm}\\
    Qiuzhen College, Tsinghua University\\
    and \\
    Chongguang Tao\footnotemark[1] \hspace{.2cm}\\
    Qiuzhen College, Tsinghua University\\
    and \\
    Yuhong Yang\thanks{Corresponding author: yyangsc@tsinghua.edu.cn} \hspace{.2cm}\\
    Yau Mathematical Sciences Center, Tsinghua University}

\begin{document}



\newcommand{\figleft}{{\em (Left)}}
\newcommand{\figcenter}{{\em (Center)}}
\newcommand{\figright}{{\em (Right)}}
\newcommand{\figtop}{{\em (Top)}}
\newcommand{\figbottom}{{\em (Bottom)}}
\newcommand{\captiona}{{\em (a)}}
\newcommand{\captionb}{{\em (b)}}
\newcommand{\captionc}{{\em (c)}}
\newcommand{\captiond}{{\em (d)}}

\newcommand{\newterm}[1]{{\bf #1}}

\def\figref#1{figure~\ref{#1}}
\def\Figref#1{Figure~\ref{#1}}
\def\twofigref#1#2{figures \ref{#1} and \ref{#2}}
\def\quadfigref#1#2#3#4{figures \ref{#1}, \ref{#2}, \ref{#3} and \ref{#4}}
\def\secref#1{section~\ref{#1}}
\def\Secref#1{Section~\ref{#1}}
\def\twosecrefs#1#2{sections \ref{#1} and \ref{#2}}
\def\secrefs#1#2#3{sections \ref{#1}, \ref{#2} and \ref{#3}}
\def\eqref#1{equation~\ref{#1}}
\def\Eqref#1{Equation~\ref{#1}}
\def\plaineqref#1{\ref{#1}}
\def\chapref#1{chapter~\ref{#1}}
\def\Chapref#1{Chapter~\ref{#1}}
\def\rangechapref#1#2{chapters\ref{#1}--\ref{#2}}
\def\algref#1{algorithm~\ref{#1}}
\def\Algref#1{Algorithm~\ref{#1}}
\def\twoalgref#1#2{algorithms \ref{#1} and \ref{#2}}
\def\Twoalgref#1#2{Algorithms \ref{#1} and \ref{#2}}
\def\partref#1{part~\ref{#1}}
\def\Partref#1{Part~\ref{#1}}
\def\twopartref#1#2{parts \ref{#1} and \ref{#2}}

\def\ceil#1{\lceil #1 \rceil}
\def\floor#1{\lfloor #1 \rfloor}
\def\1{\bm{1}}
\newcommand{\train}{\mathcal{D}}
\newcommand{\valid}{\mathcal{D_{\mathrm{valid}}}}
\newcommand{\test}{\mathcal{D_{\mathrm{test}}}}

\def\eps{{\epsilon}}

\def\reta{{\textnormal{$\eta$}}}
\def\ra{{\textnormal{a}}}
\def\rb{{\textnormal{b}}}
\def\rc{{\textnormal{c}}}
\def\rd{{\textnormal{d}}}
\def\re{{\textnormal{e}}}
\def\rf{{\textnormal{f}}}
\def\rg{{\textnormal{g}}}
\def\rh{{\textnormal{h}}}
\def\ri{{\textnormal{i}}}
\def\rj{{\textnormal{j}}}
\def\rk{{\textnormal{k}}}
\def\rl{{\textnormal{l}}}
\def\rn{{\textnormal{n}}}
\def\ro{{\textnormal{o}}}
\def\rp{{\textnormal{p}}}
\def\rq{{\textnormal{q}}}
\def\rr{{\textnormal{r}}}
\def\rs{{\textnormal{s}}}
\def\rt{{\textnormal{t}}}
\def\ru{{\textnormal{u}}}
\def\rv{{\textnormal{v}}}
\def\rw{{\textnormal{w}}}
\def\rx{{\textnormal{x}}}
\def\ry{{\textnormal{y}}}
\def\rz{{\textnormal{z}}}

\def\rvepsilon{{\mathbf{\epsilon}}}
\def\rvtheta{{\mathbf{\theta}}}
\def\rva{{\mathbf{a}}}
\def\rvb{{\mathbf{b}}}
\def\rvc{{\mathbf{c}}}
\def\rvd{{\mathbf{d}}}
\def\rve{{\mathbf{e}}}
\def\rvf{{\mathbf{f}}}
\def\rvg{{\mathbf{g}}}
\def\rvh{{\mathbf{h}}}
\def\rvu{{\mathbf{i}}}
\def\rvj{{\mathbf{j}}}
\def\rvk{{\mathbf{k}}}
\def\rvl{{\mathbf{l}}}
\def\rvm{{\mathbf{m}}}
\def\rvn{{\mathbf{n}}}
\def\rvo{{\mathbf{o}}}
\def\rvp{{\mathbf{p}}}
\def\rvq{{\mathbf{q}}}
\def\rvr{{\mathbf{r}}}
\def\rvs{{\mathbf{s}}}
\def\rvt{{\mathbf{t}}}
\def\rvu{{\mathbf{u}}}
\def\rvv{{\mathbf{v}}}
\def\rvw{{\mathbf{w}}}
\def\rvx{{\mathbf{x}}}
\def\rvy{{\mathbf{y}}}
\def\rvz{{\mathbf{z}}}

\def\erva{{\textnormal{a}}}
\def\ervb{{\textnormal{b}}}
\def\ervc{{\textnormal{c}}}
\def\ervd{{\textnormal{d}}}
\def\erve{{\textnormal{e}}}
\def\ervf{{\textnormal{f}}}
\def\ervg{{\textnormal{g}}}
\def\ervh{{\textnormal{h}}}
\def\ervi{{\textnormal{i}}}
\def\ervj{{\textnormal{j}}}
\def\ervk{{\textnormal{k}}}
\def\ervl{{\textnormal{l}}}
\def\ervm{{\textnormal{m}}}
\def\ervn{{\textnormal{n}}}
\def\ervo{{\textnormal{o}}}
\def\ervp{{\textnormal{p}}}
\def\ervq{{\textnormal{q}}}
\def\ervr{{\textnormal{r}}}
\def\ervs{{\textnormal{s}}}
\def\ervt{{\textnormal{t}}}
\def\ervu{{\textnormal{u}}}
\def\ervv{{\textnormal{v}}}
\def\ervw{{\textnormal{w}}}
\def\ervx{{\textnormal{x}}}
\def\ervy{{\textnormal{y}}}
\def\ervz{{\textnormal{z}}}

\def\rmA{{\mathbf{A}}}
\def\rmB{{\mathbf{B}}}
\def\rmC{{\mathbf{C}}}
\def\rmD{{\mathbf{D}}}
\def\rmE{{\mathbf{E}}}
\def\rmF{{\mathbf{F}}}
\def\rmG{{\mathbf{G}}}
\def\rmH{{\mathbf{H}}}
\def\rmI{{\mathbf{I}}}
\def\rmJ{{\mathbf{J}}}
\def\rmK{{\mathbf{K}}}
\def\rmL{{\mathbf{L}}}
\def\rmM{{\mathbf{M}}}
\def\rmN{{\mathbf{N}}}
\def\rmO{{\mathbf{O}}}
\def\rmP{{\mathbf{P}}}
\def\rmQ{{\mathbf{Q}}}
\def\rmR{{\mathbf{R}}}
\def\rmS{{\mathbf{S}}}
\def\rmT{{\mathbf{T}}}
\def\rmU{{\mathbf{U}}}
\def\rmV{{\mathbf{V}}}
\def\rmW{{\mathbf{W}}}
\def\rmX{{\mathbf{X}}}
\def\rmY{{\mathbf{Y}}}
\def\rmZ{{\mathbf{Z}}}

\def\ermA{{\textnormal{A}}}
\def\ermB{{\textnormal{B}}}
\def\ermC{{\textnormal{C}}}
\def\ermD{{\textnormal{D}}}
\def\ermE{{\textnormal{E}}}
\def\ermF{{\textnormal{F}}}
\def\ermG{{\textnormal{G}}}
\def\ermH{{\textnormal{H}}}
\def\ermI{{\textnormal{I}}}
\def\ermJ{{\textnormal{J}}}
\def\ermK{{\textnormal{K}}}
\def\ermL{{\textnormal{L}}}
\def\ermM{{\textnormal{M}}}
\def\ermN{{\textnormal{N}}}
\def\ermO{{\textnormal{O}}}
\def\ermP{{\textnormal{P}}}
\def\ermQ{{\textnormal{Q}}}
\def\ermR{{\textnormal{R}}}
\def\ermS{{\textnormal{S}}}
\def\ermT{{\textnormal{T}}}
\def\ermU{{\textnormal{U}}}
\def\ermV{{\textnormal{V}}}
\def\ermW{{\textnormal{W}}}
\def\ermX{{\textnormal{X}}}
\def\ermY{{\textnormal{Y}}}
\def\ermZ{{\textnormal{Z}}}

\def\vzero{{\bm{0}}}
\def\vone{{\bm{1}}}
\def\vmu{{\bm{\mu}}}
\def\vtheta{{\bm{\theta}}}
\def\va{{\bm{a}}}
\def\vb{{\bm{b}}}
\def\vc{{\bm{c}}}
\def\vd{{\bm{d}}}
\def\ve{{\bm{e}}}
\def\vf{{\bm{f}}}
\def\vg{{\bm{g}}}
\def\vh{{\bm{h}}}
\def\vi{{\bm{i}}}
\def\vj{{\bm{j}}}
\def\vk{{\bm{k}}}
\def\vl{{\bm{l}}}
\def\vm{{\bm{m}}}
\def\vn{{\bm{n}}}
\def\vo{{\bm{o}}}
\def\vp{{\bm{p}}}
\def\vq{{\bm{q}}}
\def\vr{{\bm{r}}}
\def\vs{{\bm{s}}}
\def\vt{{\bm{t}}}
\def\vu{{\bm{u}}}
\def\vv{{\bm{v}}}
\def\vw{{\bm{w}}}
\def\vx{{\bm{x}}}
\def\vy{{\bm{y}}}
\def\vz{{\bm{z}}}

\def\evalpha{{\alpha}}
\def\evbeta{{\beta}}
\def\evepsilon{{\epsilon}}
\def\evlambda{{\lambda}}
\def\evomega{{\omega}}
\def\evmu{{\mu}}
\def\evpsi{{\psi}}
\def\evsigma{{\sigma}}
\def\evtheta{{\theta}}
\def\eva{{a}}
\def\evb{{b}}
\def\evc{{c}}
\def\evd{{d}}
\def\eve{{e}}
\def\evf{{f}}
\def\evg{{g}}
\def\evh{{h}}
\def\evi{{i}}
\def\evj{{j}}
\def\evk{{k}}
\def\evl{{l}}
\def\evm{{m}}
\def\evn{{n}}
\def\evo{{o}}
\def\evp{{p}}
\def\evq{{q}}
\def\evr{{r}}
\def\evs{{s}}
\def\evt{{t}}
\def\evu{{u}}
\def\evv{{v}}
\def\evw{{w}}
\def\evx{{x}}
\def\evy{{y}}
\def\evz{{z}}

\def\mA{{\bm{A}}}
\def\mB{{\bm{B}}}
\def\mC{{\bm{C}}}
\def\mD{{\bm{D}}}
\def\mE{{\bm{E}}}
\def\mF{{\bm{F}}}
\def\mG{{\bm{G}}}
\def\mH{{\bm{H}}}
\def\mI{{\bm{I}}}
\def\mJ{{\bm{J}}}
\def\mK{{\bm{K}}}
\def\mL{{\bm{L}}}
\def\mM{{\bm{M}}}
\def\mN{{\bm{N}}}
\def\mO{{\bm{O}}}
\def\mP{{\bm{P}}}
\def\mQ{{\bm{Q}}}
\def\mR{{\bm{R}}}
\def\mS{{\bm{S}}}
\def\mT{{\bm{T}}}
\def\mU{{\bm{U}}}
\def\mV{{\bm{V}}}
\def\mW{{\bm{W}}}
\def\mX{{\bm{X}}}
\def\mY{{\bm{Y}}}
\def\mZ{{\bm{Z}}}
\def\mBeta{{\bm{\beta}}}
\def\mPhi{{\bm{\Phi}}}
\def\mLambda{{\bm{\Lambda}}}
\def\mSigma{{\bm{\Sigma}}}

\newcommand{\tens}[1]{\bm{\mathsfit{#1}}}
\def\tA{{\tens{A}}}
\def\tB{{\tens{B}}}
\def\tC{{\tens{C}}}
\def\tD{{\tens{D}}}
\def\tE{{\tens{E}}}
\def\tF{{\tens{F}}}
\def\tG{{\tens{G}}}
\def\tH{{\tens{H}}}
\def\tI{{\tens{I}}}
\def\tJ{{\tens{J}}}
\def\tK{{\tens{K}}}
\def\tL{{\tens{L}}}
\def\tM{{\tens{M}}}
\def\tN{{\tens{N}}}
\def\tO{{\tens{O}}}
\def\tP{{\tens{P}}}
\def\tQ{{\tens{Q}}}
\def\tR{{\tens{R}}}
\def\tS{{\tens{S}}}
\def\tT{{\tens{T}}}
\def\tU{{\tens{U}}}
\def\tV{{\tens{V}}}
\def\tW{{\tens{W}}}
\def\tX{{\tens{X}}}
\def\tY{{\tens{Y}}}
\def\tZ{{\tens{Z}}}

\def\gA{{\mathcal{A}}}
\def\gB{{\mathcal{B}}}
\def\gC{{\mathcal{C}}}
\def\gD{{\mathcal{D}}}
\def\gE{{\mathcal{E}}}
\def\gF{{\mathcal{F}}}
\def\gG{{\mathcal{G}}}
\def\gH{{\mathcal{H}}}
\def\gI{{\mathcal{I}}}
\def\gJ{{\mathcal{J}}}
\def\gK{{\mathcal{K}}}
\def\gL{{\mathcal{L}}}
\def\gM{{\mathcal{M}}}
\def\gN{{\mathcal{N}}}
\def\gO{{\mathcal{O}}}
\def\gP{{\mathcal{P}}}
\def\gQ{{\mathcal{Q}}}
\def\gR{{\mathcal{R}}}
\def\gS{{\mathcal{S}}}
\def\gT{{\mathcal{T}}}
\def\gU{{\mathcal{U}}}
\def\gV{{\mathcal{V}}}
\def\gW{{\mathcal{W}}}
\def\gX{{\mathcal{X}}}
\def\gY{{\mathcal{Y}}}
\def\gZ{{\mathcal{Z}}}

\def\sA{{\mathbb{A}}}
\def\sB{{\mathbb{B}}}
\def\sC{{\mathbb{C}}}
\def\sD{{\mathbb{D}}}
\def\sF{{\mathbb{F}}}
\def\sG{{\mathbb{G}}}
\def\sH{{\mathbb{H}}}
\def\sI{{\mathbb{I}}}
\def\sJ{{\mathbb{J}}}
\def\sK{{\mathbb{K}}}
\def\sL{{\mathbb{L}}}
\def\sM{{\mathbb{M}}}
\def\sN{{\mathbb{N}}}
\def\sO{{\mathbb{O}}}
\def\sP{{\mathbb{P}}}
\def\sQ{{\mathbb{Q}}}
\def\sR{{\mathbb{R}}}
\def\sS{{\mathbb{S}}}
\def\sT{{\mathbb{T}}}
\def\sU{{\mathbb{U}}}
\def\sV{{\mathbb{V}}}
\def\sW{{\mathbb{W}}}
\def\sX{{\mathbb{X}}}
\def\sY{{\mathbb{Y}}}
\def\sZ{{\mathbb{Z}}}

\def\emLambda{{\Lambda}}
\def\emA{{A}}
\def\emB{{B}}
\def\emC{{C}}
\def\emD{{D}}
\def\emE{{E}}
\def\emF{{F}}
\def\emG{{G}}
\def\emH{{H}}
\def\emI{{I}}
\def\emJ{{J}}
\def\emK{{K}}
\def\emL{{L}}
\def\emM{{M}}
\def\emN{{N}}
\def\emO{{O}}
\def\emP{{P}}
\def\emQ{{Q}}
\def\emR{{R}}
\def\emS{{S}}
\def\emT{{T}}
\def\emU{{U}}
\def\emV{{V}}
\def\emW{{W}}
\def\emX{{X}}
\def\emY{{Y}}
\def\emZ{{Z}}
\def\emSigma{{\Sigma}}

\newcommand{\etens}[1]{\mathsfit{#1}}
\def\etLambda{{\etens{\Lambda}}}
\def\etA{{\etens{A}}}
\def\etB{{\etens{B}}}
\def\etC{{\etens{C}}}
\def\etD{{\etens{D}}}
\def\etE{{\etens{E}}}
\def\etF{{\etens{F}}}
\def\etG{{\etens{G}}}
\def\etH{{\etens{H}}}
\def\etI{{\etens{I}}}
\def\etJ{{\etens{J}}}
\def\etK{{\etens{K}}}
\def\etL{{\etens{L}}}
\def\etM{{\etens{M}}}
\def\etN{{\etens{N}}}
\def\etO{{\etens{O}}}
\def\etP{{\etens{P}}}
\def\etQ{{\etens{Q}}}
\def\etR{{\etens{R}}}
\def\etS{{\etens{S}}}
\def\etT{{\etens{T}}}
\def\etU{{\etens{U}}}
\def\etV{{\etens{V}}}
\def\etW{{\etens{W}}}
\def\etX{{\etens{X}}}
\def\etY{{\etens{Y}}}
\def\etZ{{\etens{Z}}}

\newcommand{\pdata}{p_{\rm{data}}}
\newcommand{\ptrain}{\hat{p}_{\rm{data}}}
\newcommand{\Ptrain}{\hat{P}_{\rm{data}}}
\newcommand{\pmodel}{p_{\rm{model}}}
\newcommand{\Pmodel}{P_{\rm{model}}}
\newcommand{\ptildemodel}{\tilde{p}_{\rm{model}}}
\newcommand{\pencode}{p_{\rm{encoder}}}
\newcommand{\pdecode}{p_{\rm{decoder}}}
\newcommand{\precons}{p_{\rm{reconstruct}}}

\newcommand{\laplace}{\mathrm{Laplace}} 

\newcommand{\E}{\mathbb{E}}
\newcommand{\Ls}{\mathcal{L}}
\newcommand{\R}{\mathbb{R}}
\newcommand{\emp}{\tilde{p}}
\newcommand{\lr}{\alpha}
\newcommand{\reg}{\lambda}
\newcommand{\rect}{\mathrm{rectifier}}
\newcommand{\softmax}{\mathrm{softmax}}
\newcommand{\sigmoid}{\sigma}
\newcommand{\softplus}{\zeta}
\newcommand{\KL}{D_{\mathrm{KL}}}
\newcommand{\Var}{\mathrm{Var}}
\newcommand{\standarderror}{\mathrm{SE}}
\newcommand{\Cov}{\mathrm{Cov}}
\newcommand{\normlzero}{L^0}
\newcommand{\normlone}{L^1}
\newcommand{\normltwo}{L^2}
\newcommand{\normlp}{L^p}
\newcommand{\normmax}{L^\infty}

\newcommand{\parents}{Pa} 

\let\ab\allowbreak 

\maketitle

\begin{abstract}
Conformal prediction provides rigorous distribution-free finite-sample guarantees for marginal coverage under the assumption of exchangeability, but may exhibit systematic undercoverage or overcoverage for specific subpopulations. Assessing conditional validity is challenging, as standard stratification methods suffer from the curse of dimensionality. We propose Conformal Prediction Assessment (CPA), a framework that reframes the evaluation of conditional coverage as a supervised learning task by training a reliability estimator that predicts instance-level coverage probabilities. Building on this estimator, we introduce the Conditional Validity Index (CVI), which decomposes reliability into safety (undercoverage risk) and efficiency (overcoverage cost). We establish convergence rates for the reliability estimator and prove the consistency of CVI-based model selection. Extensive experiments on synthetic and real-world  datasets demonstrate that CPA effectively diagnoses local failure modes and that CC-Select, our CVI-based model selection algorithm, consistently identifies predictors with superior conditional coverage performance.
\end{abstract}

\noindent%
\vfill

\section{Introduction}\label{sec-intro}

Predictive models based on artificial intelligence (AI) and machine learning (ML) are increasingly deployed in high-stakes decision-making systems, ranging from healthcare and precision medicine \citep{topol2019high} to criminal justice \citep{berk2021fairness} and financial risk assessment \citep{fuster2022predictably}. As the consequences of an erroneous prediction are typically borne by specific individuals rather than by the population on average, it is crucial to develop tools for reliable uncertainty quantification. 
A growing body of literature emphasizes that despite achieving impressive average-level accuracy, modern models often exhibit substantial uncertainty heterogeneity across the feature space. Due to factors such as uneven data coverage, distribution shifts, or complex covariate interactions, this heterogeneity can lead to systematically overconfident or underconfident predictions for specific subpopulations \citep{chakraborti2025personalized}. \looseness=-1

\subsection{Uncertainty Quantification via Conformal Prediction}

Conformal prediction (CP) addresses the challenge of uncertainty quantification by providing a rigorous, distribution-free and model-agnostic framework that yields valid predictive intervals with finite-sample marginal coverage guarantees under the assumption of exchangeability \citep{Vovk2005}. 
The most widely adopted variant, split-conformal prediction, operates by partitioning the available data $\mathcal{D}$ into a training set $\mathcal{D}_{\text{pred}}= \{(X_i, Y_i)\}_{i \in \gI_\text{pred}}$ and a calibration set $\mathcal{D}_{\text{calib}} = \{(X_i, Y_i)\}_{i \in \gI_\text{calib}}$, with sizes $n_{\text{pred}} = |\mathcal{I}_{\text{pred}}|$ and $n_{\text{calib}} = |\mathcal{I}_{\text{calib}}|$, respectively.
First, a predictive model $\hat{\mu}$ is fitted on $\mathcal{D}_{\text{pred}}$. Then, a non-conformity score function $S(x, y) := S(x, y ; \mathcal{D}_{\text{pred}})$---typically the absolute residual $|y - \hat{\mu}(x)|$---is evaluated on the calibration set. 
 For a user-specified miscoverage rate $\alpha \in (0, 1)$, the prediction set for a new test point $x \in \mathcal{X}$ is constructed as:
\begin{align}
    \mathcal{C}_\alpha(x ; \mathcal{D}) = \{ y \in \mathcal{Y} : S(x, y ; \mathcal{D}_{\text{pred}}) \le \hat{q}_{n_\text{calib}, 1-\alpha} \},
\end{align}
where $\hat{q}_{n_\text{calib},1-\alpha}$ denotes the $\lceil (n_{\text{calib}} + 1)(1 - \alpha) \rceil$-th smallest value among the calibration scores $\{S_i := S(X_i, Y_i)\}_{i \in \mathcal{I}_{\text{calib}}}$.
Assuming the calibration data and the test point $(X_{n+1}, Y_{n+1})$ are exchangeable, split-conformal prediction guarantees \textbf{marginal coverage}:
\begin{align}
    \label{eq:MarginalCoverage}
    \mathbb{P}\left( Y_{n+1} \in \mathcal{C}_\alpha(X_{n+1}; \mathcal{D} )\right) \ge 1 - \alpha.    
\end{align}

While the marginal guarantee is desirable, it is an average property taken over both the randomness of the calibration set and the test point. However, conformal prediction may have significant performance discrepancies across the feature space. This limitation motivates the pursuit of \textbf{conditional coverage}, a stronger property requiring validity for any given feature vector $x$:
\begin{align}
    \mathbb{P}(Y_{n+1} \in \mathcal{C}_\alpha(X_{n+1}; \mathcal{D}) \mid X_{n+1}=x) \ge 1-\alpha.
\end{align}
Although asymptotic conditional coverage is attainable under specific distributional assumptions \citep{lei2014distribution,lei2018distribution}, exact conditional coverage in a finite-sample and distribution-free setting is theoretically impossible \citep{vovk2012conditional,lei2014distribution,foygel2021limits}.
This theoretical barrier implies that, without additional assumptions, achieving exact conditional coverage is infeasible. Therefore, it is imperative to develop a sound approach for assessing the conditional coverage behavior of a CP algorithm based on the data at hand.

\subsection{Our Contributions}
To address the critical lack of rigorous tools for verifying conditional validity, we propose \textbf{Conformal Prediction Assessment (CPA)}, a general data-driven framework that transforms the problem of conditional coverage assessment into a concrete supervised learning task. By training a reliability estimator $\hat{\eta}(x)$ to serve as a proxy for the unobservable conditional coverage probability, CPA enables a granular, instance-level examination of uncertainty models. Our specific contributions are threefold:

\begin{enumerate}
    \item \textbf{A rigorous, flexible, and informative framework for auditing conditional coverage.} CPA estimates the conditional coverage function $\hat{\eta}(x)$ using a data-adaptive probabilistic learner, providing a flexible alternative to coarse binning-based assessments and avoiding manual stratification. Beyond a single summary statistic, we decompose conditional coverage assessment into \textbf{Safety} (undercoverage risk) and \textbf{Efficiency} (overcoverage cost), yielding informative diagnostics of model failure modes. We further introduce the \textbf{Conditional Validity Profile (CVP)} curve, a visualization tool that summarizes the distribution of local conditional coverage probabilities. \looseness=-1

    \item \textbf{Theoretical guarantees of consistency.} We establish the convergence rates of the reliability estimator and prove that our proposed metric, the Conditional Validity Index (CVI), is a consistent estimator of the true conditional miscalibration. Furthermore, we characterize the trade-off between estimation and approximation errors, providing theoretical guidance for optimal data splitting.
    
    \item \textbf{Reliability-driven diagnosis and model selection (CC-Select).} CPA provides a principled basis for diagnosing where a conformal predictor fails, thereby indicating which methodological refinements are worth exploring. It also enables the comparison of practically relevant conformal prediction methods or variants based on conditional validity rather than predictive efficiency alone; we refer to this selection framework as \textbf{CC-Select}. We demonstrate through extensive experiments that CC-Select consistently identifies models with superior conditional coverage performance. 
\end{enumerate}

\subsection{Related Work}

\textbf{Conditional Coverage in Conformal Prediction.} \quad
Recent work has explored multiple approaches to improving or approximating conditional coverage in conformal prediction. One prominent line of work, initiated by~\citet{romano2019conformalized} with Conformalized Quantile Regression (CQR), directly models the conditional quantiles to produce adaptive intervals effective against heteroscedasticity. A second, broader strategy leverages an estimate of the full conditional distribution $P_{Y|X}$. This includes early work based on density estimates~\citep{lei2013distribution}, as well as more flexible modern approaches that reframe the problem via discretization of the output space~\citep{sesia2021conformal} or by using localized approaches that up-weight calibration points near the test point~\citep{guan2023localized}. Recognizing the theoretical limitations of localized methods, \citet{hore2025conformal} proposed RLCP, which achieves exact marginal coverage and, in a randomized sense, provides conditional coverage guarantees. Most recently, \citet{gibbs2025conformal} bridged the gap between marginal and conditional validity via a framework that interpolates between the two. Specifically, they employ augmented quantile regression on conformity scores to yield finite-sample coverage guarantees over a user-specified class of covariate shifts or subgroups. While existing methods offer theoretical improvements, their empirical validation is often limited to coarse, low-dimensional binning tests. Such approaches are inherently unscalable and may fail to detect subtle multivariate deviations. This highlights a crucial methodological gap: the lack of a scalable, quantitative, and objective framework for evaluating and comparing the conditional coverage properties of different methods. Our work is designed to directly address this issue.

\textbf{Model Selection in Conformal Prediction.} \quad
Model selection in conformal prediction has begun to receive systematic attention only recently. \citet{yang2025selection} explored the strategies of selection and aggregation and were the first to propose valid algorithmic procedures that provably maintain marginal coverage. \citet{liang2024conformal} built on \citet{yang2025selection} by introducing ModSel-CP and ModSel-CP-LOO methods that correct for model selection bias in conformal prediction without requiring additional data splitting, ensuring finite-sample validity through full conformal frameworks. While these methods primarily target efficiency (i.e., average set size), they do not explicitly address conditional validity. In contrast, our approach is specifically designed to fill this gap by directly targeting conditional validity. \looseness=-1

\section{Methodology}\label{sec-method}

In this section, we introduce our data-driven framework for learning and predicting the reliability of any given prediction interval method. Figure~\ref{fig:methodology_flowchart} illustrates the entire pipeline. 

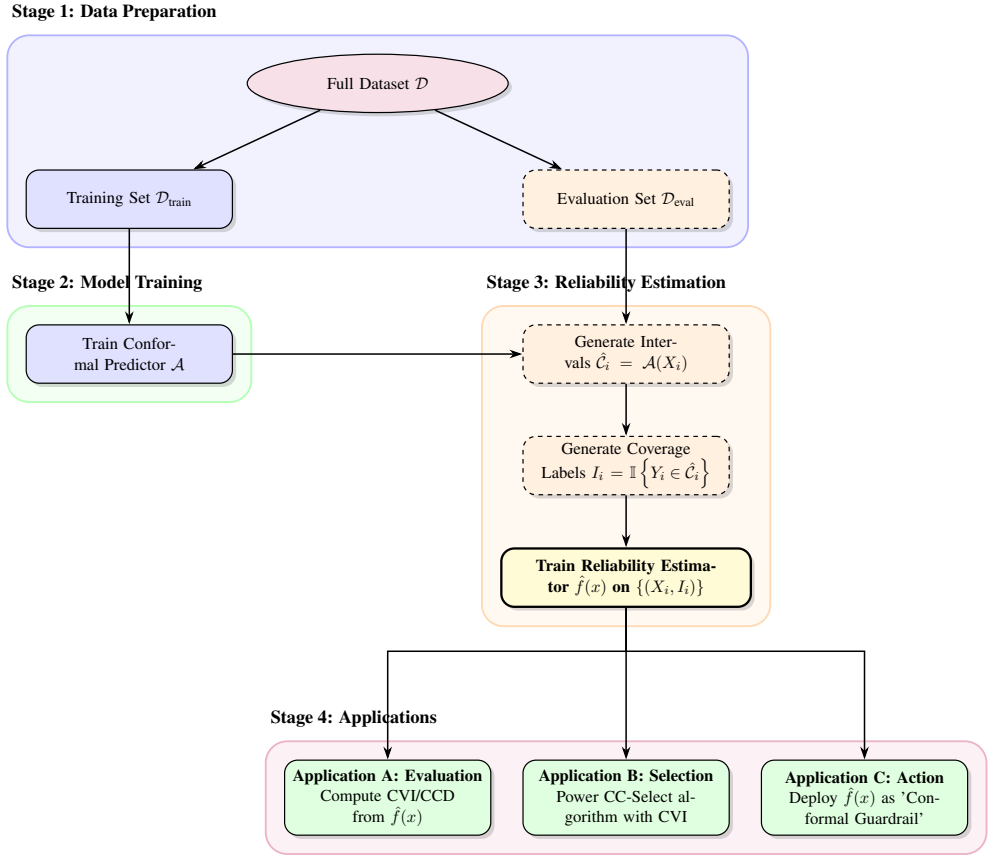
\begin{figure}[!b]
    \centering
    \resizebox{.8\textwidth}{!}{ 
        	\begin{tikzpicture}[
		node distance=1.2cm and 1.5cm,
		node-base/.style={
			rectangle, draw, thick, rounded corners=3mm, 
			minimum height=3.5em, drop shadow={opacity=0.3, yshift=-1mm}, 
			align=center, font=\normalsize
		},
		input/.style={
			node-base, ellipse, fill=purple!12, text width=4cm
		},
		process/.style={
			node-base, fill=blue!12, text width=4.5cm
		},
		data/.style={
			node-base, fill=orange!12, draw, dashed, text width=4.5cm
		},
		core-engine/.style={
			node-base, line width=1.5pt, fill=yellow!20, font=\bfseries, text width=5.5cm
		},
		application/.style={
			node-base, fill=green!12, text width=4.5cm, minimum height=4.5em
		},
		arrow/.style={
			-{Stealth[length=3mm, width=2mm]}, thick, line width=1pt
		},
		stage-label/.style={
			font=\large\bfseries, anchor=center, text=black
		}
		]
		
		\node[input] (D) {Full Dataset $\mathcal{D}$};
		\node[process, below left=1.5cm and 1.2cm of D] (D_train) {Training Set $\mathcal{D}_{\text{train}}$};
		\node[data, below right=1.5cm and 1.2cm of D] (D_eval) {Evaluation Set $\mathcal{D}_{\text{eval}}$};
		
		\node[process, below=2.2cm of D_train] (A) {Train Conformal Predictor $\mathcal{A}$};
		
		\node[data, below=2.2cm of D_eval] (C_hat) {Generate Intervals $\hat{\mathcal{C}}_i = \mathcal{A}(X_i)$};
		\node[data, below=1.2cm of C_hat] (I) {Generate Coverage Labels $I_i = \mathbb{I}\left\{Y_i \in \hat{\mathcal{C}}_i\right\}$};
		\node[core-engine, below=1.2cm of I] (f_hat) {Train Reliability Estimator $\hat{f}(x)$ on $\{(X_i, I_i)\}$};
		
		\node[application, below=3.5cm of f_hat, xshift=-5.5cm] (AppA) {\textbf{Application A: Evaluation}\\ Compute CVI/CCD from $\hat{f}(x)$};
		\node[application, below=3.5cm of f_hat] (AppB) {\textbf{Application B: Selection}\\ Power CC-Select algorithm with CVI};
		\node[application, below=3.5cm of f_hat, xshift=5.5cm] (AppC) {\textbf{Application C: Action}\\ Deploy $\hat{f}(x)$ as 'Conformal Guardrail'};
		
		\draw[arrow] (D) -- (D_train);
		\draw[arrow] (D) -- (D_eval);
		\draw[arrow] (D_train) -- (A);
		\draw[arrow] (D_eval) -- (C_hat);
		
		\draw[arrow] (A.east) -- ++(1.5,0) |- (C_hat.west);
		
		\draw[arrow] (C_hat) -- (I);
		\draw[arrow] (I) -- (f_hat);
		
		\draw[arrow] (f_hat.south) -- ++(0,-1) -| (AppA.north);
		\draw[arrow] (f_hat.south) -- ++(0,-1) -| (AppB.north);
		\draw[arrow] (f_hat.south) -- ++(0,-1) -| (AppC.north);
		
		\begin{scope}[on background layer]
			\node[draw=blue!30, fill=blue!5, very thick, rounded corners=5mm, 
			fit=(D) (D_train) (D_eval), inner sep=12pt] (bg1) {};
			
			\node[draw=green!30, fill=green!5, very thick, rounded corners=5mm, 
			fit=(A), inner sep=12pt] (bg2) {};
			
			\node[draw=orange!30, fill=orange!5, very thick, rounded corners=5mm, 
			fit=(C_hat) (I) (f_hat), inner sep=12pt] (bg3) {};
			
			\node[draw=purple!30, fill=purple!5, very thick, rounded corners=5mm, 
			fit=(AppA) (AppB) (AppC), inner sep=12pt] (bg4) {};
		\end{scope}
		
		\node[stage-label, above=5pt of bg1.north west, anchor=south west] {Stage 1: Data Preparation};
		\node[stage-label, above=5pt of bg2.north west, anchor=south west] {Stage 2: Model Training};
		\node[stage-label, above=5pt of bg3.north west, anchor=south west] {Stage 3: Reliability Estimation};
		\node[stage-label, above=5pt of bg4.north west, anchor=south west] {Stage 4: Applications};
		
	\end{tikzpicture} 
    }
    \caption{Overview of the predictive reliability learning framework. The pipeline consists of four stages: data preparation, conformal model training, reliability estimation via $\hat{\eta}(x)$, and downstream applications including evaluation, model selection, and deployment.}
    \label{fig:methodology_flowchart}
\end{figure}

\subsection{Asymptotic Behavior of the Conditional Coverage Probability}
\label{sec:asymptotics}

In our proposed CPA scheme, estimating the conditional coverage probability function plays a central role. We next study the asymptotic behavior of the conditional coverage probability function under sensible conditions.
The finite-sample conditional coverage probability function is defined as
\begin{align}
    \eta_n(x) = \mathbb{P}(Y_{n+1} \in \mathcal{C}_\alpha(X_{n+1};\mathcal{D}) \mid X_{n+1}=x). 
    \label{eq:eta_n}
\end{align}
Here the probability is taken over both the training data $\mathcal{D}$ and the future response $Y_{n+1}$. 
Consequently, $\eta_n(x)$ is a deterministic function of $x$, rather than a random quantity. Throughout the paper, $\{(X_i, Y_i)\}_{i=1}^{n+1}$ are assumed to be independent and identically distributed and to follow the regression model $Y = \mu(X) + \epsilon$, where $\mathbb{E}[\epsilon \mid X] = 0$ and the noise distribution may depend on $X$.
Within this subsection, we focus on the absolute residual as the non-conformity score, i.e., $S_i = |Y_i - \hat{\mu}_{n_\text{pred}}(X_i)|$. The corresponding prediction set is given by $\mathcal{C}_{\alpha}(x;\mathcal{D}) = [\hat{\mu}_{n_\text{pred}}(x) - \hat{q}_{n_\text{calib},1-\alpha},\hat{\mu}_{n_\text{pred}}(x) + \hat{q}_{n_\text{calib},1-\alpha}]$.

To characterize the convergence of $\eta_n(x)$, we need the following assumptions regarding the consistency of the regression estimator and the regularity of the error distribution. 

\begin{assumption}[Regularity Conditions]
\label{assumption:combined} 
\hfil
  \begin{enumerate}[label=(\alph*)]

      \item \label{assumption:sup_norm_convergence} 
      \textbf{(Stability of the Regression Estimator)}
      There exist a function $\tilde{\mu}$ and sequences 
      $\zeta_n = o(1)$ and $\rho_n = o(1)$ such that, as $n \to \infty$,
      \[
      \mathbb{P} \left( \| \hat{\mu}_{n} - \tilde{\mu} \|_{\infty} > \zeta_n \right) 
          \leq \rho_n.
      \]
      
      \item \label{assumption:conditional_density} 
      \textbf{(Bounded Conditional Density)}
      The conditional density function of the noise $\epsilon$ given $X=x$ is uniformly bounded by a constant $M > 0$.
      
      \item \label{assumption:quantile_density} 
      \textbf{(Bounded Quantile Density)}
      The density of $|\mu(X) + \epsilon - \tilde{\mu}(X)|$ is bounded below by a positive constant $M_0$ in a neighborhood of its $(1 - \alpha)$-quantile, denoted by $q_{1 - \alpha}$. 
  \end{enumerate}
\end{assumption}

Assumption~\ref{assumption:sup_norm_convergence} follows \citet{lei2018distribution} and serves as a baseline stability condition for $\hat{\mu}_n$. Typically, $\zeta_n$ takes the form $c (n/\log n)^{-\beta}$ for some $\beta > 0$, and $\rho_n$ is of the order $n^{-c}$ for some constant $c > 0$ (where the choice of $c$ only affects the constant pre-factor of $\zeta_n$). Note that we do not require $\tilde{\mu}$ to be the true regression function $\mu$. This accounts for practical scenarios where the regression model may be mis-specified; for instance, using a linear model to approximate a nonlinear regression function. In such cases, $\tilde{\mu}$ represents the best approximation of $\mu$ within the model class under the sup-norm.

\begin{theorem}[Uniform Convergence of $\eta_n(x)$]
    \label{thm::true_coverage}
    Under Assumption \ref{assumption:combined}, the finite-sample conditional coverage converges to the asymptotic coverage function $\eta(x) := \mathbb{P}(|\mu(X) + \epsilon - \tilde{\mu}(X)| \leq q_{1 - \alpha} \mid X = x)$ in the following sense: for sufficiently large $n_{\text{pred}}$,
    \[
    \| \eta_{n} - \eta \|_{\infty} \leq 18M\zeta_{n_\text{pred}} + 4\rho_{n_\text{pred}} +\exp(-2M_0^2 n_\text{calib} \zeta_{n_\text{pred}}^2).
    \]
\end{theorem}

Theorem~\ref{thm::true_coverage} ensures that as sample sizes increase, $\eta_n(x)$ stabilizes towards a fixed deterministic function $\eta(x)$. Crucially, however, this limit $\eta(x)$ is not necessarily equal to the target level $1-\alpha$ everywhere. In particular, model misspecification (where $\tilde{\mu} \neq \mu$) or unmodeled heteroscedasticity can cause $\eta(x)$ to deviate significantly from $1 - \alpha$, even 
with a large $n$. Since $\eta_n(x)$ is a stable, learnable property of the model-data pair, we estimate it to evaluate the reliability of the given conformal method.

The above result characterizes the population target underlying CPA, but it does not imply that this target $\eta_n$ is necessarily easy to estimate. In general, learning the conditional coverage function for an arbitrary conformal prediction method may be as challenging as learning the full conditional law of $Y\mid X$. Fortunately, for relatively stable conformal predictors, the conditional coverage probability function may be reasonably well behaved, as we show both theoretically and numerically in Appendix~\ref{app:eta_feasibility} of the supplement. \looseness=-1

\subsection{Estimating the Conditional Coverage Probability Function}
\label{sec:estimating_eta}

Consider a conformal prediction algorithm $\mathcal{A}$ that maps a pre-specified level $\alpha\in(0,1)$, a training dataset $\gD_0$ and a new input $x$ to a prediction set $\mathcal{C}_\alpha(x;\gD_0)$. We emphasize that $\mathcal{C}_\alpha(x;\mathcal{D}_0)$ is a generic mapping, and $\mathcal{D}_0$ will be allowed to vary in what follows. Our primary estimand of interest in this stage is the conditional coverage probability function $\eta_n(x)$ in \eqref{eq:eta_n}. Since $\eta_n(x)$ is unobservable, we adopt a sample-splitting strategy to estimate it. We partition the available full dataset $\mathcal{D}$ into two disjoint subsets: a training set $\mathcal{D}_{\text{train}}$ and an evaluation set $\mathcal{D}_{\text{eval}}$, with $n_{\text{train}}$ and $n_{\text{eval}}$ observations respectively. Let $\mathcal{I}_{\text{eval}}$ denote the set of indices corresponding to the data points in $\mathcal{D}_{\text{eval}}$.

\textbf{Step 1: Predictor Training.} The set $\mathcal{D}_{\text{train}}$ is utilized to construct the conformal predictor $\mathcal{A}$. Note that $\mathcal{D}_{\text{train}}$ may require internal splitting based on the choice of $\mathcal{A}$. Specifically, for split-conformal methods, it is partitioned into a proper training set $\mathcal{D}_{\text{pred}}$ (used to fit the underlying regression model $\hat{\mu}$) and a calibration set $\mathcal{D}_{\text{calib}}$ (used to compute the empirical quantile of the conformity scores). Once trained on $\mathcal{D}_{\text{train}}$, the predictor $\mathcal{A}$ is fixed.

\textbf{Step 2: Label Generation.} We apply the fixed predictor $\mathcal{A}$ to each data point $(X_i, Y_i)$ in the evaluation set $\mathcal{D}_{\text{eval}}$ to generate prediction sets $\hat{\mathcal{C}}_i = \mathcal{C}_\alpha(X_i; \mathcal{D}_{\text{train}})$. We then construct binary coverage indicators $I_i$, which serve as the target labels for reliability:
\begin{align}
    I_i = \mathbbm{1}\{ Y_i \in \hat{\mathcal{C}}_i \}, \quad \text{for all } i \in \mathcal{I}_{\text{eval}}.
\end{align}

\textbf{Step 3: Reliability Learning.} The problem of estimating $\eta_n(x)$ is thus reduced to a supervised binary classification problem. We construct a new dataset $\mathcal{D}_{\text{eval}}' = \{(X_i, I_i)\}_{i \in \mathcal{I}_{\text{eval}}}$ and train a probabilistic classifier $\hat{\eta}_{n_{\text{eval}}}$ to minimize a strictly proper scoring rule (e.g., the cross-entropy loss). However, raw outputs from classification models are often not well-calibrated probabilities. To ensure the predictions more faithfully reflect empirical frequencies, we apply post-hoc probability calibration---specifically Isotonic Regression---to the classifier's output. The resulting estimator outputs the predicted probability of coverage:
\begin{align}
    \hat{\eta}_{n_{\text{eval}}}(x) \approx \mathbb{P}(Y_{n+1}\in \mathcal{C}_{\alpha}(x;\gD_{\textup{train}}) \mid X_{n+1}=x).
\end{align}

Any probabilistic classifier capable of capturing nonlinear dependencies, such as random forests, gradient boosted trees (XGBoost), or calibrated neural networks, can be employed as the base learner for $\hat{\eta}_{n_{\text{eval}}}$. The detailed training procedure is presented in Algorithm~\ref{alg:cpa_training} in Appendix~\ref{app: Detailed Algorithms}, and the theoretical background on calibration is provided in Appendix~\ref{app:calibration_methodology}.

The implementation of Algorithm \ref{alg:cpa_training} involves three important practical choices: the split ratio, the number of splits, and the base learner. Regarding the \textbf{data splitting ratio} $\rho \in (0, 1)$, defined such that $n_{\text{train}} = \lfloor \rho n \rfloor$ and $n_{\text{eval}} = n - n_{\text{train}}$,  a trade-off exists between the predictive power of the conformal predictor and the sample size available for reliability learning. Based on extensive simulations presented later, we adopt an even split of $\rho = 0.5$ as it consistently yields stable performance. Second, to mitigate the variance arising from random data partitioning, we employ a \textbf{multi-split strategy} by repeating the procedure over $K$ random splits and averaging the resulting estimates. We set $K=5$ in our experiments, which effectively stabilizes the reliability estimate via aggregation. Finally, regarding the \textbf{base learner} $\mathcal{L}$, we select the best model from a candidate set via 5-fold cross-validation in our experiments. Importantly, our sensitivity analysis confirms that the framework is robust to this specific choice, provided the learner possesses sufficient capacity to capture nonlinear failure patterns and is coupled with the aforementioned isotonic calibration.

\subsection{Assessment of Conditional Coverage}
\label{sec:assessment_metrics}

With the trained reliability estimator $\hat{\eta}_{n_{\textup{eval}}}(x)$, we can quantitatively assess the conditional coverage properties of the conformal predictor. While standard marginal coverage merely validates compliance with a global average constraint, our framework captures the underlying \textbf{coverage heterogeneity}, enabling a \textbf{granular} characterization of \textbf{local reliability}.

\subsubsection{The Conditional Validity Index (CVI) and its Decomposition}

We define the primary scalar metric, the \textbf{Conditional Validity Index (CVI)}, as the mean absolute deviation of the estimated reliability from the nominal level $1-\alpha$:
\begin{align}
    \text{CVI}_{n_\text{eval}} = \frac{1}{n_{\text{eval}}} \sum_{i \in \mathcal{I}_{\text{eval}}} \left| \hat{\eta}_{n_{\text{eval}}}(X_i) - (1 - \alpha) \right|.
\end{align}
A lower CVI indicates that the predictor's local coverage probabilities are concentrated around the target level. We decompose this total deviation into \textbf{Undercoverage Risk ($\text{CVI}_{\text{U}}$)} and \textbf{Overcoverage Cost ($\text{CVI}_{\text{O}}$)}, so that $\text{CVI} = \text{CVI}_{\text{U}} + \text{CVI}_{\text{O}}$. Fix a tolerance parameter $\gamma\in[0,1]$.

\textbf{1. Safety Assessment (Undercoverage).} \quad
These metrics quantify failure to meet the target coverage.
\begin{itemize}
    \item \textbf{Undercoverage Risk ($\text{CVI}_{\text{U}}$):} The average reliability shortfall.
    \begin{equation}
        \text{CVI}_{\text{U}} = \frac{1}{n_{\text{eval}}} \sum_{i \in \mathcal{I}_{\text{eval}}} \max\left(0, (1 - \alpha) - \hat{\eta}_{n_{\text{eval}}}(X_i)\right).
    \end{equation}
    
    \item \textbf{Undercoverage Rate ($\pi_-$):} The proportion of samples whose estimated coverage falls below the target.
    \begin{equation}
        \pi_- = \frac{1}{n_{\text{eval}}} \sum_{i \in \mathcal{I}_{\text{eval}}} \mathbbm{1}\{\hat{\eta}_{n_{\text{eval}}}(X_i) < (1-\gamma)(1-\alpha)\}. 
    \end{equation}
    
    \item \textbf{Conditional Mean Undercoverage (CMU):} The average shortfall among undercovered samples.
    \begin{equation}
    \label{eq:cmu}
    \text{CMU}
    =
    \frac{
    \sum_{i \in \mathcal{I}_{\text{eval}}}
    \mathbbm{1}\!\left\{
    \hat{\eta}_{n_{\text{eval}}}(X_i) < (1-\gamma)(1-\alpha)
    \right\}
    \max\!\left(
    0, (1-\alpha) - \hat{\eta}_{n_{\text{eval}}}(X_i)
    \right)
    }{
    \sum_{i \in \mathcal{I}_{\text{eval}}}
    \mathbbm{1}\!\left\{
    \hat{\eta}_{n_{\text{eval}}}(X_i) < (1-\gamma)(1-\alpha)
    \right\}
    }.
    \end{equation}
    If $\pi_- = 0$, then $\text{CMU}$ is defined as 0.
\end{itemize}

\textbf{2. Efficiency Assessment (Overcoverage).} \quad
These metrics quantify unnecessary conservatism.
\begin{itemize}
    \item \textbf{Overcoverage Cost ($\text{CVI}_{\text{O}}$):} The average excess coverage beyond the target.
    \begin{equation}
        \text{CVI}_{\text{O}} = \frac{1}{n_{\text{eval}}} \sum_{i \in \mathcal{I}_{\text{eval}}} \max\left(0, \hat{\eta}_{n_{\text{eval}}}(X_i) - (1 - \alpha)\right).
    \end{equation}
    
    \item \textbf{Overcoverage Rate ($\pi_+$):} The proportion of samples whose estimated coverage exceeds the target.
    \begin{equation}
        \pi_+ = \frac{1}{n_{\text{eval}}} \sum_{i \in \mathcal{I}_{\text{eval}}} \mathbbm{1}\{\hat{\eta}_{n_{\text{eval}}}(X_i) >(1+\gamma) (1-\alpha)\}.
    \end{equation}

    \item \textbf{Conditional Mean Overcoverage (CMO):} The average excess coverage among overcovered samples.
   \begin{equation}
    \label{eq:cmo}
    \text{CMO}
    =
    \frac{
    \sum_{i \in \mathcal{I}_{\text{eval}}}
    \mathbbm{1}\!\left\{
    \hat{\eta}_{n_{\text{eval}}}(X_i) > (1+\gamma)(1-\alpha)
    \right\}
    \max\!\left(
    0, \hat{\eta}_{n_{\text{eval}}}(X_i) - (1-\alpha)
    \right)
    }{
    \sum_{i \in \mathcal{I}_{\text{eval}}}
    \mathbbm{1}\!\left\{
    \hat{\eta}_{n_{\text{eval}}}(X_i) > (1+\gamma)(1-\alpha)
    \right\}
    }.
    \end{equation}
    If $\pi_+ = 0$, then $\text{CMO}$ is defined as 0.
\end{itemize}

The tolerance parameter $\gamma$ is taken to be a small constant close to zero. It prevents $\pi_-$ and $\pi_+$ from being dominated by negligible numerical fluctuations around $1-\alpha$, so that these rates reflect substantively meaningful deviations from the target coverage.

\subsubsection{Visualizing Reliability: The CVP Curve}

Complementing the scalar metrics, we introduce the \textbf{Conditional Validity Profile (CVP)} curve, which visualizes the distribution of estimated conditional coverage probabilities.

\textbf{Construction.} Let $\hat{\eta}_{(1)} \le \dots \le \hat{\eta}_{(n_{\text{eval}})}$ denote the sorted reliability estimates for the evaluation set. The CVP curve plots the empirical quantile function $Q(p) = \hat{\eta}_{(\lceil p \cdot n_{\text{eval}} \rceil)}$ against the cumulative data proportion $p \in (0, 1]$.

\textbf{Geometric Interpretation.} As illustrated in Figure~\ref{fig:cvp_example}, the geometry of the CVP curve directly corresponds to our numerical indices. With $\gamma=0$, 
\begin{enumerate}
    \item \textbf{Intersection ($\pi_-$):} the proportion $p$ at which the curve intersects the target line $1-\alpha$ from below equals the Undercoverage Rate.
    \item \textbf{Area Below Target ($\text{CVI}_{\text{U}}$):} the area bounded between the curve and the target line (where $Q(p) < 1-\alpha$) represents the aggregate Safety Deficit, which equals $\text{CVI}_{\text{U}}$.
    \item \textbf{Area Above Target ($\text{CVI}_{\text{O}}$):} the area bounded between the curve and the target line (where $Q(p) > 1-\alpha$) represents the Efficiency Loss, which equals $\text{CVI}_{\text{O}}$.
\end{enumerate}

Ideally, a perfectly conditionally valid model yields a flat CVP curve aligned with the horizontal line $y = 1-\alpha$. Deviations from this reference line diagnose reliability failures; for instance, a sharp drop at the left tail indicates a subpopulation with severe undercoverage.

\begin{figure}[t]
    \centering
    \includegraphics[width=0.8\linewidth]{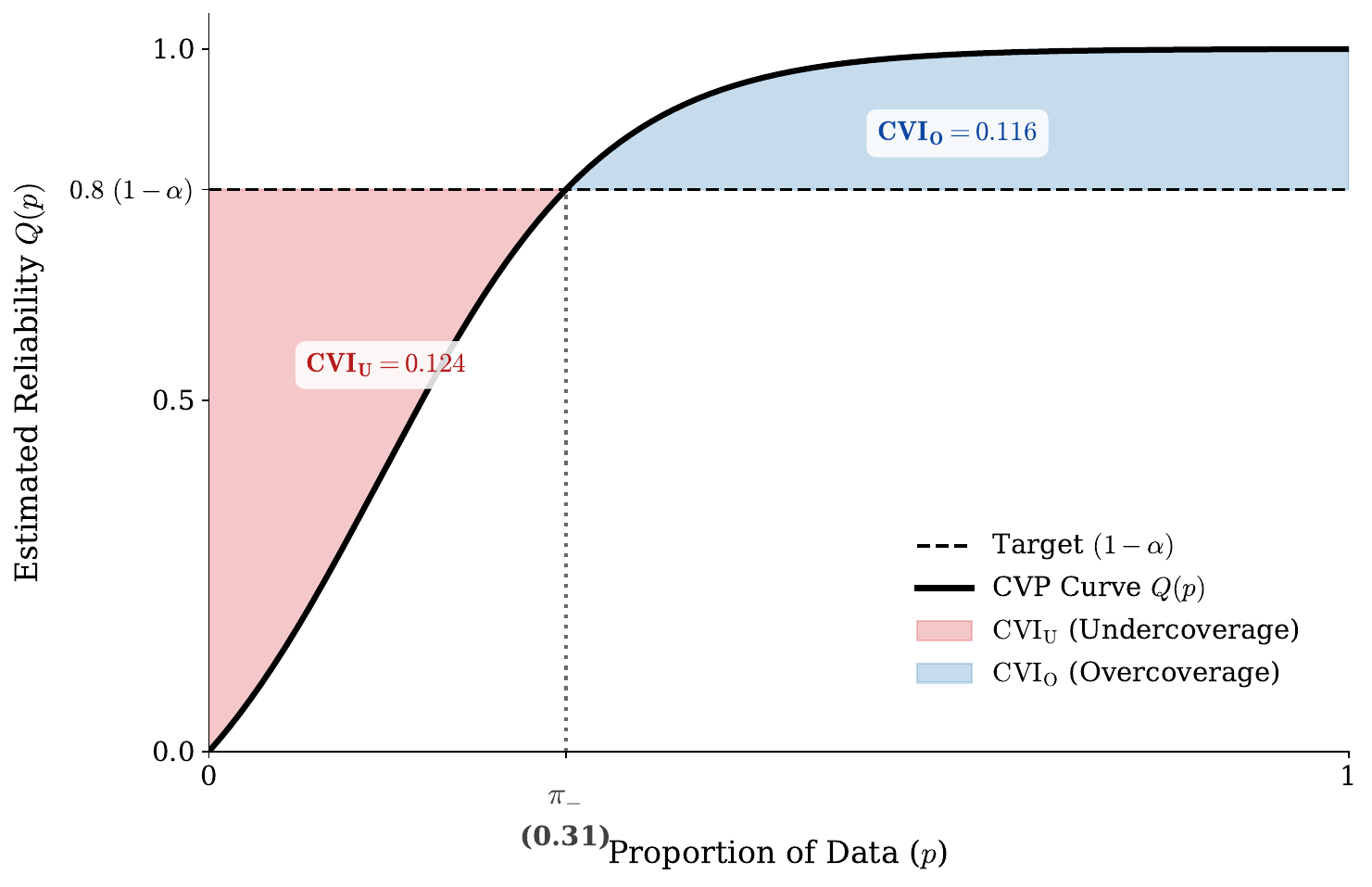}
    \caption{\textbf{Conditional Validity Profile (CVP) Curve.}
    This curve depicts conditional coverage for a standard conformal predictor. 
    The red area reflects undercoverage risk ($\text{CVI}_{\text{U}}$), the blue area represents overcoverage cost ($\text{CVI}_{\text{O}}$), and their intersection marks the fraction of undercovered samples ($\pi_-$). }
    \label{fig:cvp_example}
\end{figure}

\subsection{Reliability-Driven Model Selection and Deployment}
\label{sec:model_selection}

A primary utility of the CVI metric is to facilitate objective model selection among a set of candidate conformal prediction procedures $\{ \mathcal{A}_m \}_{m=1}^M$. While traditional selection criteria focus on minimizing the prediction set size, our framework enables selection based on the fidelity of conditional coverage. We introduce a robust selection protocol, \textbf{CC-Select}, designed to identify the model with the lowest expected conditional miscalibration.

To mitigate the high variance associated with single-shot data splitting, CC-Select employs a repeated subsampling procedure. In each iteration, the dataset is randomly partitioned, and candidate models are ranked by their CVI scores. The procedure identifies the optimal model $m^*$ that minimizes the average CVI across all $K$ splits.

For the selected procedure $\mathcal{A}_{m^*}$, the reliability estimators $\{\hat{\eta}_{m^*}^{(k)}\}_{k=1}^K$ obtained from the repeated subsampling stage are retained. After model selection, we retrain $\mathcal{A}_{m^*}$ on the full dataset $\mathcal{D}$. For a new test input $X_{\mathrm{new}}$, we apply each stored reliability estimator to the selected procedure at $X_{\mathrm{new}}$ and compute their average,
\[
\hat{c}(X_{\mathrm{new}}) = \frac{1}{K}\sum_{k=1}^{K}
\hat{\eta}_{m^*}^{(k)}(X_{\mathrm{new}}),
\]
which serves as an estimate of the conditional coverage at $X_{\mathrm{new}}$. In applications, $\hat{c}(X_{\mathrm{new}})$ can be compared to the nominal level $1-\alpha$ to identify inputs for which the conformal prediction set may exhibit poor conditional coverage. A complete description of the two-stage procedure is provided in Algorithm~\ref{alg:cc_select} in Appendix~\ref{app: Detailed Algorithms}.

\section{Theoretical Analysis}
\label{sec:theory}

In this section, we establish the theoretical foundations of the proposed Conformal Prediction Assessment (CPA) framework. First, we quantify convergence rates of the reliability estimator $\hat{\eta}_{n_{\textup{eval}}}$, explicitly characterizing the trade-off between estimation and approximation errors introduced by sample splitting. Building on this pointwise convergence, we demonstrate the consistency of the Conditional Validity Index (CVI) as a reliable proxy for true conditional miscalibration. Finally, we provide rigorous guarantees for the model selection procedure, establishing its consistency in identifying the optimal conformal predictor.

\subsection{Convergence of Conditional Coverage Estimates}

We now analyze the convergence properties of our estimator $\hat{\eta}_{n_{\textup{eval}}}(x)$, which is fitted on the evaluation data $\{ (X_i, I_i) \}_{i \in \mathcal{I}_{\textup{eval}}}$. Conditional on the feature $X_i$, the coverage indicators $I_i$ are independent Bernoulli random variables:
\begin{align}
    I_i \mid X_i \sim \text{Bernoulli}(\eta_{n_{\textup{train}}}(X_i)),
\end{align}
where $\eta_{n_{\textup{train}}}(x) := \mathbb{P}(Y \in \mathcal{C}(X; \mathcal{D}_{\textup{train}}) \mid X = x)$ denotes the conditional coverage function associated with a conformal prediction model trained with sample size $n_{\textup{train}}$.

Our objective is to bound the deviation $\| \hat{\eta}_{n_{\textup{eval}}} - \eta_{n} \|_{\infty}$, which quantifies the uniform error of our reliability estimator relative to the true conditional coverage function $\eta_n(x)$ for the full dataset $\mathcal{D}$. By the triangle inequality, we bound this error by the sum of two components:
\begin{align}
    \| \hat{\eta}_{n_{\textup{eval}}} - \eta_{n} \|_{\infty} \leq \underbrace{\| \hat{\eta}_{n_{\textup{eval}}} - \eta_{n_{\textup{train}}} \|_{\infty}}_{\text{Estimation Error}} + \underbrace{\| \eta_{n_{\textup{train}}} - \eta_{n} \|_{\infty}}_{\text{Approximation Error}}.
\end{align}
The \textit{Estimation Error} arises from estimating the function $\eta_{n_{\textup{train}}}$ using a finite evaluation set $\mathcal{D}_{\textup{eval}}$. The \textit{Approximation Error} accounts for the discrepancy between the coverage of the model trained on the subsample $\mathcal{D}_{\textup{train}}$ and that of the model trained on the full dataset $\mathcal{D}$. To bound these terms, we introduce the following assumption regarding the classifier's convergence rate.

\begin{assumption}[Classifier Convergence Rate]
    \label{assumption:classifier_convergence}
    The classifier used to estimate the conditional coverage satisfies 
    \[ \| \hat{\eta}_{n_{\textup{eval}}} - \eta_{n_{\textup{train}}} \|_{\infty} = O_p(\psi_{n_{\textup{eval}}}), \]
    for some rate $\psi_{n_{\textup{eval}}} = o(1)$ as $n_{\textup{eval}} \to \infty$.
\end{assumption}

This assumption is standard in nonparametric classification. For instance, kernel regression estimators typically achieve a rate of $\psi_n \asymp (n/\log n)^{-\gamma}$ for some $\gamma > 0$, under suitable regularity conditions \citep{stone1982optimal}. 

Next, to bound the \textit{Approximation Error} $\| \eta_{n_{\textup{train}}} - \eta_{n} \|_{\infty}$, we require the finite-sample conditional coverage function to stabilize towards a fixed deterministic limit as the sample size increases. This stability ensures that the coverage properties learned from the subsample $\mathcal{D}_{\textup{train}}$ remain representative of the full model's behavior. We formalize this requirement as follows: \looseness=-1

\begin{assumption}[Stability of the Conditional Coverage Function]
    \label{assumption:stability_of_conditional_coverage}
    There exists a deterministic limiting function $\eta: \mathcal{X} \to [0, 1]$ such that as $n \to \infty$, the finite-sample conditional coverage function converges:
    \[ \| \eta_n - \eta \|_{\infty} = O(\varphi_n), \]
    for some rate sequence $\varphi_n = o(1)$.
\end{assumption}

The validity of this assumption is grounded in our earlier asymptotic analysis. Specifically, Theorem~\ref{thm::true_coverage} establishes that for split-conformal prediction with absolute residual scores, this stability condition holds with a rate of $\varphi_n = O(\zeta_{n_\textup{pred}} + \rho_{n_\text{pred}} + \exp(-c n_\textup{calib} \zeta_{n_\textup{pred}}^2))$. By formulating this property as a general assumption, we decouple the consistency analysis of our reliability estimator from the specific mechanics of the base CP algorithm. Consequently, our framework remains applicable to any uncertainty quantification method that exhibits asymptotic stability.

\begin{theorem}[Convergence of $\hat{\eta}_{n_{\textup{eval}}}(x)$]
    \label{thm::estimated_coverage}
    Under Assumptions \ref{assumption:classifier_convergence} and \ref{assumption:stability_of_conditional_coverage}, we have
    \[
    \| \hat{\eta}_{n_{\textup{eval}}} - \eta_{n} \|_{\infty} = O_p(\psi_{n_{\textup{eval}}} + \varphi_{n_\textup{train}} + \varphi_{n}).
    \]
\end{theorem}

Since $n_{\textup{train}} < n$, the rate $\varphi_{n_{\textup{train}}}$ typically dominates $\varphi_{n}$, rendering the overall rate effectively $O_p(\psi_{n_{\textup{eval}}} + \varphi_{n_{\textup{train}}})$. This result highlights a fundamental trade-off governed by the data split: increasing the evaluation set size $n_{\textup{eval}}$ reduces the estimation error rate $\psi$, but simultaneously shrinks the training set $n_{\textup{train}}$, thereby increasing the approximation error rate $\varphi$. This theoretical insight motivates the use of a balanced split ratio to effectively manage these competing sources of error.

\subsection{Consistency of the CVI Metric}

Having established pointwise convergence of the reliability estimator, we now turn our analysis to the Conditional Validity Index (CVI). As a global summary statistic, the CVI aggregates local coverage deviations to quantify overall conditional miscalibration of the predictor. \looseness=-1

To rigorously analyze consistency, we define the population-level counterpart to our empirical metric. Let the \textbf{Oracle CVI} be the expected absolute deviation of the \textit{true} conditional coverage probability $\eta_n(X)$ from the target level $1-\alpha$, over the feature distribution $P_X$:
\begin{align}
    \textup{CVI}_{\textup{oracle},n} = \mathbb{E}_{X \sim P_X} \left[ \left| \eta_n(X) - (1 - \alpha) \right| \right].
\end{align}
This quantity represents the ideal, unobservable measure of conditional miscalibration for the trained model.

The following theorem establishes that the empirical CVI, computed on the evaluation set, is a consistent estimator of this Oracle CVI. This confirms that the metric calculated from data converges in probability to the true underlying performance characteristic.

\begin{theorem}[Consistency of CVI]
\label{thm::consistency}
Under Assumptions \ref{assumption:classifier_convergence} and \ref{assumption:stability_of_conditional_coverage}, as $n_{\textup{train}}, n_{\textup{eval}} \to \infty$, the empirical CVI is consistent for the Oracle CVI:
\[
\left| \textup{CVI}_{n_{\textup{eval}}} - \textup{CVI}_{\textup{oracle},n} \right| \xrightarrow{p} 0.
\]
\end{theorem}

This consistency result provides the theoretical foundation for using CVI as a benchmark metric in practice. It ensures that with sufficient data, the empirical assessment faithfully reflects the true conditional reliability, thereby enabling valid comparisons and rankings of different conformal prediction models.
\subsection{Consistency of Model Selection via CVI}

Finally, we investigate the consistency of model selection based on CVI. That is, we analyze whether CVI can consistently identify the algorithm with the superior conditional coverage performance. To formalize this, we consider two competing conformal prediction algorithms, indexed by $k \in \{1, 2\}$. Their respective empirical and Oracle CVI metrics are defined as:
\begin{align*}
    \text{CVI}_{n_\textup{eval}}^{(k)} &= \frac{1}{n_\textup{eval}} \sum_{i \in \mathcal{I}_{\text{eval}}} \left|\hat{\eta}_{n_\textup{eval}}^{(k)}(X_i) - (1 - \alpha)\right|, \\
    \text{CVI}_{\textup{oracle},n}^{(k)} &= \mathbb{E}_{X \sim P_X}|\eta_{n}^{(k)}(X) - (1 - \alpha)|.
\end{align*}
Under Assumption \ref{assumption:stability_of_conditional_coverage}, we have $\| \eta_{n}^{(k)} - \eta^{(k)} \|_{\infty} = o_p(1)$, where $\eta^{(k)}$ is the deterministic asymptotic coverage function for $k=1,2$, respectively. This implies that the oracle CVI also converges to a deterministic limit:
$$
\text{CVI}_{\textup{oracle},n}^{(k)} \xrightarrow{p} \mathbb{E}_{X \sim P_X}|\eta^{(k)}(X) - (1 - \alpha)| =: \text{CVI}_{\textup{oracle}}^{(k)}.
$$
A challenging scenario arises when both algorithms achieve perfect asymptotic coverage, i.e., $\textup{CVI}_{\textup{oracle}}^{(1)} = \textup{CVI}_{\textup{oracle}}^{(2)} = 0$. In this case, their finite-sample performance, dictated by their convergence rates, becomes the deciding factor. We introduce the following definition to compare their rates.

\begin{definition*}[Asymptotically Better]
    For two algorithms with oracle metrics $\textup{CVI}_{\textup{oracle},n}^{(1)}$ and $\textup{CVI}_{\textup{oracle},n}^{(2)}$, we say that algorithm 1 is asymptotically better than algorithm 2 if for any $\epsilon > 0$, there exists a constant $N > 0$ and a constant $c_{\epsilon} > 0$ such that for all $n > N$,
    $$
    \mathbb{P}\left(\textup{CVI}_{\textup{oracle},n}^{(2)} \geq (1 + c_\epsilon) \textup{CVI}_{\textup{oracle},n}^{(1)}\right) \geq 1 - \epsilon.
    $$
\end{definition*}

We now present the main theorem on the consistency of model selection.

\begin{theorem}[Consistency of Model Selection via CVI]
\label{thm::model_selection}
Under Assumptions \ref{assumption:classifier_convergence} and \ref{assumption:stability_of_conditional_coverage}, let the selected algorithm be $\hat{k} = \arg \min_{k \in \{1,2\}} \textup{CVI}_{n_\textup{eval}}^{(k)}$. The model selection is consistent in the following two cases:

\textbf{Case 1 (Different Asymptotic Performance).} \quad If $ \textup{CVI}_{\textup{oracle}}^{(1)} < \textup{CVI}_{\textup{oracle}}^{(2)}$, then $\mathbb{P}(\hat{k} = 1) \to 1$ as $n_{\textup{train}}, n_{\textup{eval}} \to \infty$.

\textbf{Case 2 (Identical Asymptotic Performance).} \quad If $ \textup{CVI}_{\textup{oracle}}^{(1)} = \textup{CVI}_{\textup{oracle}}^{(2)} = 0$ and algorithm 1 is asymptotically better than algorithm 2, then under the additional condition that
$$
\frac{n_{\textup{eval}}^{-1/2} \lor \max_{k} \| \hat{\eta}_{n_{\textup{eval}}}^{(k)} - \eta_{n_\textup{train}}^{(k)} \|_{\infty}}{\textup{CVI}_{\textup{oracle},n_{\textup{train}}}^{(2)}} = o_p(1),
$$
we have $\mathbb{P}(\hat{k} = 1) \to 1$ as $n_{\textup{train}}, n_{\textup{eval}} \to \infty$.
\end{theorem}

The additional condition in Case 2 is standard in the model selection literature \citep{yang2007consistency}. It ensures that the estimation error of the empirical CVI is asymptotically negligible compared to the Oracle CVI of the worse model. The numerator represents the estimation error rate. Under Assumption \ref{assumption:classifier_convergence}, $\| \hat{\eta}_{n_{\textup{eval}}}^{(k)} - \eta_{n_\textup{train}}^{(k)} \|_{\infty} = O_p(\psi_{n_{\textup{eval}}})$. The denominator, $\text{CVI}_{\textup{oracle},n_{\textup{train}}}^{(2)}$, reflects the rate of convergence of algorithm 2, as $\text{CVI}_{\textup{oracle},n_{\textup{train}}}^{(2)} = \mathbb{E}_{X \sim P_X}|\eta_{n_{\textup{train}}}^{(2)}(X) - \eta^{(2)}(X)| = O_p(\varphi_{n_{\textup{train}}})$. The condition thus requires that the estimation error vanishes faster than the true conditional coverage rate of the worse model.

\section{Simulation Studies}
\label{sec:simulation}
We conduct a comprehensive simulation study to evaluate CPA across diverse distributional regimes. The study has three objectives: (1) verify that the estimated conditional validity recovers ground-truth performance; (2) assess robustness to the choice of reliability estimator; and (3) extract practical guidance for deployment. We first study ranking fidelity, distributional recovery, and robustness across four synthetic scenarios (Sections \ref{subsec:oracle_recovery}--\ref{subsec:robustness}). We then investigate operational requirements---data allocation and sample complexity---to derive actionable recommendations, supported by Appendices \ref{app:sensitivity_rho}--\ref{app:sample_complexity}.

\subsection{Experimental Design and Implementation}
\label{subsec:exp_design}
\subsubsection{Data Generating Processes and Evaluation Protocol}

Our evaluation adopts a \textbf{``Selection--Deployment''} protocol that mirrors the lifecycle of real-world machine learning systems. In each replication, we generate a selection dataset $\mathcal{D}_{\text{select}}$ ($n=2000$) and an independent held-out test set $\mathcal{D}_{\text{test}}$ ($n_{\text{test}}=2000$), which remains inaccessible during model selection. We randomly split $\mathcal{D}_{\text{select}}$ into a training set $\mathcal{D}_{\text{train}}$ ($50\%$) for fitting candidate models and an evaluation set $\mathcal{D}_{\text{eval}}$ ($50\%$) for CPA-based reliability auditing. \looseness = -1

To establish the ground truth for model selection, we use a \textbf{full-data refitting protocol}. CPA performs selection using only the split data ($\mathcal{D}_{\text{train}}, \mathcal{D}_{\text{eval}}$), whereas the \textbf{Oracle ranking} is obtained by retraining \textit{all} candidate algorithms on the full $\mathcal{D}_{\text{select}}$ and evaluating their exact conditional coverage on $\mathcal{D}_{\text{test}}$. This comparison directly tests whether rankings inferred from partial data transfer to the fully deployed models.

Across all experimental settings, the data-generating processes follow the common location-scale form $Y = \mu(X) + \sigma(X)\epsilon$, where $\epsilon$ is a standardized random variable independent of the covariates $X$. Building on the benchmark constructions of \citet{lei2018distribution}, we design four synthetic scenarios that target distinct conditional coverage challenges, from nonlinear misspecification to strong feature dependence. Table \ref{tab:dgp_summary} summarizes these settings, and Appendix \ref{app:dgp_details} provides their full mathematical specifications.

\begin{table}[H]
    \centering
    \renewcommand{\arraystretch}{}
    \caption{Summary of Data Generating Processes used in our experiments.}
    \label{tab:dgp_summary}
    \resizebox{\textwidth}{!}{
    \begin{tabular}{l l l l}
        \toprule
        \textbf{Setting} & \textbf{Mean Function} $\mu(X)$ & \textbf{Noise Distribution} & \textbf{Challenge} \\
        \midrule
        A: Linear             & Sparse Linear                   & $\mathcal{N}(0, 1)$              & Baseline Validity \\
        B: Nonlinear         & Nonlinear Interactions         & $t_2$ (Heavy-tailed)             & Model Misspecification \\
        C: Heteroscedastic    & Linear                          & $\mathcal{N}(0,\sigma(X)^2)$     & Conditional Coverage \\
        D: Correlated Covariates & Linear (Dependent $X$)          & $\sigma(X) \cdot t_2$            & Feature Dependence \& Heavy-tails \\
        \bottomrule
    \end{tabular}
    }
\end{table}

\subsubsection{Benchmarked Algorithms}

To evaluate the discriminative power of CPA, we benchmark nine prediction-interval algorithms at coverage level $1-\alpha = 0.9$, chosen to span a broad spectrum of adaptivity. We include three classical baselines---\textbf{OLS}, \textbf{Adaptive Residual Bootstrap}, and \textbf{Quantile Regression Forests (QRF)}---and six conformal prediction (CP) methods. The CP methods range from procedures that primarily guarantee marginal validity (\textbf{CP-Residual}, \textbf{CV+}) to approaches designed to better adapt to conditional coverage variation (\textbf{CP-Studentized}, \textbf{CQR}, \textbf{LCP}, \textbf{RLCP}).
Implementation details and hyperparameter configurations are provided in Appendix \ref{app:benchmark_details}.

\subsubsection{Oracle Ground Truth and Ranking}

To benchmark the estimation fidelity of CPA, we compute the \textit{exact} conditional coverage probability for each deployed predictor. Unlike finite-sample empirical evaluation, the simulation environment gives us access to the true parameters of the data-generating process (DGP). In particular, leveraging the location-scale structure $Y = \mu(X) + \sigma(X)\epsilon$, we use the true conditional mean $\mu(x)$ and noise scale $\sigma(x)$.

Let $\hat{C}_m^{\text{full}}(x) = [\hat{l}_m(x), \hat{u}_m(x)]$ denote the prediction interval constructed by algorithm $\mathcal{A}_m$ after retraining on the full selection set, as described in Section \ref{subsec:exp_design}. We define the \textbf{Oracle Conditional Coverage} for algorithm $m$, denoted by $\eta_m^{\text{oracle}}(x)$, as the analytical probability that the target falls within this interval under $F_\epsilon$, the cumulative distribution function (CDF) of the standardized noise $\epsilon$:
\begin{equation}
    \eta_m^{\text{oracle}}(x) \coloneqq \mathbb{P}(Y \in \hat{C}_m^{\text{full}}(x) \mid X=x) = F_\epsilon\left(\frac{\hat{u}_m(x) - \mu(x)}{\sigma(x)}\right) - F_\epsilon\left(\frac{\hat{l}_m(x) - \mu(x)}{\sigma(x)}\right).
\end{equation}
This formulation removes finite-sample auditing variance and therefore yields a deterministic measure of reliability. Based on this ground-truth quantity, we summarize the performance of algorithm $m$ using the \textbf{Oracle Conditional Validity Index ($\text{CVI}_m^{\text{oracle}}$)}, defined as the mean absolute deviation on the independent test set $\mathcal{D}_{\text{test}}$:
\begin{equation}
    \text{CVI}_m^{\text{oracle}} = \frac{1}{n_{\text{test}}} \sum_{i \in \mathcal{D}_{\text{test}}} \left| \eta_m^{\text{oracle}}(X_i) - (1 - \alpha) \right|.
\end{equation}
The \textbf{Oracle Ranking} $\pi^{\text{oracle}}$ is then established by sorting the algorithms in ascending order of their respective $\text{CVI}_m^{\text{oracle}}$ values. This sequence serves as the gold standard for evaluating CPA's selection accuracy.

\subsection{Predictive Fidelity of Model Selection}
\label{subsec:oracle_recovery}

The utility of CPA hinges on whether it narrows the selection--deployment gap. To quantify this, we measure the \textit{rank concordance} between the CPA-estimated ranking $\hat{\pi}$, based on estimated CVI, and the ground-truth \textbf{Oracle ranking} $\pi^{\text{oracle}}$. We report three complementary metrics: \textbf{Weighted Kendall's $\tau_w$} for global correlation, \textbf{Hit@k} for top-$k$ precision, and \textbf{NDCG} for selection utility.

\subsubsection{Ranking Consistency Analysis}

\begin{table}[!b]
    \centering
    \renewcommand{\arraystretch}{}
    \caption{Rank concordance between CPA-estimated and Oracle rankings. Results are averaged over independent experimental runs, with standard deviations given in parentheses.}
    \label{tab:ranking_consistency}
    \setlength{\tabcolsep}{4pt} 
    \begin{tabular}{l c c c c c}
    \toprule
    \textbf{Setting} & \textbf{Kendall's $\tau_w$} & \textbf{Spearman's $\rho$} & \textbf{NDCG@1} & \textbf{NDCG@3} & \textbf{Hit@3} \\
    \midrule
    A (Linear)          & 0.902 (0.014) & 0.771 (0.021) & 0.809 (0.028) & 0.836 (0.019) & 0.580 (0.035) \\
    B (Heavy-Tailed)    & 0.784 (0.024) & 0.718 (0.029) & 0.953 (0.015) & 0.965 (0.005) & 0.740 (0.026) \\
    C (Heteroscedastic) & 0.902 (0.009) & 0.800 (0.019) & 0.852 (0.018) & 0.964 (0.004) & 0.920 (0.020) \\
    D (Correlated Covariates)     & 0.797 (0.025) & 0.696 (0.034) & 0.955 (0.015) & 0.972 (0.004) & 0.673 (0.033) \\
    \bottomrule
    \end{tabular}
\end{table}

Table \ref{tab:ranking_consistency} demonstrates that CPA retains strong discriminative power across diverse regimes. In the heteroscedastic setting, where local coverage varies substantially (Setting C), CPA closely matches the Oracle ranking ($\tau_w > 0.9$). In the high-noise and strongly correlated settings (Settings B and D), exact ordering and selection utility diverge: stochastic fluctuations can permute statistically similar mid-ranked methods, which lowers $\tau_w$, but the uniformly high NDCG scores show that CPA still identifies the best candidate reliably. Even in the saturation regime (Setting A), where several methods perform nearly identically and Hit@3 is therefore less informative, the high global correlation indicates that CPA correctly separates the strong-performing methods from clearly inferior ones.

\begin{figure}[!b]
    \centering
    \includegraphics[width=.8\textwidth]{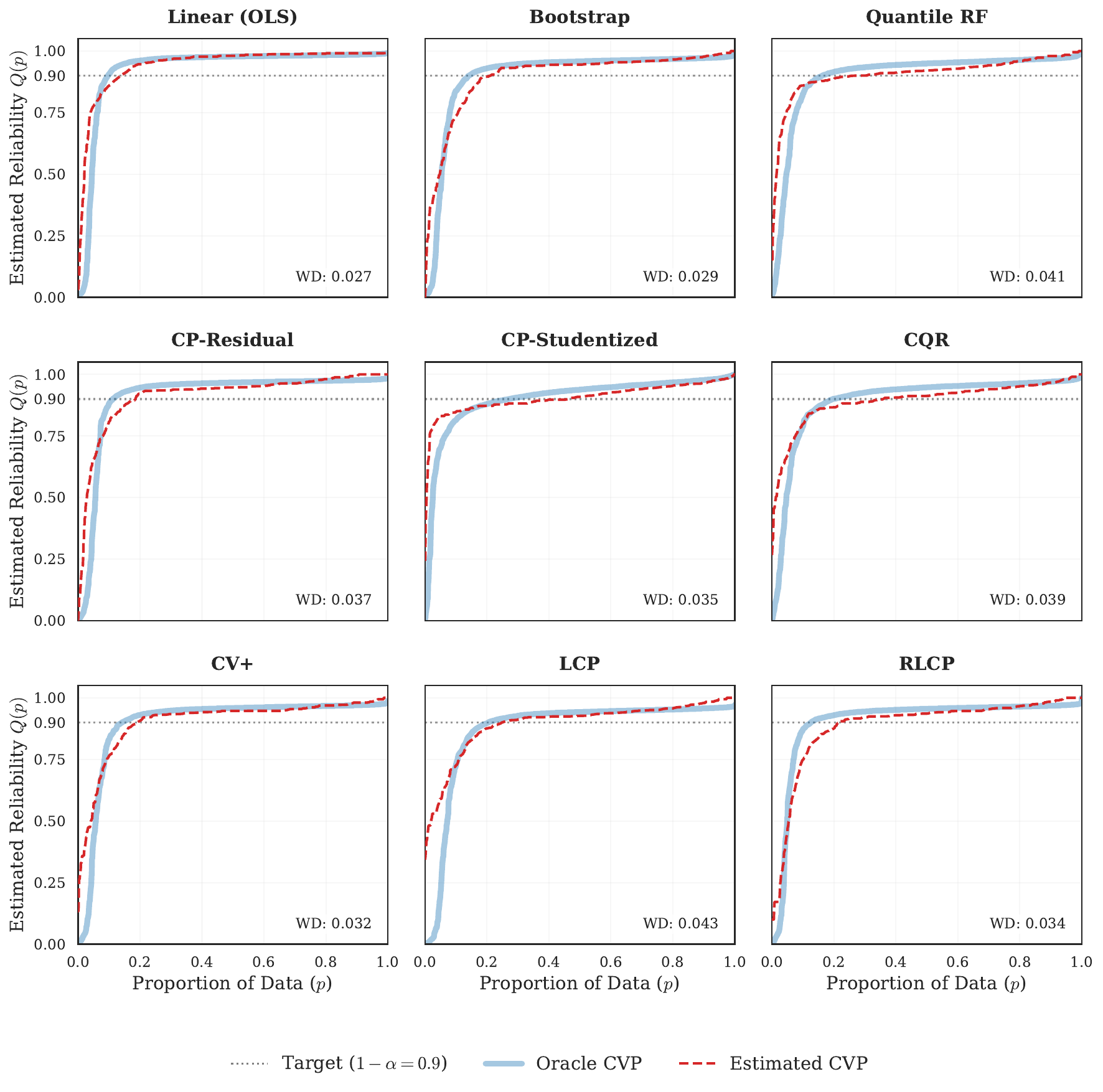}
    \caption{\textbf{Distributional recovery in Setting B (nonlinear, Heavy-tailed).} 
    The estimated CVP (\textbf{red dashed}) closely tracks the ground-truth Oracle distribution (\textbf{blue}), demonstrating that CPA effectively recovers the full reliability landscape. The annotated \textbf{WD} denotes the Wasserstein Distance, quantifying the discrepancy between the estimated and Oracle distributions (lower is better).}
    \label{fig:cvp_matrix_setting_B}
\end{figure}

\subsubsection{Distributional Recovery via CVP}
Beyond scalar metrics, we also examine \textbf{distributional recovery} by comparing the estimated reliability profile with the \textbf{Oracle CVP}, the empirical quantile function of the true conditional coverage probabilities. Figure~\ref{fig:cvp_matrix_setting_B} for Setting B shows that the estimated CVP closely tracks the Oracle distribution. CPA recovers the characteristic profile of rigid methods such as CP-Residual, which tend to over-cover easy points while failing on outliers, and it also preserves the more balanced profiles of adaptive methods such as CQR and CP-Studentized. These results indicate that CPA captures the full distribution of reliability rather than only an aggregate summary, enabling a more granular diagnosis than marginal coverage alone.\looseness=-1

\subsection{Robustness to Reliability Estimator Misspecification}
\label{subsec:robustness}

A critical practical question is how sensitive CPA is to the specification of the reliability estimator $\hat{\eta}$. To study this systematically, we compare our recommended \textbf{AutoML Baseline}---an ensemble with isotonic calibration---against six perturbations chosen to expose specific failure modes. These variants range from uncalibrated estimators to models with severe \textbf{structural misspecification}; Appendix \ref{app:estimator_configs} provides the full specifications.

\begin{figure}[!b]
    \centering
    \includegraphics[width=0.8\textwidth]{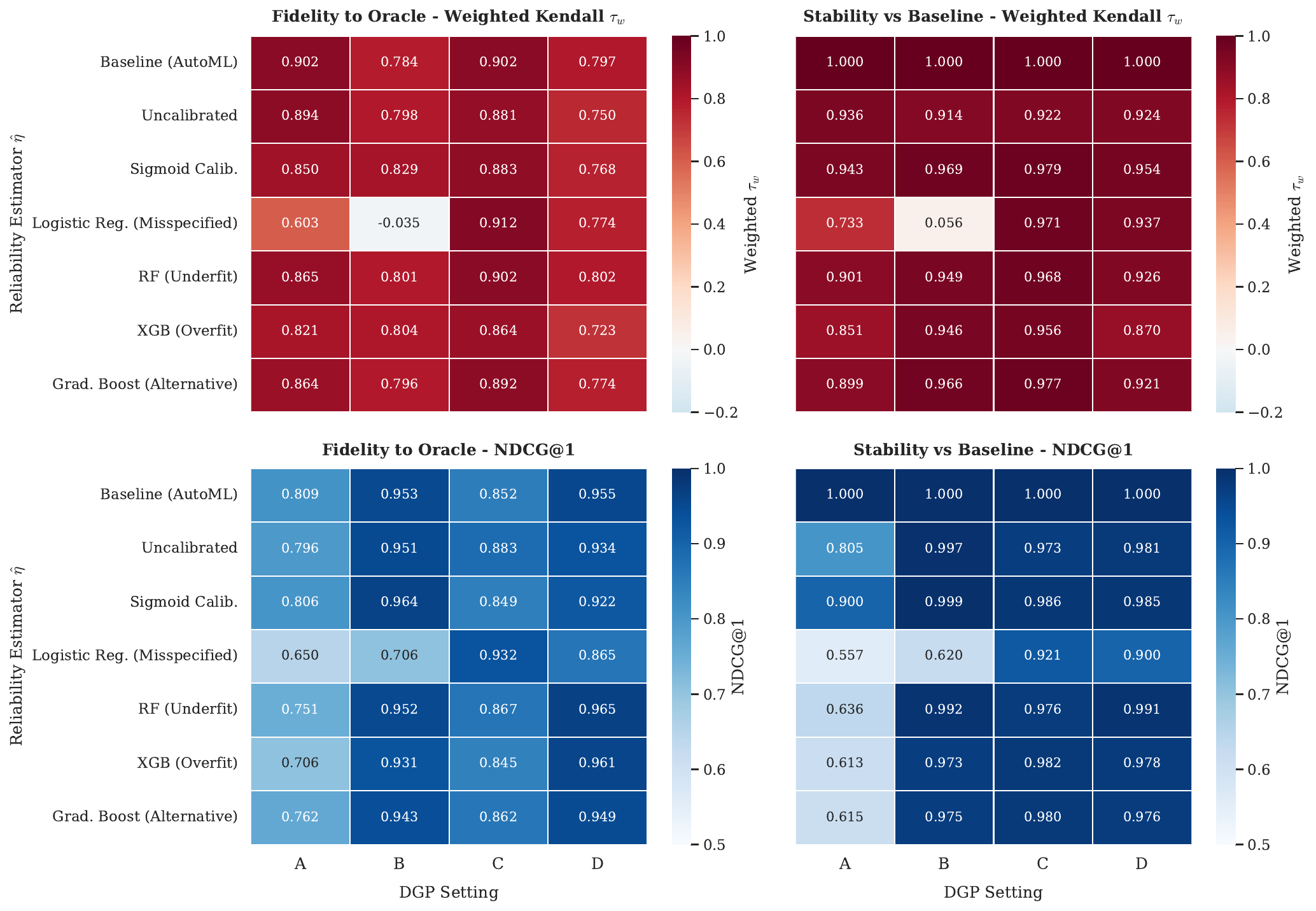}
    \caption{\textbf{Robustness analysis of ranking fidelity (Left) and stability (Right).} 
    Heatmaps compare reliability estimators ($y$-axis) across DGP settings ($x$-axis). 
    The \textbf{Left Column} measures fidelity against the ground-truth Oracle ($\hat{\pi}$ vs. $\pi^{\text{oracle}}$), while the \textbf{Right Column} assesses stability relative to the Baseline ($\hat{\pi}$ vs. $\pi^{\text{base}}$).}
    \label{fig:robustness_heatmap}
\end{figure}

Figure~\ref{fig:robustness_heatmap} visualizes how these perturbations affect ranking fidelity. Two patterns emerge:
\begin{itemize}
   \item \textbf{Selection utility is stable.} CPA usually identifies the top-performing models, as reflected by the consistently high NDCG@1 values. The main exception is \texttt{Logistic Regression} in nonlinear regimes such as Setting B, where insufficient model capacity limits performance.
  \item \textbf{Monotonicity primarily drives ranking.} In most settings, the \texttt{Baseline} and \texttt{Uncalibrated} estimators achieve similar NDCG values, suggesting that preserving the correct ordering of reliability is more important for model selection than perfect probabilistic calibration. Calibration remains essential for accurate absolute auditing, but it is secondary for relative ranking.
\end{itemize}

However, nonlinear Setting B exposes a clear failure mode: the \texttt{Logistic Regression} estimator suffers a severe drop in rank correlation. As the bias analysis in Appendix \ref{app:bias_analysis} explains, this breakdown is caused by \textbf{structural misspecification}: a linear decision boundary cannot capture the complex level sets of the true conditional coverage function. This finding highlights a key practical prerequisite for CPA: the reliability estimator must have sufficient \textbf{model complexity}.

\subsection{Discussion of Simulation Results}
\label{subsec:sim_discussion}

Overall, the simulation results validate CPA as a robust tool for auditing uncertainty quantification. The CVI metric and CVP curves consistently recover the ground-truth reliability landscape, both for ranking and for distributional diagnosis. Based on these findings, we recommend an \textbf{AutoML-based ensemble with Isotonic Calibration} as the default reliability estimator because it is robust to model mismatch. For data allocation, the sensitivity analysis in \textbf{Appendix \ref{app:sensitivity_rho}} suggests that a balanced split ratio ($\rho=0.5$) best balances approximation and estimation error. Finally, reliable auditing requires enough data to learn the failure boundary; our empirical results suggest a minimum sample size of roughly $N \approx 800$ for standard regression tasks (Appendix \ref{app:sample_complexity}).

\section{Real-Data Applications}

\subsection{Datasets and Experimental Protocol}
\label{subsec:real-data-protocol}

\subsubsection{Benchmark Datasets}

Table~\ref{tab:real-datasets-main} summarizes the nine benchmark datasets used in our evaluation, which are drawn from the real-data benchmarks considered in \citet{romano2019conformalized} and \citet{agarwal2025pcs}. This choice aligns our empirical study with established evaluation protocols in the literature, rather than relying on arbitrarily selected datasets. The table reports their scale, dimensionality, prediction target, and application domain. All tasks are standard real-valued regression problems.

\begin{table}[b]
  \centering
  \renewcommand{\arraystretch}{} %
  \caption{Summary of real-world regression tasks used in the empirical study.}
  \label{tab:real-datasets-main}
  \resizebox{0.8\linewidth}{!}{
  \begin{tabular}{lS[group-separator={,}]rl l}
    \toprule
    Dataset & {$n$} & {$p$} & Target Variable & Domain \\
    \midrule
    Bike          & 10886 & 18  & Hourly rental count & Transportation \\
    Computer      & 8192  & 21  & CPU execution time  & Computer Systems \\
    Debutanizer   & 2400  & 7   & Butane content      & Industrial Control \\
    Kin8nm        & 8192  & 8   & Kinematic response  & Physics \\
    Meps\_21      & 15656 & 139 & Health expenditure  & Healthcare \\
    Miami\_2016  & 13932 & 15  & Sale price          & Real Estate \\
    Parkinsons    & 5875  & 18  & UPDRS motor score   & Biomedicine \\
    Qsar          & 5742  & 500 & Molecular activity  & Chemistry \\
    Temperature   & 7590  & 21  & Hourly temperature  & Weather \\
    \bottomrule
  \end{tabular}}
\end{table}

\subsubsection{Experimental Protocol}

For each dataset, we use a nested data-splitting design to separate model assessment from final evaluation. Across $10$ independent repetitions, the data are split into an 80\% \emph{Master Train} set and a 20\% \emph{Master Test} set, with the latter reserved for final assessment.

Within the \emph{Master Train} set, we perform $K=5$ random 50/50 splits to train candidate conformal predictors and estimate reliability. In split $k$, models are fitted on an \emph{Internal Train} subset and evaluated on an \emph{Internal Validation} subset, yielding coverage labels $I_i=\mathbbm{1}\{Y_i\in\widehat{\mathcal{C}}(X_i)\}$ for fitting the reliability estimator $\hat{\eta}^{(k)}(x)$. CC-Select uses these estimators for model choice. For final evaluation, the selected predictor is retrained on the full \emph{Master Train} set and assessed on the \emph{Master Test} set, while the reliability score for a test point $X_j$ is computed by ensembling the inner estimators, $\hat{\eta}(X_j)=\frac{1}{K}\sum_{k=1}^K \hat{\eta}^{(k)}(X_j)$. These scores are used to construct reliability diagrams and conditional performance metrics on the held-out test data.

\subsection{Reliability Checking}
\label{sec:reliability-checking}
We assess the calibration of the reliability estimator $\hat{\eta}(x)$ using two standard diagnostics: \textbf{Reliability Diagrams}, which compare predicted probabilities with empirical coverage, and the \textbf{Expected Calibration Error (ECE)}, which measures the weighted average absolute deviation from perfect calibration. Formal definitions and implementation details are provided in Appendix~\ref{app:calibration_methodology}.

\textbf{Results.} \;
Figure~\ref{fig:real-reliability-grid} shows that $\hat{\eta}(x)$ is well calibrated across all nine datasets: empirical coverage closely tracks predicted reliability, with ECE values typically below 0.01. Although these diagnostics do not establish learnability in full generality, they support that the learned reliability surface is a reliable practical proxy for conditional coverage, validating the use of CPA for auditing conformal predictors.

\begin{figure}[!b]
  \centering
  \includegraphics[width=.8\textwidth]{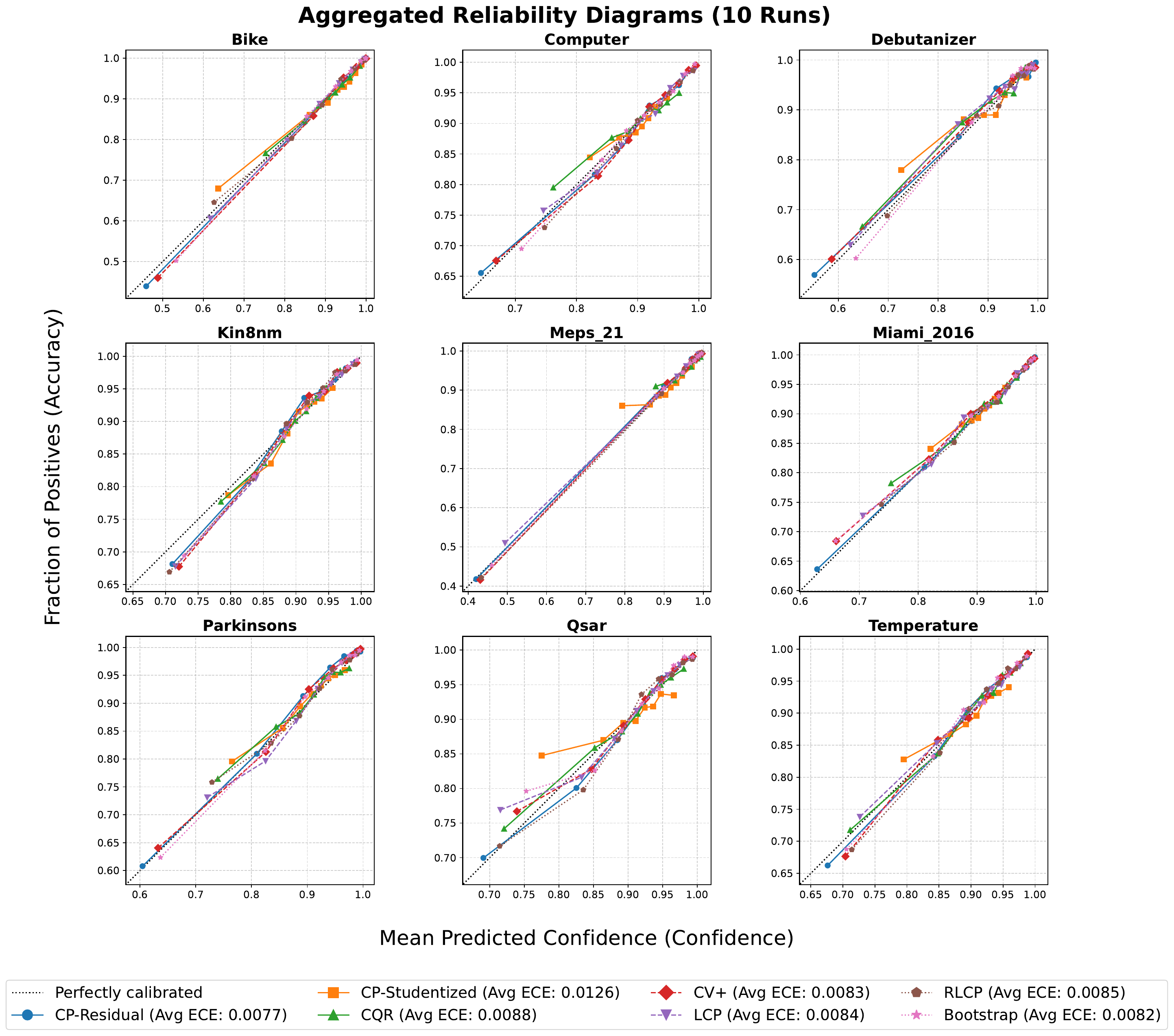}
  \caption{\textbf{Reliability diagrams for CPA across nine datasets.}
    Each plot compares empirical coverage (y-axis) with predicted reliability (x-axis) using $B=10$ equal-frequency bins.\looseness = -1}
  \label{fig:real-reliability-grid}
\end{figure}

\subsection{Model Selection}
\label{subsec:real-data-model-selection}

We assess the efficacy of the \textbf{CC-Select} algorithm (Algorithm \ref{alg:cc_select}) in identifying conformal prediction procedures that optimize conditional coverage. This is the operational step through which CPA turns diagnostics into action when multiple plausible refinements are available. \looseness = -1

\subsubsection{Experimental Setup and Metrics}

We apply the nested evaluation protocol defined in Section \ref{subsec:real-data-protocol} over $R=10$ independent splits. The candidate pool includes diverse conformal predictors such as residual-based methods (\texttt{CP-Residual}, \texttt{CP-Studentized}), quantile-based methods (\texttt{CQR}), and adaptive strategies (\texttt{LCP}, \texttt{RLCP}, \texttt{CV+}). Crucially, the reliability estimator $\hat{\eta}$ used for selection is derived from an inner-loop automated pipeline. This pipeline involves a grid search over diverse classifiers optimized via 5-fold cross-validation, and finalized with isotonic calibration. 

Performance on the held-out $\emph{Master Test}$ set is evaluated using three complementary metrics: \textbf{Marginal Coverage} (target $1-\alpha=0.9$), \textbf{Average Interval Length} (predictive efficiency), and the \textbf{Worst-Slab Coverage (WSC)} \citep{JMLR:v22:20-753}. WSC, a practical finite-sample proxy for conditional validity, measures the minimum empirical coverage across disjoint strata of the covariate space. Details on WSC are provided in Appendix \ref{app:wsc_details}.

\subsubsection{Results}

\begin{figure}[!b]
    \centering
    \includegraphics[width=.9\textwidth]{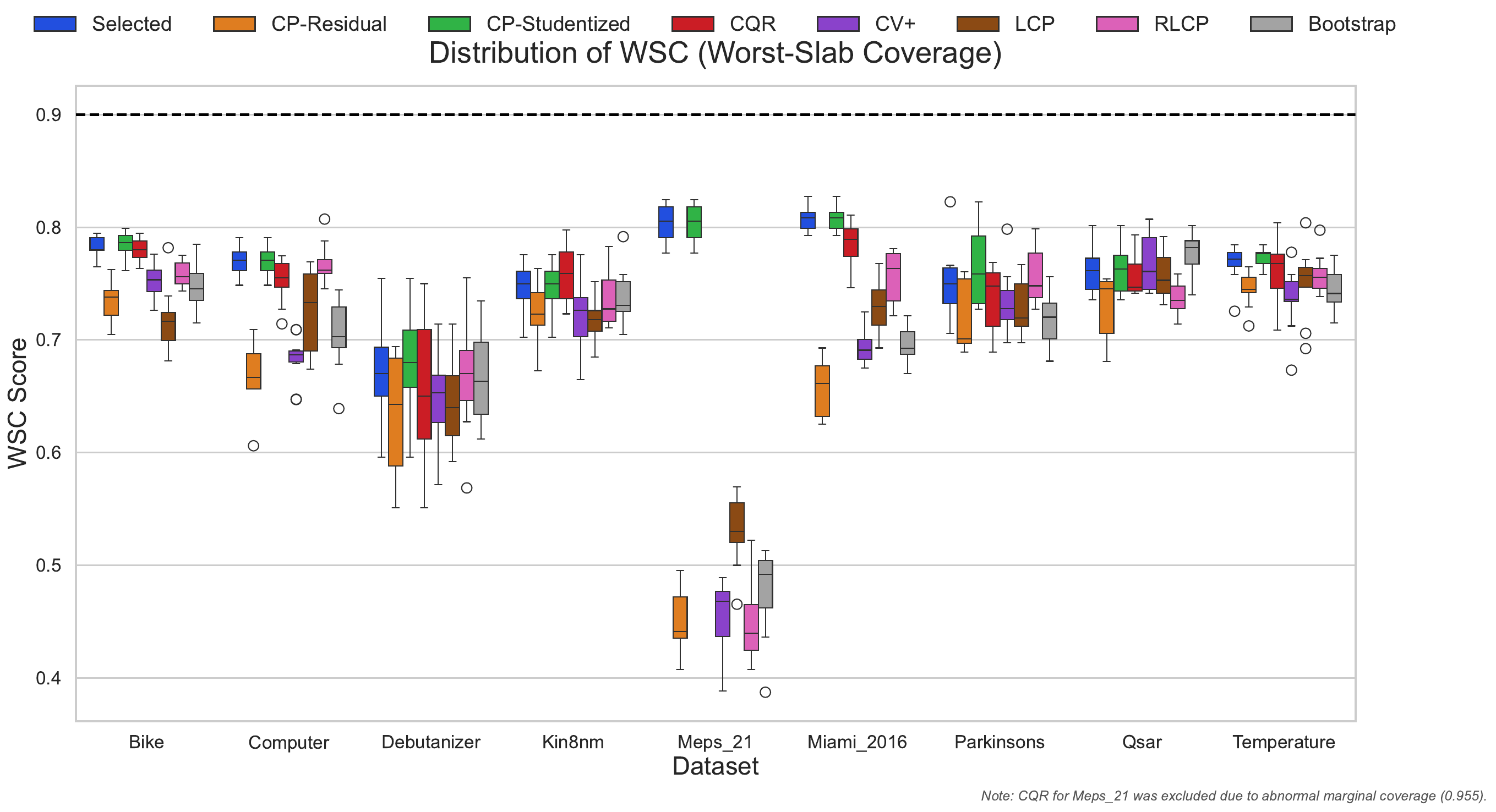}
    \caption{\textbf{Distribution of Worst-Slab Coverage (WSC) scores.}  The procedure selected by CC-Select (``Selected'', blue box) consistently achieves high WSC scores.}
    \label{fig:real-data-wsc}
\end{figure}

Figure~\ref{fig:real-data-wsc} displays the distribution of WSC scores across the nine datasets. The procedure selected by CC-Select (labeled ``Selected'') consistently achieves high WSC scores, often outperforming individual candidates. Detailed numerical results, including mean interval lengths and marginal coverage levels, are available in Table \ref{tab:full_results_appendix} in Appendix \ref{app:detailed_results}. It is important to note that, on certain datasets, some 
CP methods can perform extremely poorly (e.g., on the Meps\_21 dataset). This observation highlights the necessity of CPA-type tools for systematically auditing CP methods, in order to detect and prevent inadequate or even severely miscalibrated coverage behavior. Collectively, these empirical results confirm that the CVI metric serves as an effective, data-driven criterion for auditing and selecting safe conformal predictors in complex real-world environments.

\subsection{Practical Utility of CPA: Diagnosis, Refinement, and Selection}
\label{subsec:bike-case-study}

CPA is useful not only for auditing conditional validity, but also for diagnosing failure modes of a conformal predictor and guiding the choice among practical refinements. The framework identifies where an existing method fails and compares alternative conformal pipelines through their estimated reliability patterns.

We illustrate this role on the Bike dataset, where the target is hourly rental count and the covariates include temporal, calendar, and weather variables. We take CP-Residual as the baseline conformal predictor and consider two more adaptive variants, CP-Studentized and CQR. The goal is to show how CPA diagnoses the baseline method and helps select a more conditionally valid alternative.

We begin with CP-Residual. Although it attains marginal coverage close to the nominal level, its Conditional Validity Profile in Figure~\ref{fig:cvp_cp_residual} reveals substantial heterogeneity in conditional coverage, with both undercoverage and overcoverage across the population. The issue is therefore not marginal calibration, but poor conditional validity across the feature space. To localize this heterogeneity, we examine reliability as a function of \texttt{hour}. Figure~\ref{fig:hour_reliability} shows clear undercoverage during peak commuting hours, where a homoscedastic residual score is inadequate.

\begin{figure}[!t]
\centering

\begin{subfigure}[t]{0.49\textwidth}
    \centering
    \includegraphics[width=\textwidth]{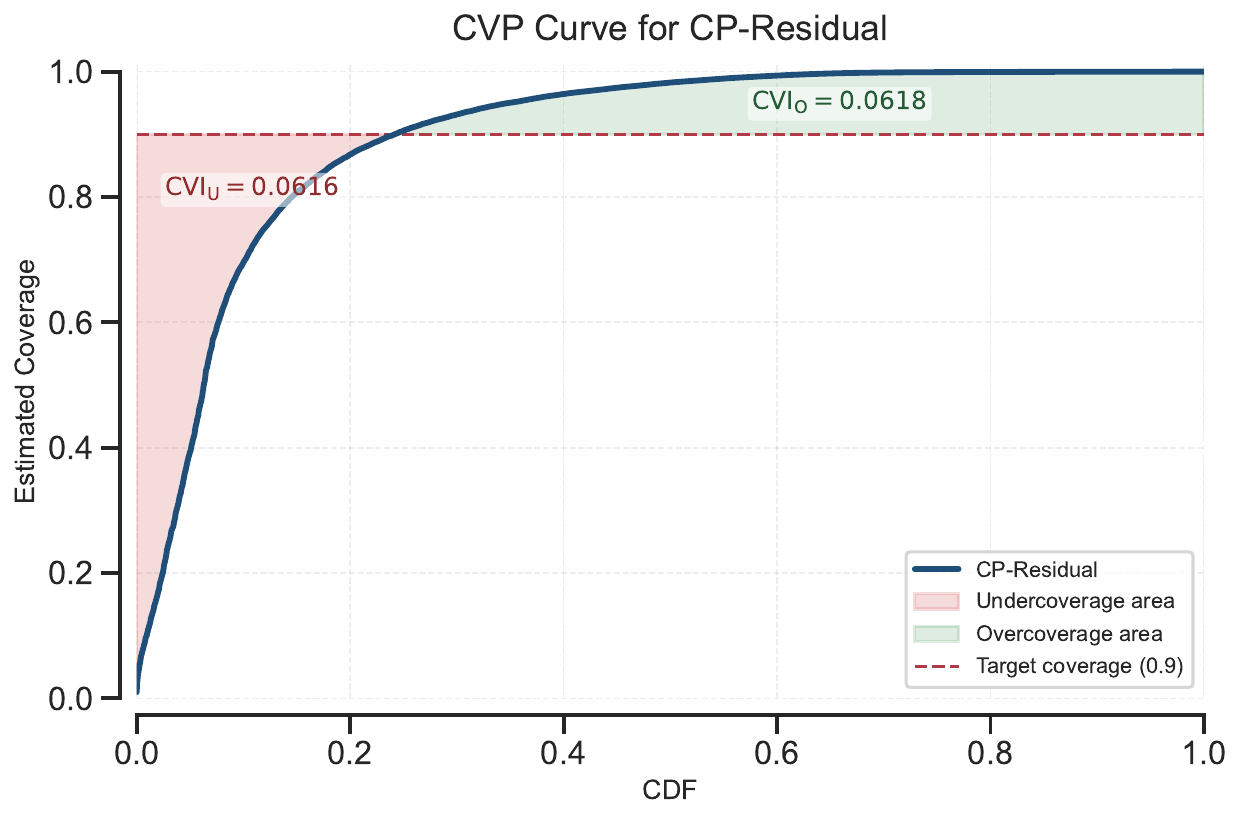}
    \caption{CVP curve for CP-Residual.}
    \label{fig:cvp_cp_residual}
\end{subfigure}
\hfill
\begin{subfigure}[t]{0.49\textwidth}
    \centering
    \includegraphics[width=\textwidth]{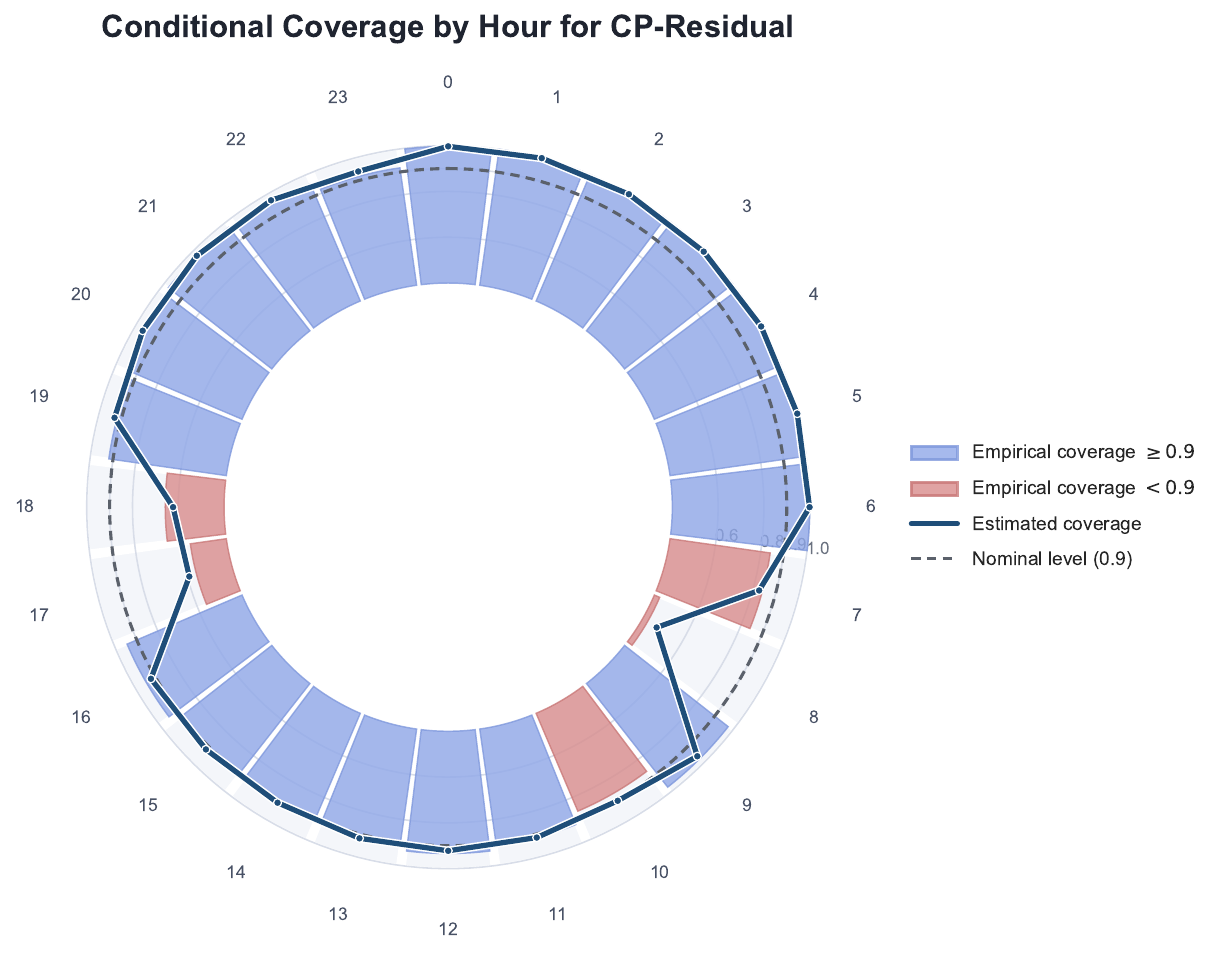}
    \caption{Conditional coverage for CP-Residual.}
    \label{fig:hour_reliability}
\end{subfigure}

\caption{Diagnostics for the baseline conformal predictor CP-Residual on the Bike dataset. 
Left: the Conditional Validity Profile (CVP) reveals substantial heterogeneity in conditional coverage despite near-nominal marginal coverage. 
Right: the hourly reliability pattern shows clear undercoverage during peak commuting hours, and the estimated curve broadly follows the empirical trend.}
\label{fig:cp_residual_diagnostics}
\end{figure}

This diagnosis naturally motivates more adaptive score constructions. A failure pattern concentrated in high-variability regimes suggests methods that better adapt interval width to local uncertainty. CP-Studentized rescales residuals by a local variability estimate, whereas CQR constructs intervals through conditional quantile estimation. Figure~\ref{fig:cvp_refinements} shows that both refinements yield CVP curves closer to the nominal target than CP-Residual, indicating improved conditional validity at the population level. The same pattern appears at the feature level: in Figure~\ref{fig:hour_refinements}, both methods reduce the undercoverage during peak commuting hours, with CQR exhibiting the most stable hourly reliability pattern.

\begin{figure}[t]
\centering

\begin{subfigure}[t]{0.48\textwidth}
    \centering
    \includegraphics[width=\textwidth]{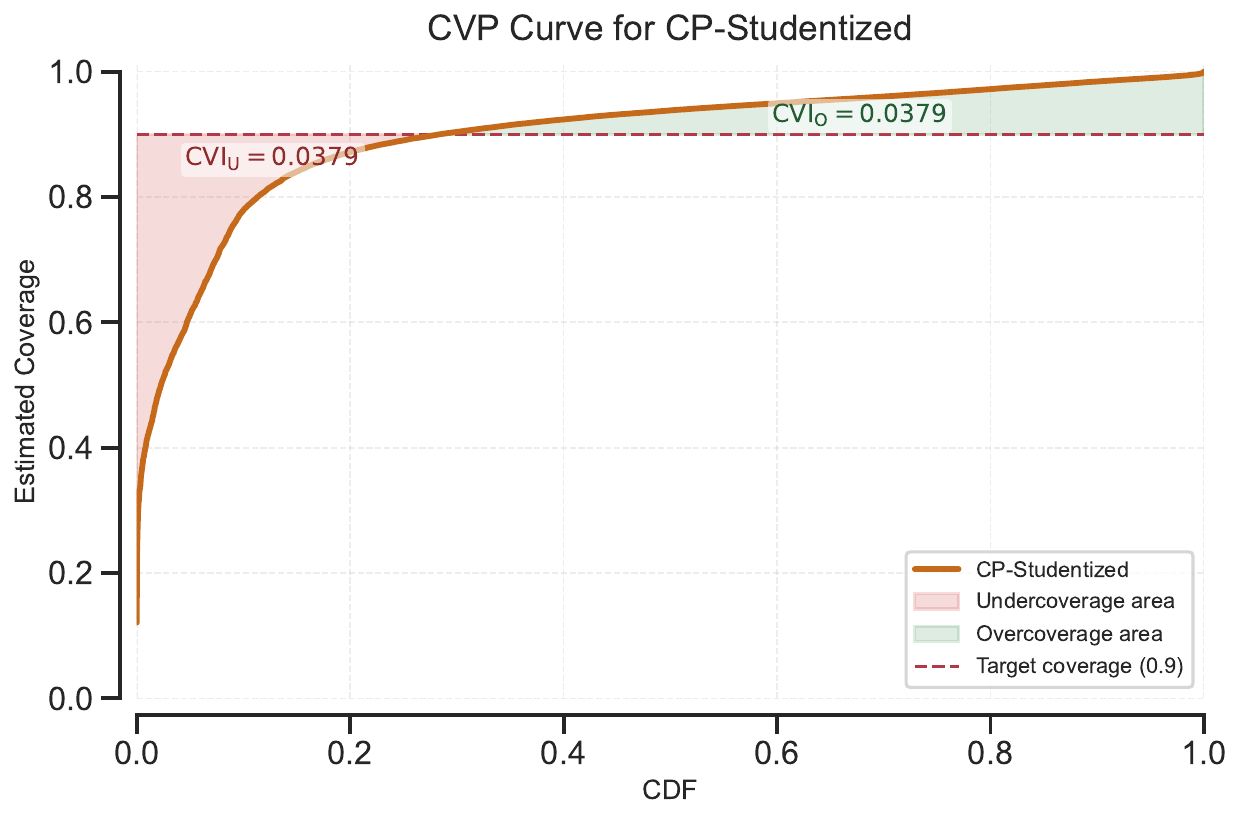}
    \caption{CP-Studentized}
    \label{fig:cvp_cp_student}
\end{subfigure}
\hfill
\begin{subfigure}[t]{0.48\textwidth}
    \centering
    \includegraphics[width=\textwidth]{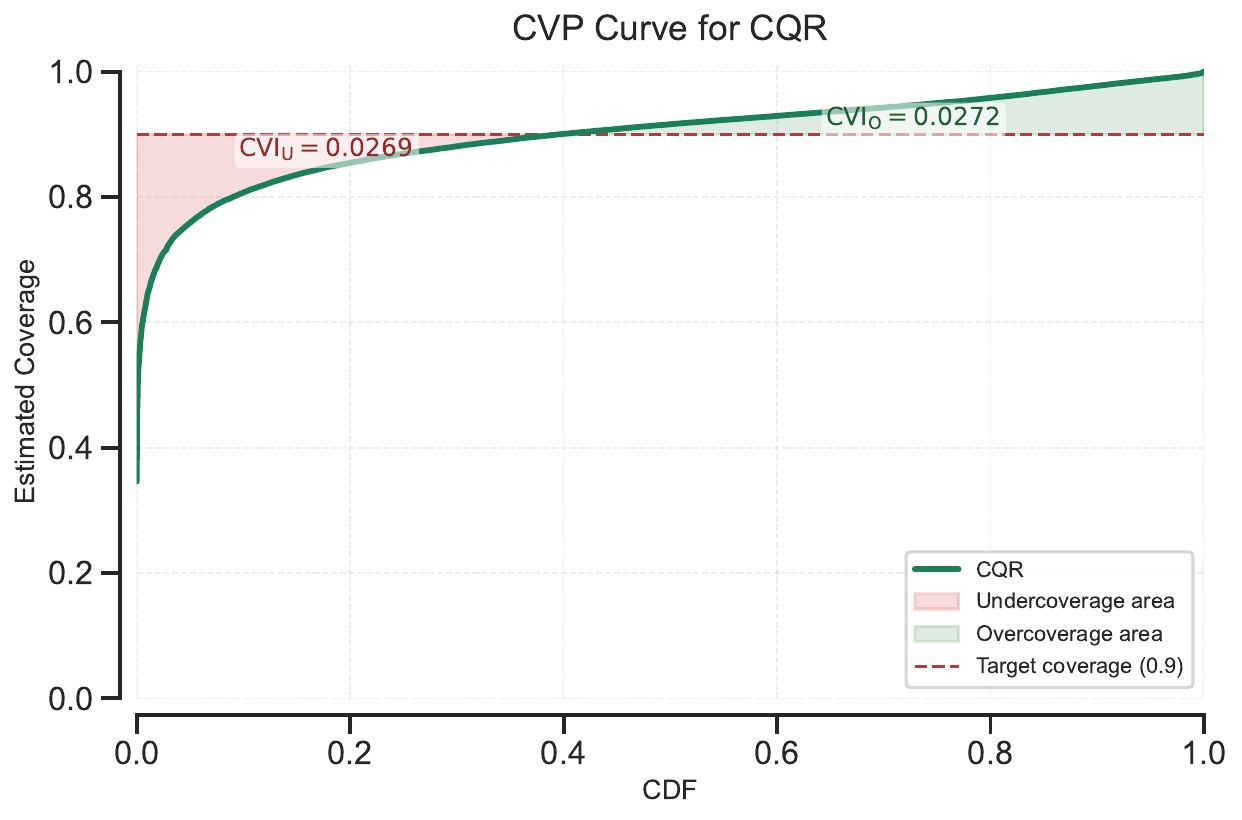}
    \caption{CQR}
    \label{fig:cvp_cqr}
\end{subfigure}

\caption{Conditional Validity Profile (CVP) curves for two adaptive refinements on the Bike dataset. Compared with CP-Residual, both CP-Studentized and CQR exhibit improved conditional validity, with CVP curves closer to the nominal target.}
\label{fig:cvp_refinements}
\end{figure}

Accordingly, CPA ranks the three methods as CQR first, CP-Studentized second, and CP-Residual last on this dataset. This example makes the practical role of CPA concrete: it reveals where a baseline conformal predictor fails, motivates targeted refinements, and supports the selection of a more reliable conformal method. Additional two-dimensional diagnostics over \texttt{temp} and \texttt{hour} are reported in Appendix~\ref{app:bike-multivariate}.

\begin{figure}[!b]
\centering

\begin{subfigure}[t]{0.48\textwidth}
    \centering
    \includegraphics[width=\textwidth]{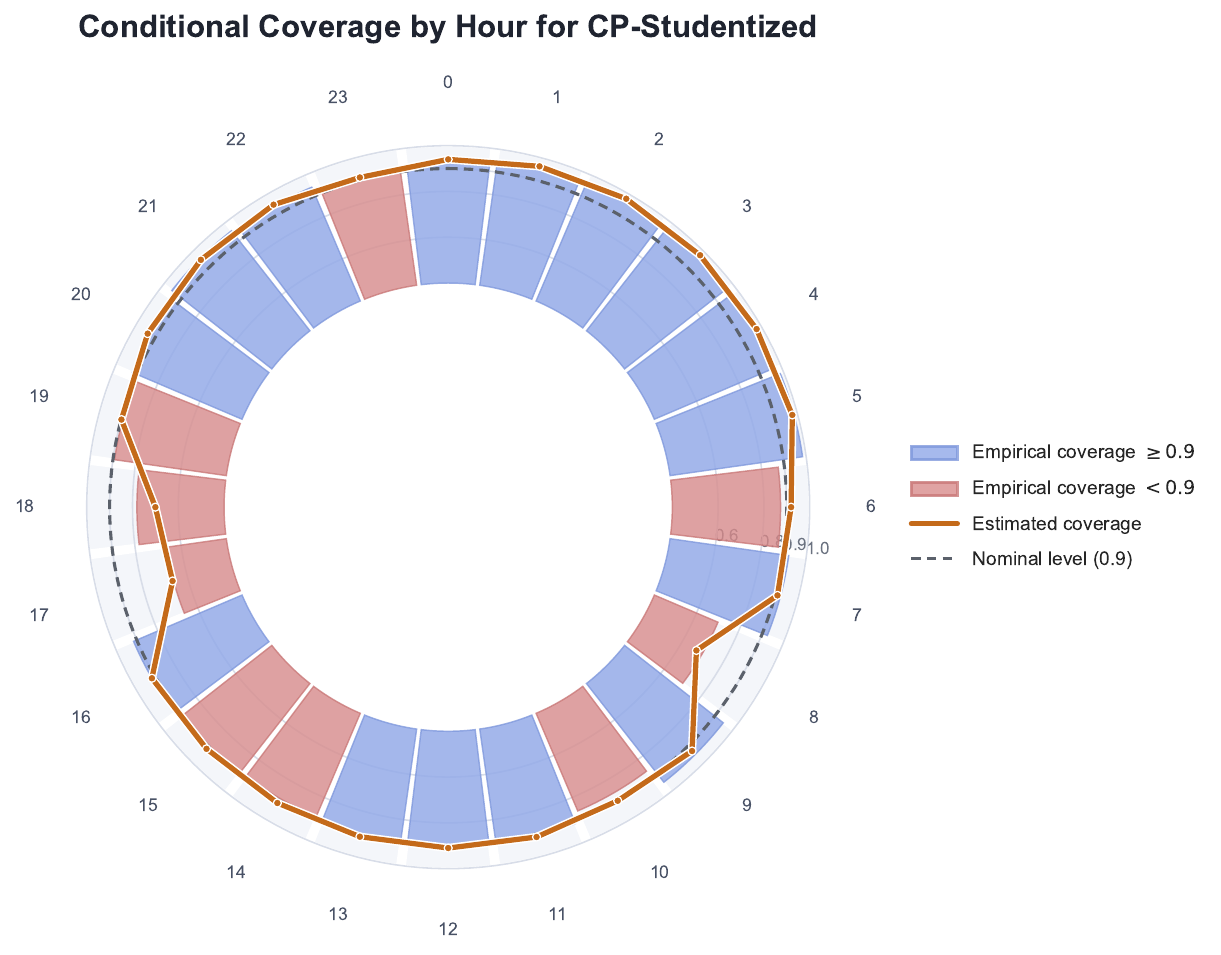}
    \caption{CP-Studentized}
    \label{fig:hour_cp_student}
\end{subfigure}
\hfill
\begin{subfigure}[t]{0.48\textwidth}
    \centering
    \includegraphics[width=\textwidth]{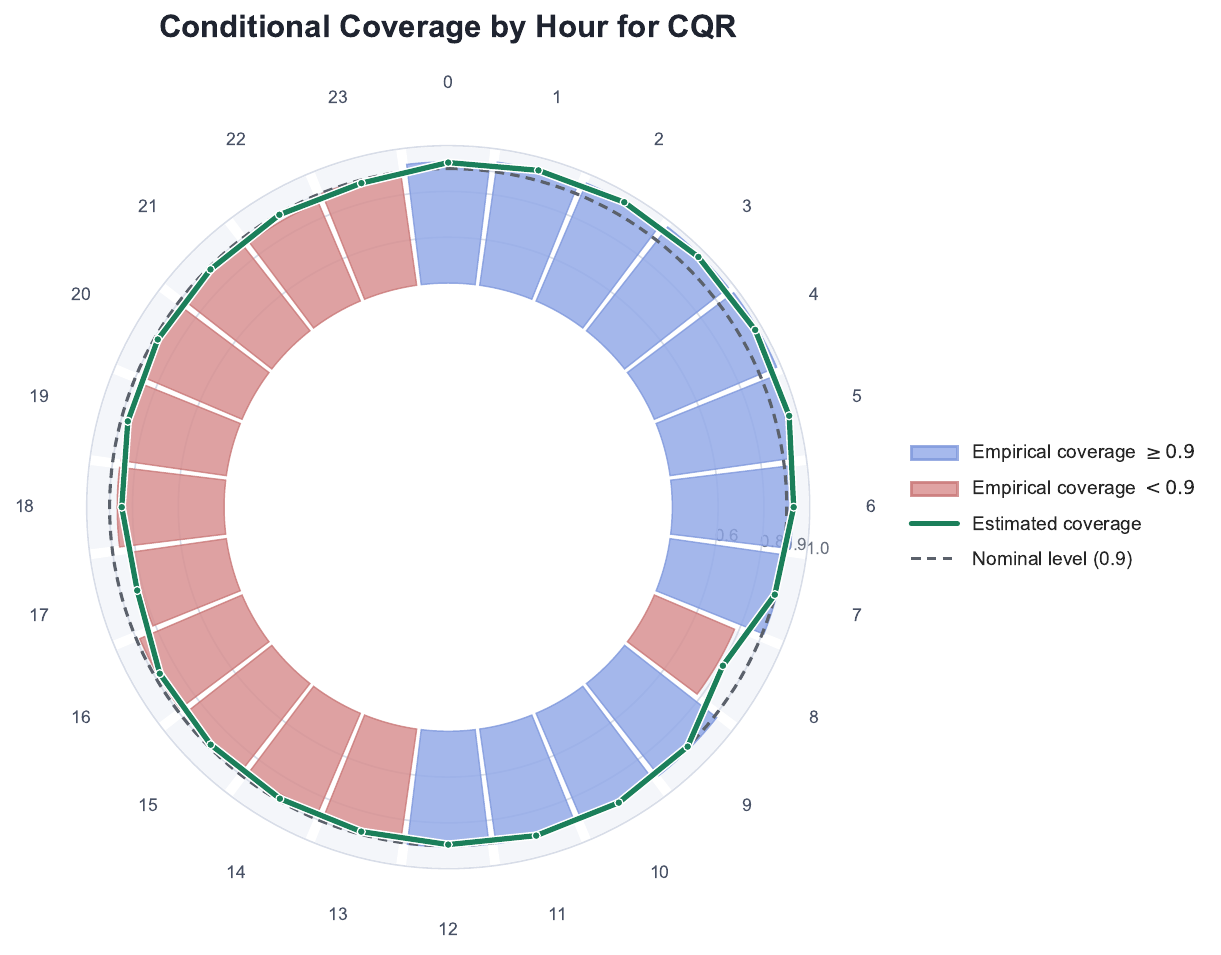}
    \caption{CQR}
    \label{fig:hour_cqr}
\end{subfigure}

\caption{Conditional coverage by hour for two adaptive refinements on the Bike dataset. Compared with CP-Residual, both CP-Studentized and CQR reduce the undercoverage during peak commuting hours, with CQR exhibiting the most stable hourly pattern.}
\label{fig:hour_refinements}
\end{figure}

\section{Conclusion}

We have introduced Conformal Prediction Assessment (CPA), a unified framework for evaluating and comparing the conditional coverage behaviors of conformal prediction methods. By formulating conditional validity assessment as a supervised learning problem, CPA enables scalable, instance-level auditing of coverage reliability beyond marginal guarantees.

CPA is built around a learned reliability estimator that approximates the conditional coverage function. This estimator underlies the Conditional Validity Index (CVI), which, together with a decomposition into safety and efficiency components, provides an interpretable summary of conditional miscalibration. We establish convergence of the reliability estimator, consistency of the CVI, and asymptotic validity of CVI-based model selection, grounding the proposed assessment in a sound theoretical framework.

Practically, CPA offers a principled alternative to binning-based diagnostics and demonstrates how conditional validity can guide model selection through the CC-Select procedure. The learned reliability scores may also be used as deployment-time indicators for identifying regions of unreliable uncertainty quantification. More broadly, CPA is most useful when it helps practitioners improve existing conformal pipelines through diagnosis-guided refinement. \looseness = -1

Overall, CPA formalizes conditional coverage as a learnable object, enabling systematic evaluation and selection of conformal prediction methods, with direct implications for reliable deployment in heterogeneous settings.

\bibliography{references}

@article{topol2019high,
  title={High-performance medicine: the convergence of human and artificial intelligence},
  author={Topol, Eric J},
  journal={Nature medicine},
  volume={25},
  number={1},
  pages={44--56},
  year={2019},
  publisher={Nature Publishing Group US New York}
}

@article{berk2021fairness,
  title={Fairness in criminal justice risk assessments: The state of the art},
  author={Berk, Richard and Heidari, Hoda and Jabbari, Shahin and Kearns, Michael and Roth, Aaron},
  journal={Sociological Methods \& Research},
  volume={50},
  number={1},
  pages={3--44},
  year={2021},
  publisher={Sage Publications Sage CA: Los Angeles, CA}
}

@article{fuster2022predictably,
  title={Predictably unequal? The effects of machine learning on credit markets},
  author={Fuster, Andreas and Goldsmith-Pinkham, Paul and Ramadorai, Tarun and Walther, Ansgar},
  journal={The Journal of Finance},
  volume={77},
  number={1},
  pages={5--47},
  year={2022},
  publisher={Wiley Online Library}
}

@book{Vovk2005,
  title={Algorithmic Learning in a Random World},
  author={Vovk, Vladimir G and Gammerman, Alex and Shafer, Glenn},
  year={2005},
  publisher={Springer Science \& Business Media}
}

@article{lei2014distribution,
  title={Distribution-free prediction bands for non-parametric regression},
  author={Lei, Jing and Wasserman, Larry},
  journal={Journal of the Royal Statistical Society Series B: Statistical Methodology},
  volume={76},
  number={1},
  pages={71--96},
  year={2014},
  publisher={Oxford University Press}
}

@article{lei2018distribution,
  title={Distribution-free predictive inference for regression},
  author={Lei, Jing and G’Sell, Max and Rinaldo, Alessandro and Tibshirani, Ryan J and Wasserman, Larry},
  journal={Journal of the American Statistical Association},
  volume={113},
  number={523},
  pages={1094--1111},
  year={2018},
  publisher={Taylor \& Francis}
}

@inproceedings{vovk2012conditional,
  title={Conditional validity of inductive conformal predictors},
  author={Vovk, Vladimir},
  booktitle={Asian conference on machine learning},
  pages={475--490},
  year={2012},
  organization={PMLR}
}

@article{foygel2021limits,
  title={The limits of distribution-free conditional predictive inference},
  author={Foygel Barber, Rina and Candes, Emmanuel J and Ramdas, Aaditya and Tibshirani, Ryan J},
  journal={Information and Inference: A Journal of the IMA},
  volume={10},
  number={2},
  pages={455--482},
  year={2021},
  publisher={Oxford University Press}
}

@article{yang2007consistency,
author = {Yuhong Yang},
title = {{Consistency of cross validation for comparing regression procedures}},
volume = {35},
journal = {The Annals of Statistics},
number = {6},
publisher = {Institute of Mathematical Statistics},
pages = {2450 -- 2473},
year = {2007},
}

@article{romano2019conformalized,
  title={Conformalized quantile regression},
  author={Romano, Yaniv and Patterson, Evan and Candes, Emmanuel},
  journal={Advances in neural information processing systems},
  volume={32},
  year={2019}
}

@article{lei2013distribution,
  title={Distribution-free prediction sets},
  author={Lei, Jing and Robins, James and Wasserman, Larry},
  journal={Journal of the American Statistical Association},
  volume={108},
  number={501},
  pages={278--287},
  year={2013},
  publisher={Taylor \& Francis}
}

@article{sesia2021conformal,
  title={Conformal prediction using conditional histograms},
  author={Sesia, Matteo and Romano, Yaniv},
  journal={Advances in neural information processing systems},
  volume={34},
  pages={6304--6315},
  year={2021}
}

@article{guan2023localized,
  title={Localized conformal prediction: A generalized inference framework for conformal prediction},
  author={Guan, Leying},
  journal={Biometrika},
  volume={110},
  number={1},
  pages={33--50},
  year={2023},
  publisher={Oxford University Press}
}

@article{hore2025conformal,
  title={Conformal prediction with local weights: randomization enables robust guarantees},
  author={Hore, Rohan and Barber, Rina Foygel},
  journal={Journal of the Royal Statistical Society Series B: Statistical Methodology},
  volume={87},
  number={2},
  pages={549--578},
  year={2025},
  publisher={Oxford University Press UK}
}

@article{gibbs2025conformal,
  title={Conformal prediction with conditional guarantees},
  author={Gibbs, Isaac and Cherian, John J and Cand{\`e}s, Emmanuel J},
  journal={Journal of the Royal Statistical Society Series B: Statistical Methodology},
  volume = {87},
  number = {4},
  pages = {1100-1126},
  year={2025},
  month = {03},
  publisher={Oxford University Press UK}
}

@article{yang2025selection,
  title={Selection and aggregation of conformal prediction sets},
  author={Yang, Yachong and Kuchibhotla, Arun Kumar},
  journal={Journal of the American Statistical Association},
  volume={120},
  number={549},
  pages={435--447},
  year={2025},
  publisher={Taylor \& Francis}
}

@article{liang2024conformal,
  title={Conformal prediction after efficiency-oriented model selection},
  author={Liang, Ruiting and Zhu, Wanrong and Barber, Rina Foygel},
  journal={arXiv preprint arXiv:2408.07066},
  year={2024}
}

@article{chakraborti2025personalized,
  title={Personalized uncertainty quantification in artificial intelligence},
  author={Chakraborti, Tapabrata and Banerji, Christopher RS and Marandon, Ariane and Hellon, Vicky and Mitra, Robin and Lehmann, Brieuc and Br{\"a}uninger, Leandra and McGough, Sarah and Turkay, Cagatay and Frangi and others},
  journal={Nature Machine Intelligence},
  volume={7},
  number={4},
  pages={522--530},
  year={2025},
}

@article{agarwal2025pcs,
  title={PCS-UQ: Uncertainty Quantification via the Predictability-Computability-Stability Framework},
  author={Agarwal, Abhineet and Xiao, Michael and Barter, Rebecca and Ronen, Omer and Fan, Boyu and Yu, Bin},
  journal={arXiv preprint arXiv:2505.08784},
  year={2025}
}

@article{barber2021predictive,
  title={Predictive inference with the jackknife+},
  author={Barber, Rina Foygel and Candes, Emmanuel J and Ramdas, Aaditya and Tibshirani, Ryan J},
  journal={The Annals of Statistics},
  volume={49},
  number={1},
  pages={486--507},
  year={2021},
  publisher={Institute of Mathematical Statistics}
}

@article{JMLR:v22:20-753,
  author  = {Maxime Cauchois and Suyash Gupta and John C. Duchi},
  title   = {Knowing what You Know: valid and validated confidence sets in multiclass and multilabel prediction},
  journal = {Journal of Machine Learning Research},
  year    = {2021},
  volume  = {22},
  number  = {81},
  pages   = {1--42}
}

@article{stone1982optimal,
  title={Optimal global rates of convergence for nonparametric regression},
  author={Stone, Charles J},
  journal={The annals of statistics},
  pages={1040--1053},
  year={1982},
  publisher={JSTOR}
}

@inproceedings{niculescu2005predicting,
  title={Predicting good probabilities with supervised learning},
  author={Niculescu-Mizil, Alexandru and Caruana, Rich},
  booktitle={Proceedings of the 22nd international conference on Machine learning},
  pages={625--632},
  year={2005}
}

@inproceedings{guo2017calibration,
  title={On Calibration of Modern Neural Networks},
  author={Guo, Chuan and Pleiss, Geoff and Sun, Yu and Weinberger, Kilian Q},
  booktitle={International Conference on Machine Learning},
  pages={1321--1330},
  year={2017},
  organization={PMLR}
}

@inproceedings{zadrozny2002transforming,
  title={Transforming classifier scores into accurate multiclass probability estimates},
  author={Zadrozny, Bianca and Elkan, Charles},
  booktitle={Proceedings of the eighth ACM SIGKDD international conference on Knowledge discovery and data mining},
  pages={694--699},
  year={2002}
}

@inproceedings{platt1999probabilistic,
  title={Probabilistic outputs for support vector machines and comparisons to regularized likelihood methods},
  author={Platt, John},
  booktitle={Advances in large margin classifiers},
  volume={10},
  number={3},
  pages={61--74},
  year={1999}
}

@article{dawid1982well,
  title={The well-calibrated bayesian},
  author={Dawid, A Philip},
  journal={Journal of the American Statistical Association},
  volume={77},
  number={379},
  pages={605--610},
  year={1982}
}

@article{liero1989strong,
  author  = {Liero, Hannelore},
  title   = {Strong Uniform Consistency of Nonparametric Regression Function Estimates},
  journal = {Probability Theory and Related Fields},
  year    = {1989},
  volume  = {82},
  pages   = {587--614},
  doi     = {10.1007/BF00341285}
}

\clearpage
\appendix
\thispagestyle{empty}

\phantomsection\label{supplementary-material}
\bigskip

\addcontentsline{toc}{section}{Appendix}
\renewcommand{\thesection}{\Alph{section}}
\setcounter{section}{0}

\begin{center}
{\large\bf SUPPLEMENTARY MATERIAL for "Conformal Prediction Assessment: A Framework for Conditional Coverage Evaluation and Selection"}
\end{center}

\section*{Overview}

This document serves as the technical supplement to the main paper. It provides rigorous theoretical proofs, detailed algorithmic specifications, and comprehensive experimental configurations that support the findings presented in the paper. The material is organized as follows:

\begin{itemize}
    \setlength{\itemsep}{0.5em}
    
    \item \textbf{Theoretical Foundations (Appendix \ref{app:proofs}).} Contains the complete proofs for the main theorems, establishing the consistency of the CVI estimator and the asymptotic validity of the model selection framework.

    \item \textbf{Feasibility of Learning the Conditional Coverage Function (Appendix \ref{app:eta_feasibility}).} Explains why the conditional coverage surface induced by a fixed conformal predictor can be substantially simpler than the full conditional law in practically relevant regimes, and provides supporting synthetic and real-data evidence.

    \item \textbf{Algorithmic Specifications (Appendix \ref{app: Detailed Algorithms}).} Provides the formal pseudocode for the core components of our framework, specifically the Ensemble Reliability Estimator Training (\texttt{CPA-Train}) and the selection workflow (\texttt{CC-Select}).
    
    \item \textbf{Experimental Setup (Appendices \ref{app:dgp_details}--\ref{app:benchmark_details}).} Delineates the reproducibility details for the simulation studies. Appendix \ref{app:dgp_details} specifies the mathematical formulas for the four synthetic Data Generating Processes (DGPs). Appendix \ref{app:benchmark_details} details the implementation, hyperparameters, and base learners for all nine benchmarked uncertainty quantification methods.
    
    \item \textbf{Metric Definitions (Appendix \ref{app:ranking_metrics}).} Formally defines the ranking metrics used to evaluate selection fidelity, including Weighted Kendall's $\tau_w$, NDCG, and Hit@k.

    \item \textbf{Robustness and Failure Mechanisms (Appendices \ref{app:estimator_configs}--\ref{app:bias_analysis}).} Supports the robustness analysis in Section \ref{subsec:robustness}. Appendix \ref{app:estimator_configs} lists the exact configurations for the reliability estimator perturbations. Appendix \ref{app:bias_analysis} provides a deep-dive investigation into the mechanism of failure (Bias Analysis), elucidating why linear estimators fail in nonlinear regimes.
    
    \item \textbf{Extended Real-Data Analysis (Appendix \ref{appendix:real-data-details}).} Expands upon the real-world experiments. It includes the formal definition of the Worst Slab Coverage (WSC) metric and provides the full table of numerical results (average length, marginal coverage, and WSC) for all datasets.
\end{itemize}

\bigskip

\section{Proofs and Illustrations of the Main Theorems}\label{app:proofs}

\subsection{Illustration of the Limiting Conditional Coverage}
\label{app:illustration-true-coverage}

This subsection provides illustrative examples to complement Theorem~\ref{thm::true_coverage}. Figure~\ref{fig:con_cov_grid} visualizes the asymptotic conditional coverage function $\eta(x)$ under several representative forms of model misspecification, illustrating that uniform convergence of $\eta_n(x)$ does not generally imply convergence to the nominal target level $1-\alpha$.

\begin{figure}[h]
\centering
\resizebox{0.85\textwidth}{!}{
\begin{minipage}{\textwidth}
    \centering
    
    \begin{minipage}{0.48\textwidth}
        \centering
        \includegraphics[width=\textwidth]{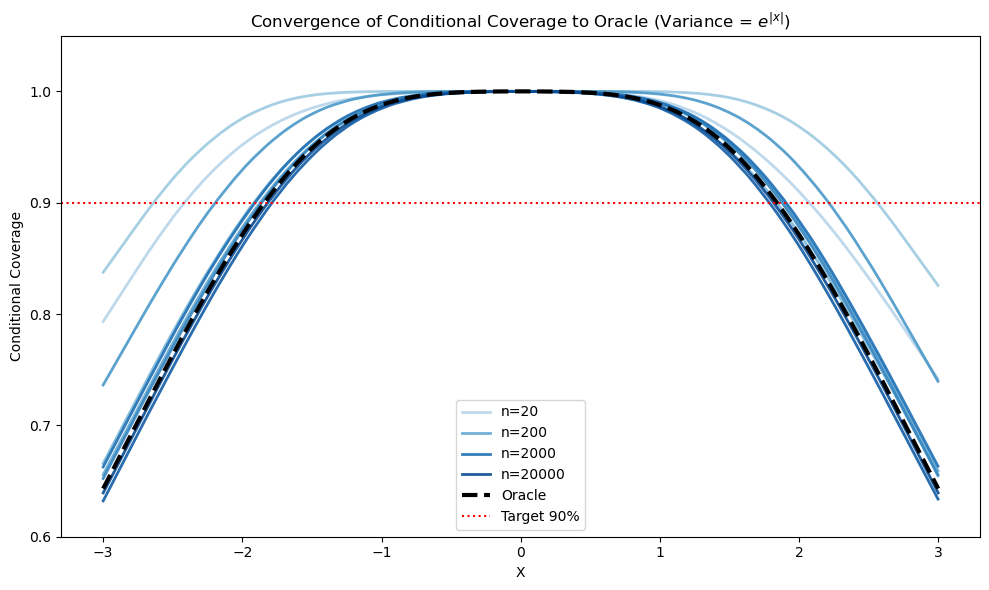}
        \subcaption{}
    \end{minipage}
    \hfill
    \begin{minipage}{0.48\textwidth}
        \centering
        \includegraphics[width=\textwidth]{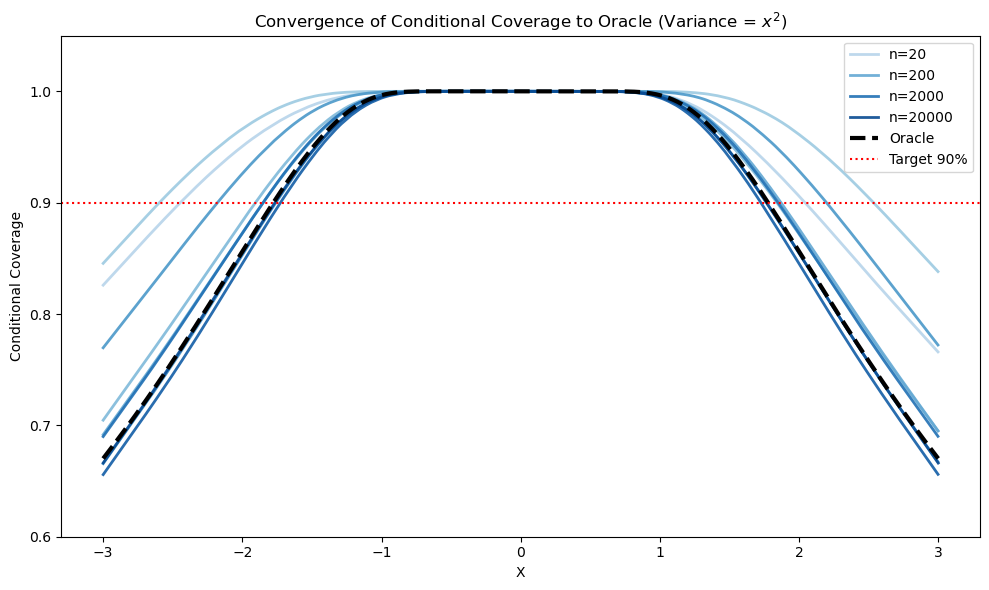}
        \subcaption{}
    \end{minipage}

    \vspace{0.8em}

    \begin{minipage}{0.48\textwidth}
        \centering
        \includegraphics[width=\textwidth]{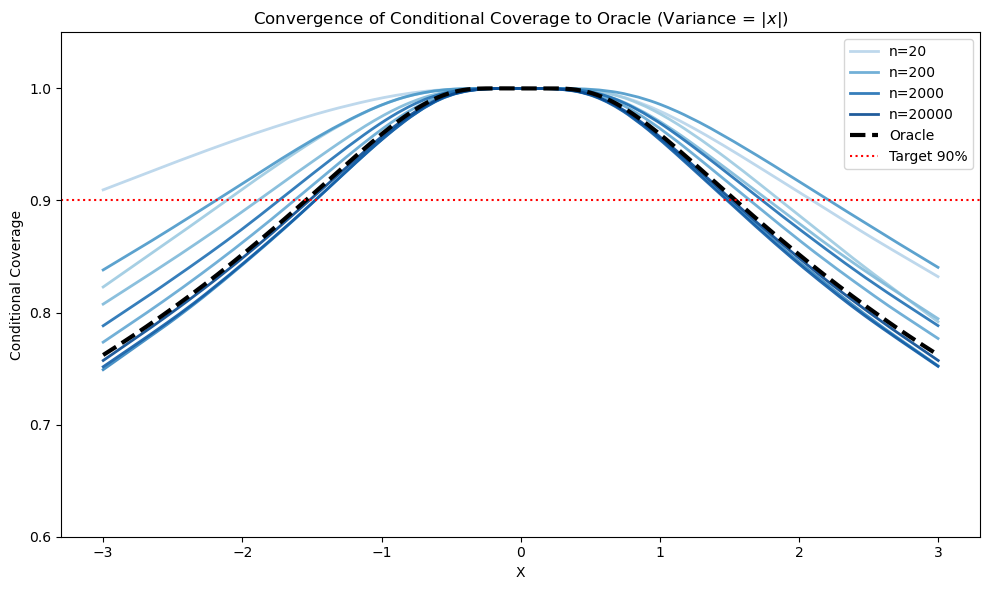}
        \subcaption{}
    \end{minipage}
    \hfill
    \begin{minipage}{0.48\textwidth}
        \centering
        \includegraphics[width=\textwidth]{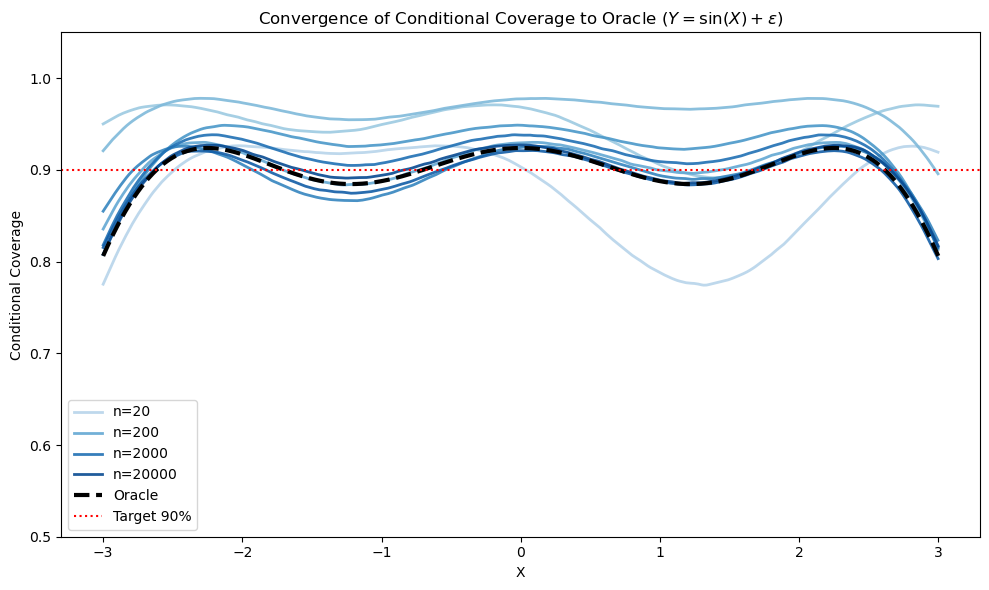}
        \subcaption{}
    \end{minipage}

\end{minipage}
}
\caption{Illustration of the limiting conditional coverage function $\eta(x)$. Panels (a)--(c) correspond to settings where the true regression function is linear with heteroscedastic noise, while panel (d) considers a nonlinear (sine) regression function with homoscedastic noise. In all cases, the fitted regression model is linear, demonstrating that $\eta(x)$ may deviate substantially from the nominal level $1-\alpha$ even asymptotically.}
\label{fig:con_cov_grid}
\end{figure}

\subsection{Proof of Theorem 1}

\begin{proof}
    With a slight abuse of notation, in the proof of Theorem~\ref{thm::true_coverage}, let $(X ,Y)=(X_{n+1},Y_{n+1})$. Let $\epsilon$ denote the residual $Y-\E[Y\mid X]$.
    
    Define the event $A_{n_{\textup{pred}},n_{\textup{calib}}}=\{|\mu(X)+\epsilon-\hat{\mu}_{n_{\textup{pred}}}(X)|\leq \hat{q}_{n_{\textup{calib}},1-\alpha}\}$, $B_{n_{\textup{pred}}}=\{|\mu(X)+\epsilon-\hat{\mu}_{n_{\textup{pred}}}(X)|\leq q_{1-\alpha}\}$. Here $q_{1-\alpha}$ denotes the\footnote{Under assumption~\ref{assumption:quantile_density}, $|\mu(X)+\epsilon-\tilde{\mu}(X)|$ admits a unique upper $\alpha$ quantile.} upper $\alpha$ quantile of $|\mu(X)+\epsilon-\tilde{\mu}(X)|$. And let $B = \left\{|\mu(X)+\epsilon-\tilde{\mu}(X)|\leq q_{1-\alpha}\right\}$.
    
    From the triangle inequality, we have
    \begin{align*}
        &| \sP(A_{n_{\textup{pred}},n_{\textup{calib}}} | X = x, \hat{\mu}_{n_{\textup{pred}}}, \hat{q}_{n_{\textup{calib}},1-\alpha}) - \sP(B | X = x) | \\
        &\leq | \sP(A_{n_{\textup{pred}},n_{\textup{calib}}} | X = x, \hat{\mu}_{n_{\textup{pred}}}, \hat{q}_{n_{\textup{calib}},1-\alpha}) - \sP(B_{n_{\textup{pred}}} | X = x , \hat{\mu}_{n_{\textup{pred}}}) | \\
        &+ | \sP(B_{n_{\textup{pred}}} | X = x , \hat{\mu}_{n_{\textup{pred}}}) - \sP(B | X = x) | \\
        &\leq | \sP(A_{n_{\textup{pred}},n_{\textup{calib}}} \triangle B_{n_{\textup{pred}}} | X = x, \hat{\mu}_{n_{\textup{pred}}}, \hat{q}_{n_{\textup{calib}},1-\alpha}) | + | \sP(B_{n_{\textup{pred}}} \triangle B | X = x , \hat{\mu}_{n_{\textup{pred}}}) |.
    \end{align*}
    
    So we have 
    \begin{align*}
        &\mathbb{E} \| \sP(A_{n_{\textup{pred}},n_{\textup{calib}}} | X = x, \hat{\mu}_{n_{\textup{pred}}}, \hat{q}_{n_{\textup{calib}},1-\alpha}) - \sP(B | X = x) \|_{\infty} \\
        &\leq \mathbb{E} \| \sP(A_{n_{\textup{pred}},n_{\textup{calib}}} \triangle B_{n_{\textup{pred}}} | X = x, \hat{\mu}_{n_{\textup{pred}}}, \hat{q}_{n_{\textup{calib}},1-\alpha}) \|_{\infty} + \mathbb{E} \| \sP(B_{n_{\textup{pred}}} \triangle B | X = x , \hat{\mu}_{n_{\textup{pred}}}) \|_{\infty}
    \end{align*}
    
    We first analyze the behavior of $\sP(B_{n_{\textup{pred}}} \triangle B | X = x , \hat{\mu}_{n_{\textup{pred}}})$. Denote by $\hat{f}_{n_{\textup{pred}}}(X),\tilde{f}(X)$ the functions $|\mu(X)+\epsilon-\hat{\mu}_{n_{\textup{pred}}}(X)|$, $|\mu(X)+\epsilon-\tilde{\mu}(X)|$. The randomness of $\hat{f}_{n_{\textup{pred}}}(X)$ comes from $\epsilon$, $\hat{\mu}_{n_{\textup{pred}}}$ and $X$. The randomness of $\tilde{f}(X)$ comes from $\epsilon$ and $X$. Define $E_{n_{\textup{pred}}}=\{||\hat{\mu}_{n_{\textup{pred}}}-\tilde{\mu}||_{\infty}<\zeta_{n_{\textup{pred}}}\}$. Then
    \begin{align*}
        \sP(B \setminus B_{n_{\textup{pred}}}, E_{n_{\textup{pred}}} | X = x , \hat{\mu}_{n_{\textup{pred}}}) &\leq \sP(\{|\hat{f}_{n_{\textup{pred}}}(X)\in (q_{1-\alpha},q_{1-\alpha}+\zeta_{n_{\textup{pred}}}]\}, E_{n_{\textup{pred}}}|X=x, \hat{\mu}_{n_{\textup{pred}}}) \\
        &\leq \sP(\{|\tilde{f}(X)\in (q_{1-\alpha}-\zeta_{n_{\textup{pred}}},q_{1-\alpha}+2\zeta_{n_{\textup{pred}}}]\}|X=x) \\
        &=: p_{n_{\textup{pred}}}(x)
    \end{align*}
    and
    \begin{align*}
        \sP(B_{n_{\textup{pred}}} \setminus B, E_{n_{\textup{pred}}} | X = x , \hat{\mu}_{n_{\textup{pred}}}) &\leq \sP(\{|\tilde{f}(X)\in (q_{1-\alpha},q_{1-\alpha}+\zeta_{n_{\textup{pred}}}]\}, E_{n_{\textup{pred}}}|X=x) \\
        &=\sP(\{|\tilde{f}(X)\in (q_{1-\alpha},q_{1-\alpha}+\zeta_{n_{\textup{pred}}}]\}|X=x)\leq p_{n_{\textup{pred}}}(x).
    \end{align*}
    
    To conclude, we apply Bayes' rule:
    \begin{align*}
        &\| \sP(B_{n_{\textup{pred}}} \triangle B | X = x , \hat{\mu}_{n_{\textup{pred}}}) \|_{\infty} \\
        &= \|\sP(B_{n_{\textup{pred}}} \triangle B , E_{n_{\textup{pred}}}| X = x , \hat{\mu}_{n_{\textup{pred}}})+ \sP(B_{n_{\textup{pred}}} \triangle B , E_{n_{\textup{pred}}}^c| X = x , \hat{\mu}_{n_{\textup{pred}}})\|_{\infty}\\
        &\leq \|2p_{n_{\textup{pred}}}(x)\|_{\infty}+1(E_{n_{\textup{pred}}}^c).
    \end{align*}
    
    Taking expectation with respect to $\hat{\mu}_{n_{\textup{pred}}}$, we have 
    $$
    \mathbb{E} \| \sP(B_{n_{\textup{pred}}} \triangle B | X = x , \hat{\mu}_{n_{\textup{pred}}}) \|_{\infty} \leq 2\|p_{n_{\textup{pred}}}(x)\|_{\infty} + \rho_{n_{\textup{pred}}}.
    $$
    
    For the behavior of $\sP(A_{n_{\textup{pred}},n_{\textup{calib}}} \triangle B_{n_{\textup{pred}}} | X = x, \hat{\mu}_{n_{\textup{pred}}}, \hat{q}_{n_{\textup{calib}},1-\alpha})$, we have
    \allowdisplaybreaks
    \begin{align*}
        &\sP(A_{n_{\textup{pred}},n_{\textup{calib}}} \triangle B_{n_{\textup{pred}}} | X = x, \hat{\mu}_{n_{\textup{pred}}}, \hat{q}_{n_{\textup{calib}},1-\alpha}) \\
        &= \sP \left(q_{1- \alpha} < \hat{f}_{n_{\textup{pred}}}(X) < \hat{q}_{n_{\textup{calib}},1-\alpha} \mid X=x, \hat{\mu}_{n_{\textup{pred}}}, \hat{q}_{n_{\textup{calib}},1-\alpha} \right) \\
        &+ \sP \left(\hat{q}_{n_{\textup{calib}},1-\alpha} < \hat{f}_{n_{\textup{pred}}}(X) < q_{1- \alpha} \mid X=x, \hat{\mu}_{n_{\textup{pred}}}, \hat{q}_{n_{\textup{calib}},1-\alpha} \right) \\
        &\leq \sP[q_{1-\alpha}<|\hat{f}_{n_{\textup{pred}}}(X)|\leq \hat{q}_{n_{\textup{calib}},1-\alpha},E_{n_{\textup{pred}}} \mid X=x, \hat{\mu}_{n_{\textup{pred}}}, \hat{q}_{n_{\textup{calib}},1-\alpha} ]+ 1[E_{n_{\textup{pred}}}^c] \\
        &+\sP[ \hat{q}_{n_{\textup{calib}},1-\alpha}<|\hat{f}_{n_{\textup{pred}}}(X)|\leq q_{1-\alpha}, E_{n_{\textup{pred}}} \mid X=x, \hat{\mu}_{n_{\textup{pred}}}, \hat{q}_{n_{\textup{calib}},1-\alpha} ]+ 1[E_{n_{\textup{pred}}}^c] \\
        &\leq 2\cdot 1[E_{n_{\textup{pred}}}^c] +\sP[q_{1-\alpha}-\zeta_{n_{\textup{pred}}}<|\tilde{f}(X)|\leq \hat{q}_{n_{\textup{calib}},1-\alpha}+\zeta_{n_{\textup{pred}}}\mid X=x, \hat{q}_{n_{\textup{calib}},1-\alpha}] \\
        &+\sP[ \hat{q}_{n_{\textup{calib}},1-\alpha}-\zeta_{n_{\textup{pred}}}<|\tilde{f}(X)|\leq q_{1-\alpha}+\zeta_{n_{\textup{pred}}} \mid X=x, \hat{q}_{n_{\textup{calib}},1-\alpha}]\\
        &\leq  2\cdot 1[E_{n_{\textup{pred}}}^c]\\
        &+1[\hat{q}_{n_{\textup{calib}},1-\alpha}-q_{1-\alpha}>2\zeta_{n_{\textup{pred}}}]\\
        &+\sP[q_{1-\alpha}-\zeta_{n_{\textup{pred}}}<|\tilde{f}(X)|\leq \hat{q}_{n_{\textup{calib}},1-\alpha}+\zeta_{n_{\textup{pred}}}, \hat{q}_{n_{\textup{calib}},1-\alpha}-q_{1-\alpha}\leq 2\zeta_{n_{\textup{pred}}}\mid X = x, \hat{q}_{n_{\textup{calib}},1-\alpha}]\\
        &+1[q_{1-\alpha}-\hat{q}_{n_{\textup{calib}},1-\alpha}>2\zeta_{n_{\textup{pred}}}]\\
        &+\sP[ \hat{q}_{n_{\textup{calib}},1-\alpha}-\zeta_{n_{\textup{pred}}}<|\tilde{f}(X)|\leq q_{1-\alpha}+\zeta_{n_{\textup{pred}}},q_{1-\alpha}-\hat{q}_{n_{\textup{calib}},1-\alpha}\leq 2\zeta_{n_{\textup{pred}}}\mid X = x, \hat{q}_{n_{\textup{calib}},1-\alpha}]\\
        &\leq  2\cdot 1[E_{n_{\textup{pred}}}^c \mid \hat{\mu}_{n_{\textup{pred}}}]+1[|q_{1-\alpha}-\hat{q}_{n_{\textup{calib}},1-\alpha}|>2\zeta_{n_{\textup{pred}}}]\\
        &+\sP[q_{1-\alpha}-\zeta_{n_{\textup{pred}}}<|\tilde{f}(X)|\leq q_{1-\alpha}+ 3\zeta_{n_{\textup{pred}}}\mid X=x]\\
        &+\sP[ q_{1-\alpha} - 3\zeta_{n_{\textup{pred}}}<|\tilde{f}(X)|\leq q_{1-\alpha}+\zeta_{n_{\textup{pred}}}\mid X=x]\\
        &\leq 2\cdot 1[E_{n_{\textup{pred}}}^c] +1[|q_{1-\alpha}-\hat{q}_{n_{\textup{calib}},1-\alpha}|>2\zeta_{n_{\textup{pred}}}]+2\sP[|q_{1-\alpha}-|\tilde{f}(X)||\leq 3\zeta_{n_{\textup{pred}}}\mid X=x].
    \end{align*}
    
    Taking supremum over $x$ and expectation with respect to $\hat{\mu}_{n_{\textup{pred}}},\hat{q}_{n_{\textup{calib}},1-\alpha}$, we have
    \begin{align*}
        &\mathbb{E} \| \sP(A_{n_{\textup{pred}},n_{\textup{calib}}} \triangle B_{n_{\textup{pred}}} | X = x, \hat{\mu}_{n_{\textup{pred}}}, \hat{q}_{n_{\textup{calib}},1-\alpha}) \|_{\infty} \\
        &\leq 2\rho_{n_{\textup{pred}}} + \sP[|q_{1-\alpha}-\hat{q}_{n_{\textup{calib}},1-\alpha}|>2\zeta_{n_{\textup{pred}}}] + 2\| \sP[|q_{1-\alpha}-|\tilde{f}(X)||\leq 3\zeta_{n_{\textup{pred}}}\mid X=x] \|_{\infty}.
    \end{align*}
    
    Summarizing all the above, we have
    \begin{align*}
        &\mathbb{E} \| \sP(A_{n_{\textup{pred}},n_{\textup{calib}}} | X = x, \hat{\mu}_{n_{\textup{pred}}}, \hat{q}_{n_{\textup{calib}},1-\alpha}) - \sP(B | X = x) \|_{\infty} \\  
        &\leq 3\rho_{n_{\textup{pred}}} + \sP[|q_{1-\alpha}-\hat{q}_{n_{\textup{calib}},1-\alpha}|>2\zeta_{n_{\textup{pred}}}] \\
        &+ 2\| \sP[|q_{1-\alpha}-|\tilde{f}(X)||\leq 3\zeta_{n_{\textup{pred}}}\mid X=x] \|_{\infty} + 2\|p_{n_{\textup{pred}}}(x)\|_{\infty}.
    \end{align*}
    
    Under the assumption~\ref{assumption:conditional_density}, $||p_{n_{\textup{pred}}}||_{\infty}\leq 3M\zeta_{n_{\textup{pred}}}$ and 
    \[
        \|\sP[|q_{1-\alpha}-|\tilde{f}(X)||\leq 3\zeta_{n_{\textup{pred}}}\mid X=x] \|_{\infty}\leq 6M\zeta_{n_{\textup{pred}}}.
    \]
    Define $\tilde{q}_{n_{\textup{calib}},1-\alpha}$ to be the sample $(1-\alpha)$ quantile of $|Y_i-\tilde{\mu}(X_i)|_{i\in\gI_{\textup{calib}}}$. Then from $\sP[|\tilde{q}_{n_{\textup{calib}},1-\alpha}-\hat{q}_{n_{\textup{calib}},1-\alpha}|>\zeta_{n_{\textup{pred}}}]\leq \sP[|\tilde{\mu}-\hat{\mu}_{n_{\textup{pred}}}|>\zeta_{n_{\textup{pred}}}]\leq \rho_{n_{\textup{pred}}}$, we derive that 
    $$
    \begin{aligned}
        &\mathbb{E} \| \sP(A_{n_{\textup{pred}},n_{\textup{calib}}} | X = x, \hat{\mu}_{n_{\textup{pred}}}, \hat{q}_{n_{\textup{calib}},1-\alpha}) - \sP(B | X = x) \|_{\infty}\\
        &\leq 18M\zeta_{n_{\textup{pred}}}+4\rho_{n_{\textup{pred}}}+\sP[|q_{1-\alpha}-\tilde{q}_{n_{\textup{calib}},1-\alpha}|>\zeta_{n_{\textup{pred}}}].
    \end{aligned}
    $$
    
    Using the standard concentration inequality for sample quantiles, the tail probability $\sP[|q_{1-\alpha}-\tilde{q}_{n_{\textup{calib}},1-\alpha}|>\zeta_{n_{\textup{pred}}}]$ has the upper bound $2\exp(-2n_{\textup{calib}}\epsilon_{n_{\textup{calib}},n_{\textup{pred}},\alpha}^2)$, where $F$ denotes the cumulative distribution function of $|\mu(X)+\epsilon-\tilde{\mu}(X)|$ and 
    \[
        \epsilon_{n_{\textup{calib}},n_{\textup{pred}},\alpha}=\min\left(F(q_{1-\alpha}+\zeta_{n_{\textup{pred}}})-(1-\alpha),(1-\alpha)-F(q_{1-\alpha}-\zeta_{n_{\textup{pred}}})\right).
    \]
    From assumption~\ref{assumption:quantile_density}, $\epsilon_{n_{\textup{calib}},n_{\textup{pred}},\alpha} \geq   M_0\zeta_{n_{\textup{pred}}}$. Therefore, we have 
    $$
    \begin{aligned}
        \gE&:=\mathbb{E} \| \sP(A_{n_{\textup{pred}},n_{\textup{calib}}} | X = x, \hat{\mu}_{n_{\textup{pred}}}, \hat{q}_{n_{\textup{calib}},1-\alpha}) - \sP(B | X = x) \|_{\infty} \\
        &\leq 18M\zeta_{n_{\textup{pred}}} + 4\rho_{n_{\textup{pred}}} + \exp(-2M_0^2n_{\textup{calib}}\zeta_{n_{\textup{pred}}}^2).
    \end{aligned}
    $$ 
    Combining this bound with the following observation finishes the proof:
    \[
    \gE \geq\| \mathbb{E}\left[\sP(A_{n_{\textup{pred}},n_{\textup{calib}}} | X = x, \hat{\mu}_{n_{\textup{pred}}}, \hat{q}_{n_{\textup{calib}},1-\alpha})\right] - \sP(B | X = x) \|_{\infty}=\| \eta_{n_{\textup{pred}}}(x) - \eta(x) \|_{\infty}.
    \]
\end{proof}

\subsection{Proof of Theorem 2}

\begin{proof}
    By the triangle inequality, we have
    \[
    \| \hat{\eta}_{n_{\text{eval}}} - \eta_{n} \|_{\infty} 
    \leq 
    \underbrace{\| \hat{\eta}_{n_{\text{eval}}} - \eta_{n_{\text{train}}} \|_{\infty}}_{\text{Estimation Error}} 
    + 
    \underbrace{\| \eta_{n_{\text{train}}} - \eta_{n} \|_{\infty}}_{\text{Approximation Error}}.
    \]
    For the estimation error term, Assumption \ref{assumption:classifier_convergence} implies that 
    \[
    \| \hat{\eta}_{n_{\text{eval}}} - \eta_{n_{\text{train}}} \|_{\infty} = O_p(\psi_{n_{\text{eval}}}).
    \]
    Regarding the approximation error, we further upper bound it by
    \[
    \| \eta_{n_{\text{train}}} - \eta_n \|_{\infty} 
    \leq 
    \| \eta_{n_{\text{train}}} - \eta \|_{\infty} 
    + 
    \| \eta_{n} - \eta \|_{\infty}.
    \]
    The two terms on the right-hand side are of order 
    $O_p(\varphi_{n_{\text{train}}})$ and $O_p(\varphi_{n})$, respectively.  
    Therefore,
    \[
    \| \eta_{n_{\text{train}}} - \eta_{n} \|_{\infty} 
    = O_p(\varphi_{n_{\text{train}}} + \varphi_{n}).
    \]
    Combining the two parts completes the proof.
\end{proof}

\subsection{Proof of Theorem 3}

We first prove the following lemma:
\begin{lemma}   
\label{lm::CLT_surrogate_CVI}

    Under Assumption \ref{assumption:stability_of_conditional_coverage}, one has 
    \[
        \frac{1}{n_{\textup{eval}}}\sum_{i \in \mathcal{I}_{{\textup{eval}}}}|\eta_{n_{\textup{train}}}(X_i)-(1-\alpha)|-\E_{i \in \mathcal{I}_{{\textup{eval}}}}\left[|\eta_{n_{\textup{train}}}(X_i)-(1-\alpha)|\mid \gD_{n_{\textup{train}}}\right]=O_p\left(n_{\textup{eval}}^{-1/2}\right).
    \]
\end{lemma}
\begin{proof}[Proof of lemma~\ref{lm::CLT_surrogate_CVI}]
    Given any $i \in \mathcal{I}_{{\textup{eval}}}$, define the random variables:
    \[
        D_i=|\eta_{n_{\textup{train}}}(X_i)-(1-\alpha)|,\;\;m=\E[D_i\mid \gD_{n_{\textup{train}}}].
    \]
    They take values in $[0,1]$.
    
    Conditioning on $\gD_{n_{\textup{train}}}$, $D_i(i\in  \mathcal{I}_{\text{eval}})$ are i.i.d. Thus 
    \[
        \E\left[\left(\sum_{i \in \mathcal{I}_{{\textup{eval}}}}(D_i-m)\right)^2\mid \gD_{n_{\textup{train}}}\right]=\sum_{i \in \mathcal{I}_{{\textup{eval}}}}\E\left[(D_i-m)^2\mid \gD_{n_{\textup{train}}}\right]\leq n_{\textup{eval}}.
    \]
    Taking expectation again gives us 
    \[
        \E\left[\left(\sum_{i \in \mathcal{I}_{{\textup{eval}}}}(D_i-m)\right)^2\right]\leq n_{\textup{eval}}.
    \]
    The desired conclusion thus follows.
\end{proof}

\begin{proof}[Proof of Theorem~\ref{thm::consistency}]
    It is clear that the following inequalities hold:
    \[
        \left|\textup{CVI}_{n_{\textup{eval}}}-\frac{1}{n_{\textup{eval}}} \sum_{i \in \mathcal{I}_{{\textup{eval}}}} | \eta_{n_{\textup{train}}}(X_i) - (1 - \alpha) |\right|\leq ||\eta_{n_{\textup{train}}}-\hat{\eta}_{n_{\textup{eval}}}||_{\infty},
    \]
    \[
        \left|\text{CVI}_{\textup{oracle},{n}}-\mathbb{E}_{X \sim P_X} \left[| \eta(X) - \left(1-\alpha\right) |\right] \right|\leq ||\eta-\eta_{n}||_{\infty},
    \]
    \[
        \left|\mathbb{E} \left[| \eta_{n_{\textup{train}}}(X) - \left(1-\alpha\right) |\mid \gD_{n_{\textup{train}}}\right]-\mathbb{E}_{X \sim P_X} \left[| \eta(X) - \left(1-\alpha\right) |\right] \right|\leq ||\eta_{n_{\textup{train}}}-\eta||_{\infty}.
    \]
    Summing these inequalities, applying Lemma~\ref{lm::CLT_surrogate_CVI}, and using the triangle inequality, we obtain
    \[
        \begin{aligned}
            &\left|\text{CVI}_{\textup{oracle},{n}}-\textup{CVI}_{n_{\textup{eval}}} \right|\\
            &\leq \left| \frac{1}{n_{\textup{eval}}}\sum_{i \in \mathcal{I}_{{\textup{eval}}}}|\eta_{n_{\textup{train}}}(X_i)-(1-\alpha)|-\E\left[|\eta_{n_{\textup{train}}}(X)-(1-\alpha)|\mid \gD_{n_{\textup{train}}}\right] \right|\\
            &+||\eta_{n_{\textup{train}}}-\hat{\eta}_{n_{\textup{eval}}}||_{\infty}+||\eta-\eta_{n}||_{\infty}+||\eta_{n_{\textup{train}}}-\eta||_{\infty}\\
            &\leq O_p\left(n_{\textup{eval}}^{-1/2}\right)+||\eta_{n_{\textup{train}}}-\hat{\eta}_{n_{\textup{eval}}}||_{\infty}+||\eta-\eta_{n}||_{\infty}+||\eta_{n_{\textup{train}}}-\eta||_{\infty}=o_p(1).
        \end{aligned}
    \]
    The last equality is from assumptions \ref{assumption:classifier_convergence} and \ref{assumption:stability_of_conditional_coverage}.
\end{proof}

\subsection{Proof of Theorem 4}

\begin{proof}
    \textbf{Case 1.} Since $\textup{CVI}_{\textup{oracle}}^{(1)} < \textup{CVI}_{\textup{oracle}}^{(2)}$, there exists a constant $c_0 > 0$ such that $\textup{CVI}_{\textup{oracle}}^{(2)} \geq (1 + c_0) \textup{CVI}_{\textup{oracle}}^{(1)}$. By Theorem~\ref{thm::consistency}, we have 
    $$
    \text{CVI}_{n_\textup{eval}}^{(k)} = \text{CVI}_{\textup{oracle},n_{\textup{total}}}^{(k)} + o_p(1) \quad \text{for } k=1,2.
    $$
    Therefore, for any $\epsilon > 0$, there exists a constant $N > 0$ such that for all $n_{\textup{train}}, n_{\textup{eval}} > N$,
    $$
    \mathbb{P}\left(\text{CVI}_{n_\textup{eval}}^{(2)} \geq (1 + c_0/2) \text{CVI}_{n_\textup{eval}}^{(1)}\right) \geq 1 - \epsilon.
    $$
    This implies that $\mathbb{P}(\hat{k} = 1) \to 1$ as $n_{\textup{train}}, n_{\textup{eval}} \to \infty$.

    \textbf{Case 2.} Choosing the better CP algorithm is equivalent to the following event:
    {\footnotesize
    \begin{align*}
        \sP(\hat{k} = 1) &= \sP(\textup{CVI}_{n_\textup{eval}}^{(1)} \leq \textup{CVI}_{n_\textup{eval}}^{(2)}) \\
        &= \sP\Biggl( \frac{1}{n_\textup{eval}} \sum_{i \in \mathcal{I}_\text{eval}} |\hat{\eta}_{n_\textup{eval}}^{(1)}(X_i) - (1 - \alpha)| \leq \frac{1}{n_\textup{eval}} \sum_{i \in \mathcal{I}_\text{eval}} |\hat{\eta}_{n_\textup{eval}}^{(2)}(X_i) - (1 - \alpha)| \Biggr) \\
        &\geq \sP\Biggl( \frac{1}{n_\textup{eval}} \sum_{i \in \mathcal{I}_\text{eval}} |\eta_{n_\text{train}}^{(1)}(X_i) - (1 - \alpha)| + \| \hat{\eta}_{n_\textup{eval}}^{(1)} - \eta_{n_\text{train}}^{(1)} \|_{\infty} \\
        &\qquad \leq \frac{1}{n_\textup{eval}} \sum_{i \in \mathcal{I}_\text{eval}} |\eta_{n_\text{train}}^{(2)}(X_i) - (1 - \alpha)| -  \| \hat{\eta}_{n_\textup{eval}}^{(2)} - \eta_{n_\text{train}}^{(2)} \|_{\infty} \Biggr) \\
        &= \sP\Biggl( \| \hat{\eta}_{n_\textup{eval}}^{(1)} - \eta_{n_\text{train}}^{(1)} \|_{\infty} +  \| \hat{\eta}_{n_\textup{eval}}^{(2)} - \eta_{n_\text{train}}^{(2)} \|_{\infty} \\
        &\qquad \leq \frac{1}{n_\textup{eval}} \sum_{i \in \mathcal{I}_\text{eval}} |\eta_{n_\text{train}}^{(2)}(X_i) - (1 - \alpha)| -  \frac{1}{n_\textup{eval}} \sum_{i \in \mathcal{I}_\text{eval}} |\eta_{n_\text{train}}^{(1)}(X_i) - (1 - \alpha)|  \Biggr)
    \end{align*}
    }
    For the right-hand side of the inequality in the probability and $k=1,2$, we have 
    {\footnotesize
    \begin{align*}
        \frac{1}{n_\textup{eval}} \sum_{i \in \mathcal{I}_\textup{eval}} |\eta_{n_\text{train}}^{(k)}(X_i) - (1 - \alpha)| 
        &=: \mathbb{E}_{X \sim P_X} [|\eta_{n_\text{train}}^{(k)}(X) - (1 - \alpha)| \mid \mathcal{D}_{n_\text{train}}] + \Delta_{n_\text{eval}}^{(k)} \\
        &= \text{CVI}_{\text{oracle},n_\text{train}}^{(k)} + \Delta_{n_\text{eval}}^{(k)},
    \end{align*}    
    }
    where $\Delta_{n_\text{eval}}^{(k)} = O_p(n_{\textup{eval}}^{-1/2})$ from Lemma \ref{lm::CLT_surrogate_CVI}. Therefore, 
    {\footnotesize
    \begin{align*}
        &\sP\Biggl( \| \hat{\eta}_{n_\textup{eval}}^{(1)} - \eta_{n_\text{train}}^{(1)} \|_{\infty} +  \| \hat{\eta}_{n_\textup{eval}}^{(2)} - \eta_{n_\text{train}}^{(2)} \|_{\infty} \\
        &\qquad \leq \frac{1}{n_\textup{eval}} \sum_{i \in \mathcal{I}_\text{eval}} |\eta_{n_\text{train}}^{(2)}(X_i) - (1 - \alpha)| -  \frac{1}{n_\textup{eval}} \sum_{i \in \mathcal{I}_\text{eval}} |\eta_{n_\text{train}}^{(1)}(X_i) - (1 - \alpha)|  \Biggr) \\
        &= \sP\Biggl( \| \hat{\eta}_{n_\textup{eval}}^{(1)} - \eta_{n_\text{train}}^{(1)} \|_{\infty} +  \| \hat{\eta}_{n_\textup{eval}}^{(2)} - \eta_{n_\text{train}}^{(2)} \|_{\infty} \\
        &\qquad \leq \text{CVI}_{\text{oracle},n_\text{train}}^{(2)} - \text{CVI}_{\text{oracle},n_\text{train}}^{(1)} + \Delta_{n_\text{eval}}^{(2)} - \Delta_{n_\text{eval}}^{(1)}  \Biggr) \\
        &\geq \sP\Biggl( \| \hat{\eta}_{n_\text{eval}}^{(1)} - \eta_{n_\text{train}}^{(1)} \|_{\infty} + \| \hat{\eta}_{n_\text{eval}}^{(2)} - \eta_{n_\text{train}}^{(2)} \|_{\infty} \\
        &\qquad + \textup{CVI}_{\text{oracle},n_\text{train}}^{(1)} + |\Delta_{n_\text{eval}}^{(1)}| + |\Delta_{n_\text{eval}}^{(2)}| \leq \textup{CVI}_{\text{oracle},n_\text{train}}^{(2)}   \Biggr)
    \end{align*}
    }

  By the definition of asymptotically better, for any $\epsilon > 0$, there exists $N > 0$ and $c_{\epsilon} > 0$ such that when $n_{\text{train}} > N$, we have $ \sP\left(\textup{CVI}_{\text{oracle},n_{\text{train}}}^{(2)} \geq (1 + c_\epsilon) \textup{CVI}_{\text{oracle},n_{\text{train}}}^{(1)}\right) \geq 1- \epsilon.$ Then we have 
{\footnotesize
\begin{align*}
    &\sP\Biggl( \| \hat{\eta}_{n_{\text{eval}}}^{(1)} - \eta_{n_{\text{train}}}^{(1)} \|_{\infty} + \| \hat{\eta}_{n_{\text{eval}}}^{(2)} - \eta_{n_{\text{train}}}^{(2)} \|_{\infty} \\
    &\qquad + \textup{CVI}_{\text{oracle},n_{\text{train}}}^{(1)} + |\Delta_{n_{\text{eval}}}^{(1)}| + |\Delta_{n_{\text{eval}}}^{(2)}| \leq \textup{CVI}_{\text{oracle},n_{\text{train}}}^{(2)} \Biggr) \\
    &\geq 1 - \sP\Biggl( \frac{c_\epsilon}{1 + c_\epsilon} \textup{CVI}_{\text{oracle},n_{\text{train}}}^{(2)} \leq \| \hat{\eta}_{n_{\text{eval}}}^{(1)} - \eta_{n_{\text{train}}}^{(1)} \|_{\infty} \\
    &\qquad\qquad + \| \hat{\eta}_{n_{\text{eval}}}^{(2)} - \eta_{n_{\text{train}}}^{(2)} \|_{\infty} + |\Delta_{n_{\text{eval}}}^{(1)}| + |\Delta_{n_{\text{eval}}}^{(2)}| \Biggr) \\
    &\quad - \sP\Biggl( \frac{\textup{CVI}_{\text{oracle},n_{\text{train}}}^{(2)} }{1 + c_\epsilon} \leq \textup{CVI}_{\text{oracle},n_{\text{train}}}^{(1)} \Biggr) \\
    &\geq 1 - \epsilon - \sP\Biggl( \frac{c_\epsilon}{1 + c_\epsilon} \textup{CVI}_{\text{oracle},n_{\text{train}}}^{(2)} \leq \| \hat{\eta}_{n_{\text{eval}}}^{(1)} - \eta_{n_{\text{train}}}^{(1)} \|_{\infty} \\
    &\qquad\qquad + \| \hat{\eta}_{n_{\text{eval}}}^{(2)} - \eta_{n_{\text{train}}}^{(2)} \|_{\infty} + |\Delta_{n_{\text{eval}}}^{(1)}| + |\Delta_{n_{\text{eval}}}^{(2)}| \Biggr)
\end{align*}
}
Under the condition $\frac{n_{\text{eval}}^{-1/2} \lor \| \hat{\eta}_{n_{\text{eval}}}^{(k)} - \eta_{n_{\text{train}}}^{(k)} \|_{\infty}}{\textup{CVI}_{\textup{oracle},n_{\text{train}}}^{(2)}} = o_p(1)$ for $k=1,2$, we have
\begin{align*}
    \lim_{n \to \infty} \sP\Biggl( \frac{c_\epsilon}{1 + c_\epsilon} \textup{CVI}_{\text{oracle},n_{\text{train}}}^{(2)} 
    &\leq \| \hat{\eta}_{n_{\text{eval}}}^{(1)} - \eta_{n_{\text{train}}}^{(1)} \|_{\infty} + \| \hat{\eta}_{n_{\text{eval}}}^{(2)} - \eta_{n_{\text{train}}}^{(2)} \|_{\infty} \\
    &\quad + |\Delta_{n_{\text{eval}}}^{(1)}| + |\Delta_{n_{\text{eval}}}^{(2)}| \Biggr) = 0.
\end{align*}
This is because the right-hand side of the inequality in the probability is $o_p( \textup{CVI}_{\text{oracle},n_{\text{train}}}^{(2)} )$ and both sides are non-negative. Therefore, summarizing the above results, we have for any $\epsilon > 0$, 
\begin{align*}
    \lim_{n \to \infty} \sP(\textup{CVI}_{n_{\text{eval}}}^{(1)} \leq \textup{CVI}_{n_{\text{eval}}}^{(2)}) \geq 1 - \epsilon.
\end{align*}
So we have the consistency of the model selection based on CVI:
$$
\sP(\hat{k} = 1) \to 1.
$$

\end{proof}

\section{Feasibility of Learning the Conditional Coverage Function}
\label{app:eta_feasibility}

The proposed CPA crucially relies on the estimation of the conditional coverage probability. In full generality, estimating the conditional coverage function of an arbitrary CP method need not be easier than learning the full conditional law of $Y\mid X=x$. For a generic interval-valued predictor $C(x)=[l(x),u(x)]$, one has
\[
\eta(x)=\mathbb{P}(Y\in C(x)\mid X=x)=F_{Y\mid X}(u(x)\mid x)-F_{Y\mid X}(l(x)\mid x).
\]
Hence $\eta(x)$ is a functional of the conditional distribution evaluated at the two interval endpoints. For an arbitrary conformal prediction method, estimating $x\mapsto \eta(x)$ may therefore be essentially as difficult as learning the conditional distribution itself.

Fortunately, CPA audits a \emph{given conformal predictor} rather than attempting to recover the full conditional law of $Y \mid X$. In the residual-score setting analyzed in Section~\ref{sec:asymptotics}, Theorem~\ref{thm::true_coverage} shows that the finite-sample conditional coverage function $\eta_n(x)$ converges uniformly to
\[
\eta(x)=\mathbb{P}\bigl(|\mu(X)+\epsilon-\tilde{\mu}(X)|\le q_{1-\alpha}\mid X=x\bigr).
\]
Accordingly, the target of CPA is not an arbitrary functional of $Y\mid X=x$, but the probability of a specific residual-based event induced by the fitted center $\tilde{\mu}(x)$ and the global threshold $q_{1-\alpha}$. This representation makes it possible to see practically relevant regimes in which the induced reliability surface $x\mapsto \eta(x)$ is substantially simpler than the full conditional law itself, as illustrated below.

\begin{example}
\label{ex:location_scale_eta_appendix}
Consider the well-specified case $\tilde\mu(x)=\mu(x)$. Then
\[
\eta(x)=\mathbb{P}\bigl(|\epsilon|\le q_{1-\alpha}\mid X=x\bigr).
\]
If we denote the conditional distribution function of the residual by
\[
G_x(t):=\mathbb{P}(\epsilon\le t\mid X=x),
\]
then
\[
\eta(x)=G_x(q_{1-\alpha})-G_x(-q_{1-\alpha}).
\]
Thus, in the residual-score setting, the target is not the full conditional law of $Y\mid X=x$, but only the probability of a particular residual event. This target can be much simpler whenever the residual law depends on $x$ through a low-dimensional feature, or in the extreme case does not depend on $x$ at all.

A particularly transparent regime is the homoscedastic case $\epsilon\perp X$. Then $G_x\equiv G$ does not vary with $x$, so
\[
\eta(x)=\mathbb{P}(|\epsilon|\le q_{1-\alpha})
\]
is constant in $x$. In this regime, a good conformal method induces a trivial reliability surface even when the regression function $\mu(x)$ remains difficult to estimate.

As a concrete example, let $\mu(x)=\|x\|^\beta$ for some $0<\beta<1$, let $X$ take values in $[-1,1]^d$, and assume $\epsilon\perp X$. Then $\mu(x)$ is $\beta$-H\"older but non-differentiable at the origin. Classical minimax theory implies that estimating the regression function over a $d$-dimensional H\"older class can be statistically difficult, with a lower bound of order $(n/\log n)^{-\beta/(2\beta+d)}$ under sup-norm loss \citep{stone1982optimal}. At the same time, uniformly consistent nonparametric regression estimators still exist under standard conditions \citep{liero1989strong}. The key point for CPA is that the induced target $\eta(x)$ may remain much simpler than the full conditional law even when the regression problem itself is high-dimensional and non-smooth.
\end{example}

To numerically demonstrate this phenomenon, we conducted a high-dimensional simulation in which
\[
Y=\mu(X)+\sigma(X_1)\epsilon,\qquad \mu(X)=\|X\|^{1/2},\qquad \sigma(X_1)=0.5+|X_1|,
\]
with $X\sim \mathrm{Unif}([-1,1]^{200})$ and $\epsilon\sim N(0,1)$ independent of $X$. In this design, the conditional law of $Y\mid X=x$ depends on a high-dimensional, non-smooth regression function, whereas for a fixed residual-score conformal predictor the corresponding conditional coverage function is driven mainly by the one-dimensional feature $X_1$. We use $n_{\mathrm{train}}=2000$ observations to construct the conformal predictor and an independent evaluation sample of size $n_{\mathrm{eval}}=500$ to learn $\eta(x)$. On an independent test sample, the pointwise true conditional coverage is approximated by Monte Carlo replication.

\begin{figure}[t]
\centering
\includegraphics[width=.9\textwidth]{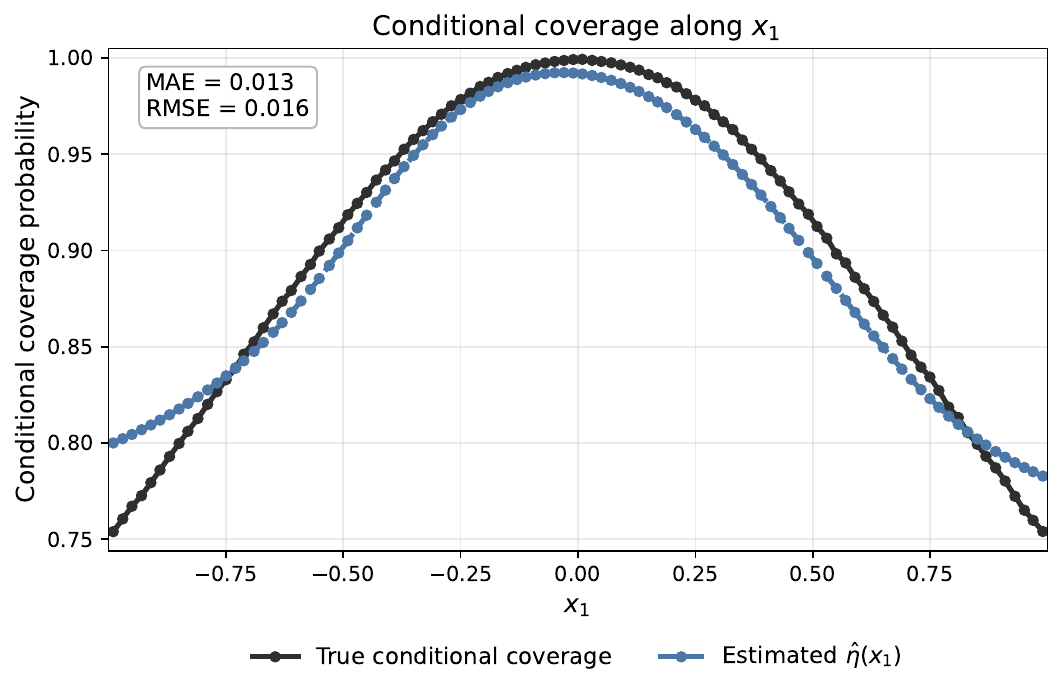}
\caption{High-dimensional illustration of the feasibility claim. Although the ambient dimension is $d=200$, the conditional coverage surface induced by the fixed residual-score conformal predictor varies mainly with $x_1$, and a one-dimensional smoother recovers this pattern closely. The MAE and RMSE are computed against Monte Carlo approximations of the pointwise true conditional coverage.}
\label{fig:eta_x1_appendix}
\end{figure}

Figure~\ref{fig:eta_x1_appendix} shows that the estimated $\hat{\eta}(x_1)$ tracks the true conditional coverage curve closely despite the ambient dimension being $d=200$. This experiment is not intended as a universal guarantee. Rather, it provides a concrete illustration that for auditing a fixed conformal predictor, learning the induced reliability surface can be considerably easier than learning the entire conditional law.

We also repeated this experiment for $d\in\{20,50,100,200\}$ while keeping the sample sizes and estimation procedure fixed. Table~\ref{tab:dimension_sensitivity_appendix} shows that the MAE and RMSE remain small across dimensions, consistent with the fact that in this example the conditional coverage function is still driven primarily by a low-dimensional feature.

\begin{table}[b]
\centering
\caption{Estimation accuracy of the conditional coverage function under different ambient dimensions. Entries report mean (SD) over 10 repetitions.}
\label{tab:dimension_sensitivity_appendix}
\begin{tabular}{lllll}
\toprule
$d$ & 20 & 50 & 100 & 200 \\
\midrule
MAE & 0.026 (0.007) & 0.024 (0.004) & 0.022 (0.006) & 0.026 (0.009) \\
RMSE & 0.034 (0.009) & 0.031 (0.005) & 0.030 (0.009) & 0.033 (0.012) \\
\bottomrule
\end{tabular}
\end{table}

Finally, the real-data experiments provide further practical support for this perspective. As shown in Figure~\ref{fig:real-reliability-grid}, the learned reliability estimator is well calibrated across the nine benchmark datasets, with empirical coverage closely tracking predicted reliability. Although such real-data evidence cannot deliver the same oracle comparison as the synthetic example above, it does indicate that in realistic applications the learned reliability surface can still serve as a meaningful and empirically well-behaved proxy for conditional coverage assessment.

\section{Detailed Algorithms} \label{app: Detailed Algorithms}

{
\begin{algorithm}[H]
\caption{Ensemble Reliability Estimator Training (CPA-Train)}
\label{alg:cpa_training}
\DontPrintSemicolon
\SetKwInOut{Input}{Input}
\SetKwInOut{Output}{Output}

\Input{Full dataset $\mathcal{D}$, nominal level $\alpha$, base CP algorithm $\mathcal{A}$, base learner $\mathcal{L}$, split ratio $\rho$ (default 0.5), number of splits $K$ (default 5).}
\Output{Ensembled reliability estimator $\hat{\eta}(\cdot)$.}

\BlankLine
Initialize estimator set $\mathcal{H} \leftarrow \emptyset$\;
\For{$k \leftarrow 1$ \KwTo $K$}{
    \tcp{Random Split}
    Partition $\mathcal{D}$ into $\mathcal{D}_{\text{train}}^{(k)}$ and $\mathcal{D}_{\text{eval}}^{(k)}$ such that $|\mathcal{D}_{\text{train}}^{(k)}| \approx \rho |\mathcal{D}|$\;
    Let $\mathcal{I}_{\text{eval}}^{(k)}$ be the indices of $\mathcal{D}_{\text{eval}}^{(k)}$\;
    \BlankLine
    \tcp{Step 1: Train CP Model}
    Train conformal predictor on current split: $\mathcal{C}^{(k)} \leftarrow \mathcal{A}(\mathcal{D}_{\text{train}}^{(k)})$\;
    \BlankLine
    \tcp{Step 2: Generate Reliability Labels}
    Initialize local reliability dataset $\mathcal{D}_{\eta}^{(k)} \leftarrow \emptyset$\;
    \For{$i \in \mathcal{I}_{\text{eval}}^{(k)}$}{
        Get data point $(X_i, Y_i)$ from $\mathcal{D}_{\text{eval}}^{(k)}$\;
        Compute coverage indicator: $I_i \leftarrow \mathbbm{1}\{Y_i \in \mathcal{C}^{(k)}(X_i)\}$\;
        $\mathcal{D}_{\eta}^{(k)} \leftarrow \mathcal{D}_{\eta}^{(k)} \cup \{(X_i, I_i)\}$\;
    }
    \BlankLine
    \tcp{Step 3: Train Base Estimator}
    Train probabilistic classifier $\hat{\eta}^{(k)}$ on $\mathcal{D}_{\eta}^{(k)}$ using algorithm $\mathcal{L}$\;
    $\mathcal{H} \leftarrow \mathcal{H} \cup \{ \hat{\eta}^{(k)} \}$\;
}
\BlankLine
\tcp{Step 4: Aggregation}
Define final estimator as the ensemble average: $\hat{\eta}(x) \leftarrow \frac{1}{K} \sum_{k=1}^K \hat{\eta}^{(k)}(x)$\;
\Return $\hat{\eta}(\cdot)$\;
\end{algorithm}
}

{
\begin{algorithm}[H]
\caption{CC-Select: Model Selection and Trust Assessment}
\label{alg:cc_select}
\DontPrintSemicolon
\SetKwInOut{Input}{Input}
\SetKwInOut{Output}{Output}

\Input{Full dataset $\mathcal{D}$, target $\alpha$, candidate procedures $\{\mathcal{A}_m\}_{m=1}^M$, base learner $\mathcal{L}$, number of splits $K$, split ratio $\rho$.}
\Output{The optimal procedure $\mathcal{A}_{m^*}$ and Trust Score $\hat{c}(\cdot)$.}
\BlankLine

\tcc{Stage 1: Model Selection via Repeated Splitting}
Initialize storage: $\text{Scores}_m \gets []$, $\text{Ensemble}_m \gets []$ for $m=1,\dots,M$\;
\For{$k = 1$ \KwTo $K$}{
    Randomly partition $\mathcal{D}$ into $\mathcal{D}_{\text{train}}^{(k)}$ and $\mathcal{D}_{\text{eval}}^{(k)}$ using ratio $\rho$\;
    \ForEach{candidate procedure $\mathcal{A}_m$}{
        \tcp{Train CP and Reliability Estimator}
        Train $\mathcal{A}_m$ on $\mathcal{D}_{\text{train}}^{(k)}$\;
        $\hat{\eta}_m^{(k)} \gets$ Train reliability estimator on $\mathcal{D}_{\text{eval}}^{(k)}$ (via Alg.~\ref{alg:cpa_training})\;
        $v_m^{(k)} \gets$ Compute CVI of $\hat{\eta}_m^{(k)}$ on $\mathcal{D}_{\text{eval}}^{(k)}$\;
        
        \tcp{Store results}
        Append $v_m^{(k)}$ to $\text{Scores}_m$\;
        Append $\hat{\eta}_m^{(k)}$ to $\text{Ensemble}_m$\;
    }
}
\BlankLine
\tcc{Aggregate results and select the best procedure}
Compute average CVI: $\overline{\text{CVI}}_m \gets \text{mean}(\text{Scores}_m)$ for $m=1,\dots,M$\;
Identify the best procedure: $m^* \gets \arg\min_m \overline{\text{CVI}}_m$\;

\BlankLine
\tcc{Stage 2: Prediction and Trust Assessment}
\tcp{Given a new data point $X_{\text{new}}$}
$\mathcal{A}_{\text{final}} \gets$ \textbf{Retrain} procedure $\mathcal{A}_{m^*}$ on full dataset $\mathcal{D}$\;
$\hat{C}_{\text{new}} \gets \mathcal{A}_{\text{final}}(X_{\text{new}})$\;
$\hat{c}_{\text{new}} \gets \frac{1}{K}\sum_{k=1}^{K} \hat{\eta} (X_{\text{new}})$ where $\hat{\eta} \in \text{Ensemble}_{m^*}$\;
\KwRet{$\hat{C}_{\text{new}}$, $\hat{c}_{\text{new}}$} \tcp{Return prediction set and trust score}
\end{algorithm}
}

\section{Calibration of the Reliability Estimator}
\label{app:calibration_methodology}

CPA learns a \emph{reliability estimator} $\hat{\eta}(x)$ to approximate the conditional coverage probability
\begin{equation}
    \eta(x) = \mathbb{P}\bigl(Y \in \mathcal{C}_\alpha(X; \mathcal{D}_{\mathrm{train}}) \mid X=x\bigr),
\end{equation}
by training on the coverage indicators $I_i = \mathbbm{1}\{Y_i \in \widehat{\mathcal{C}}(X_i)\}$ derived from the evaluation split. Because both the Conditional Validity Index (CVI) and the deployment-time trust score rely on the exact \emph{magnitude} of $\hat{\eta}(x)$, we complement the predictive learning phase with probabilistic \emph{calibration} and rigorous diagnostic checks. Crucially, unlike standard machine learning pipelines that calibrate the primary predictor for $Y$, CPA uniquely calibrates a secondary \emph{reliability estimator} dedicated to the coverage event itself. This appendix details (i) the theoretical scope of calibration within CPA, (ii) the post-hoc calibration procedure employed, and (iii) the specific metrics used for verification.

\subsection{The Role and Scope of Calibration in CPA}
\label{app:calibration_definition}

Calibration is a fundamental validity property for probability forecasts \citep{dawid1982well}. Let $I \in \{0,1\}$ denote the coverage indicator in CPA, and let $g: \mathcal{X} \to [0,1]$ be any probabilistic predictor for $I$. The estimator $g$ is defined as \emph{perfectly calibrated} if
\begin{equation}
\label{eq:perfect_calibration_cpa}
    \mathbb{P}(I=1 \mid g(X)=p) = p, \qquad \forall p \in [0,1] \ \text{(a.s.)}.
\end{equation}
Conceptually, among all samples assigned a predicted probability of $p$, the empirical frequency of true coverage should exactly equal $p$.

\textbf{Necessity for CVI.} 
CVI aggregates the absolute deviations $|\hat{\eta}(X)-(1-\alpha)|$, and the resulting trust score operates as an absolute probability of coverage. Consequently, systematic over- or under-confidence in $\hat{\eta}$ can severely distort the magnitude of these metrics, even if the estimator possesses excellent ranking performance.

\textbf{Calibration vs. Resolution.}
Importantly, perfect calibration alone does \emph{not} guarantee that $g(x)$ accurately recovers the true heterogeneity of $\eta(x)$. For instance, a naive marginal predictor $g(x) \equiv \mathbb{E}[I]$ is perfectly calibrated but completely fails to capture variations across the feature space. Therefore, the efficacy of CPA hinges on two complementary properties: (i) the \emph{resolution} (predictive capacity) of the base model to capture local coverage heterogeneity, and (ii) proper \emph{calibration} to ensure these probabilities are numerically interpretable and comparable across candidate methods.

\subsection{Post-Hoc Calibration Procedure}
\label{app:posthoc_calibration}

High-capacity learners, such as tree ensembles or neural networks, can achieve high classification accuracy yet produce notoriously miscalibrated probability outputs \citep{niculescu2005predicting}. To rectify this, we apply a post-hoc calibration step to the raw probability predictor $\hat{f}(x)$ trained on $\{(X_i, I_i)\}$. This is achieved by learning a monotonically increasing mapping $h: [0,1] \to [0,1]$, yielding the final calibrated estimator $\hat{\eta}(x) = h(\hat{f}(x))$.

\textbf{Isotonic Regression.}
We adopt isotonic regression \citep{zadrozny2002transforming} as our primary calibration technique due to its non-parametric flexibility and model-agnostic nature. Given a calibration subset $\mathcal{I}_{\mathrm{cal}}$ (strictly disjoint from the data used to fit $\hat{f}$), we solve:
\begin{equation}
\label{eq:isotonic}
    \hat{h}_{\mathrm{iso}} = \arg\min_{h \ \text{non-decreasing}} \sum_{i \in \mathcal{I}_{\mathrm{cal}}} \bigl(I_i - h(\hat{f}(X_i))\bigr)^2.
\end{equation}
To prevent overfitting and data leakage, particularly when sample sizes are constrained, we implement this calibration via a cross-fitting strategy within $\mathcal{D}_{\mathrm{eval}}$.

\textbf{Parametric Alternatives.}
While parametric approaches like Platt scaling or temperature scaling \citep{platt1999probabilistic, guo2017calibration} are viable alternatives, isotonic regression is preferred in our pipeline. It robustly accommodates the diverse array of base learners utilized in our AutoML reliability search without imposing restrictive distributional assumptions.

\subsection{Diagnostic Metrics: Reliability Diagrams and ECE}
\label{app:calibration_metrics}

We empirically assess calibration quality using reliability diagrams and the Expected Calibration Error (ECE), adhering to standard conventions \citep{guo2017calibration}. For a set of test samples $\{(X_j, I_j)\}_{j=1}^n$, we partition the predictions into $K$ equal-frequency bins $\{B_1, \dots, B_K\}$ based on the estimated probabilities $\hat{\eta}(X_j)$. For each bin $B_k$, the mean predicted confidence and empirical accuracy are computed as:
\begin{equation}
\label{eq:reliability_bins}
    \mathrm{conf}(B_k) = \frac{1}{|B_k|} \sum_{j \in B_k} \hat{\eta}(X_j), \qquad \mathrm{acc}(B_k) = \frac{1}{|B_k|} \sum_{j \in B_k} I_j.
\end{equation}
The reliability diagram plots $(\mathrm{conf}(B_k), \mathrm{acc}(B_k))$ against the ideal identity line $y=x$. This visualization is quantified by the binned ECE:
\begin{equation}
\label{eq:ece}
    \widehat{\mathrm{ECE}} = \sum_{k=1}^K \frac{|B_k|}{n} \left| \mathrm{acc}(B_k) - \mathrm{conf}(B_k) \right|.
\end{equation}
It is crucial to emphasize that $\widehat{\mathrm{ECE}}$ is a discretized diagnostic metric highly dependent on the chosen binning scheme. Within the CPA framework, it serves as a comparative diagnostic tool to verify calibration quality, rather than a distribution-free theoretical certificate.

\section{Details of Data Generating Processes}
\label{app:dgp_details}

In this appendix, we provide the full mathematical specifications for the four synthetic settings described in Section \ref{subsec:exp_design}. For all settings, the covariate dimension is fixed at $p=10$.

\subsection{Setting A: Linear, Homoscedastic}
The covariates are generated as $X_{i,j} \sim_{\text{i.i.d.}} \mathcal{N}(0, 1)$. The response $Y_i$ is generated from a sparse linear model:
\begin{equation}
    Y_i = X_i^\top \beta + \epsilon_i, \quad \epsilon_i \sim \mathcal{N}(0, 1).
\end{equation}
The coefficient vector $\beta \in \mathbb{R}^{10}$ is sparse, with the first 5 components set to 1 and the remaining 5 components set to 0.

\subsection{Setting B: Nonlinear, Heavy-Tailed}
To introduce model misspecification and test robustness against distributional violations, the response is generated as a complex nonlinear function with interactions:
\begin{equation}
    Y_i = \sin(2\pi X_{i,1}) + 2\cos(\pi X_{i,2}) + 3 X_{i,3} X_{i,4} + X_{i,5} + \epsilon_i,
\end{equation}
where the noise term $\epsilon_i \sim t_2$ follows a Student's $t$-distribution with 2 degrees of freedom. This heavy-tailed noise distribution implies infinite variance, posing a significant challenge for standard uncertainty quantification methods.

\subsection{Setting C: Heteroscedastic}
This setting is explicitly designed to evaluate the adaptivity of CP methods to local uncertainty. The mean function is linear, but the noise variance is a function of the covariates:
\begin{equation}
    Y_i = X_i^\top \beta + \epsilon_i, \quad \epsilon_i \sim \mathcal{N}(0, \sigma(X_i)^2),
\end{equation}
where the conditional standard deviation is given by $\sigma(X_i) = \exp(0.5 X_{i,1})$. The coefficients $\beta$ are the same as in Setting A.

\subsection{Setting D: Heteroscedastic with Complex Covariates}
Let $p$ be the feature dimension. We draw a $5$-sparse coefficient vector $\beta\in\mathbb{R}^p$ by sampling a support set
$S\subset\{1,\dots,p\}$ uniformly without replacement with $|S|=5$, and setting $\beta_j=\mathbbm{1}\{j\in S\}$.

For each sample $i$, generate raw features $\tilde X_{i,j}$ independently from an equal-weight mixture:
\[
\tilde X_{i,j}\sim \tfrac13\mathcal N(0,1)+\tfrac13\,\mathrm{SN}(\text{shape}=5,\text{loc}=0,\text{scale}=1)+\tfrac13\,\mathrm{Bern}(0.5),
\]
where $\mathrm{Bern}(0.5)$ takes values in $\{0,1\}$. We then induce within-row dependence by the recursion (computed in increasing $j$ with initialization $X=\tilde X$ and clamping indices below $1$):
\[
X_{i,j}=0.7\,\tilde X_{i,j}+0.3\,X_{i,\max(1,j-3)},\qquad j=1,\dots,p.
\]

Define $\mu_i=X_i^\top\beta$. The heteroscedastic scale is
\[
\sigma_i=1+2\frac{|\mu_i|^3}{\mathbb E(|X^\top\beta|^3)},
\]
where $\mathbb E(|X^\top\beta|^3)$ is approximated by Monte Carlo using $10{,}000$ auxiliary draws from the same covariate model (with a $10^{-9}$ stabilizer). Responses follow
\[
Y_i=\mu_i+\sigma_i\,\xi_i,\qquad \xi_i\stackrel{i.i.d.}{\sim} t_{2}.
\]
For testing, we additionally generate $R$ i.i.d.\ replicates $\{Y_{i,r}\}_{r=1}^R$ per fixed test covariate $X_i$ using the same $\mu_i,\sigma_i$.

\section{Implementation Details of Benchmark Methods}
\label{app:benchmark_details}

We evaluated nine conformal prediction methods as baselines to provide a comprehensive comparative analysis. All methods were configured to produce prediction intervals with a target marginal coverage of $1 - \alpha = 0.9$. This section details the specific algorithmic configurations and underlying base learners employed for each benchmark.

\subsection{Implementation Details of Classical Statistical Baselines}
\label{app:classical_baselines}
We first outline the classical statistical baselines. Each method serves as a tailored parametric or asymptotic oracle for a specific DGP setting, allowing us to establish a robust traditional benchmark against which the conformal assessment framework can be evaluated. All classical routines are implemented utilizing the \texttt{statsmodels} and \texttt{scikit-learn} libraries in Python. \looseness=-1

\subsubsection{Ordinary Least Squares (OLS) for Standard Regimes}
\textbf{Target:} Setting A and Setting D.

We employed the standard OLS prediction interval as a classical parametric baseline. Under the well-specified linear Gaussian regime (Setting A), OLS serves as the theoretical oracle. Conversely, in Setting D, following the evaluation protocol of \citet{lei2018distribution}, OLS is included as a negative control to demonstrate the failure of classical variance-based methods when fundamental parametric assumptions (e.g., finite variance and feature independence) are violated.

\textbf{Formulation and Implementation:} Assuming the model $Y = \mathbf{X}\beta + \epsilon$ with $\epsilon \sim \mathcal{N}(0, \sigma^2 \mathbf{I})$, the $(1-\alpha)$ prediction interval is analytically derived based on the Student's $t$-distribution with $n-p$ degrees of freedom. We explicitly added a constant intercept term to the design matrix. The intervals were computed using the \texttt{get\_prediction} interface in \texttt{statsmodels}, which incorporates both the estimated error variance $\hat{\sigma}^2$ and the uncertainty in parameter estimation.

\subsubsection{Spline-Based Generalized Additive Models (GAM) for Nonlinearity}
\textbf{Target:} Setting B.

To address model misspecification in the nonlinear setting without resorting to black-box kernels, we implemented a Generalized Additive Model (GAM) using basis expansions.
\begin{itemize}
    \item \textbf{Basis Expansion:} We utilized the \texttt{patsy} library to transform each feature $X_j$ into a B-spline basis with degree 3 (cubic splines) and 5 internal knots placed at uniform empirical quantiles.
    \item \textbf{Interval Construction:} Given the heavy-tailed nature of Setting B ($t_2$ noise), standard Gaussian assumptions fail. We instead adopted a naive empirical residual approach. Let $\mathcal{R} = \{y_i - \hat{\mu}(x_i)\}_{i=1}^n$ be the residuals on the training set. The prediction interval for a test point $x_{n+1}$ is constructed as $[\hat{\mu}(x_{n+1}) + \hat{Q}_{\alpha/2}(\mathcal{R}), \; \hat{\mu}(x_{n+1}) + \hat{Q}_{1-\alpha/2}(\mathcal{R})]$, where $\hat{Q}_\tau$ denotes the empirical quantile function.
\end{itemize}

\subsubsection{Two-Stage Weighted Least Squares (WLS) for Heteroscedasticity}
\textbf{Target:} Setting C.

Standard OLS generates constant-width intervals, which are invalid under heteroscedasticity. We implemented a two-stage Feasible Generalized Least Squares (FGLS) approach to explicitly model the conditional variance.
\begin{enumerate}
    \item \textbf{Mean Estimation:} An initial OLS model is fitted to obtain the residuals $r_i = y_i - \hat{\mu}_{\text{OLS}}(x_i)$.
    \item \textbf{Variance Estimation:} To capture the unknown variance structure $\sigma^2(x)$, we trained a non-parametric variance model $\hat{g}(x)$ to predict the log-squared residuals $\log(r_i^2)$. We used a Gradient Boosting Regressor for this step, configured with 200 estimators, a learning rate of 0.05, and a maximum depth of 3. The estimated weights are $w_i = 1/\exp(\hat{g}(x_i))$.
    \item \textbf{WLS Refinement:} The final mean model is re-estimated using Weighted Least Squares (WLS) with weights $w_i$. The prediction interval is locally adaptive:
    $$\hat{y}(x) \pm t_{df, 1-\alpha/2} \sqrt{\text{SE}_{\text{mean}}^2 + \exp(\hat{g}(x))}.$$
\end{enumerate}

\subsection{Adaptive Residual Bootstrap}
\label{app:method_bootstrap}

While conformal prediction offers finite-sample guarantees, traditional bootstrap methods remain a powerful frequentist benchmark for asymptotic validity. We implemented an \textit{Adaptive Residual Bootstrap} procedure that dynamically selects the underlying regression model to match the data structure.

\textbf{Methodology.} \quad
The procedure consists of three stages: model selection, residual calibration, and bootstrap aggregation.
\begin{enumerate}
    \item \textbf{Adaptive Model Selection:} To avoid specifying a functional form (linear vs. nonlinear) a priori, we evaluate two candidate learners on the training set $\mathcal{D}_{\text{train}}$: a linear \texttt{Lasso} and a nonlinear \texttt{RandomForest}. The optimal base learner $\hat{f}$ is selected via 5-fold Cross-Validation (CV) based on the lowest Mean Squared Error (MSE).
    \item \textbf{Residual Calibration:} Standard in-sample residuals often underestimate prediction error. We explicitly compute \textit{out-of-sample} residuals $r_i = y_i - \hat{f}^{-k(i)}(x_i)$ using $K$-fold CV, where $\hat{f}^{-k(i)}$ is trained on folds excluding $i$. These residuals are then centered: $\tilde{r}_i = r_i - \bar{r}$.
    \item \textbf{Bootstrap Inference:} We generate $B=500$ bootstrap replications. In each iteration $b$:
    \begin{itemize}
        \item A bootstrap sample $\mathcal{D}^*_b$ is drawn from $\mathcal{D}_{\text{train}}$ with replacement.
        \item The selected estimator is re-trained on $\mathcal{D}^*_b$ to obtain $\hat{f}^*_b$.
        \item For a test point $x_{n+1}$, we simulate the predictive distribution as $y^*_{b} = \hat{f}^*_b(x_{n+1}) + \epsilon^*_b$, where $\epsilon^*_b$ is drawn \textit{with replacement} from the centered residuals $\{\tilde{r}_i\}$.
    \end{itemize}
    The $(1-\alpha)$ prediction interval is given by the empirical $\alpha/2$ and $1-\alpha/2$ quantiles of $\{y^*_b\}_{b=1}^B$.
\end{enumerate}

\textbf{Implementation Details.} \quad
\begin{itemize}
    \item \textbf{Hyperparameter Tuning:} For Random Forest, we performed a grid search over \texttt{n\_estimators} $\in \{200, 500\}$, \texttt{max\_depth} $\in \{5, 10, \text{None}\}$, \texttt{min\_samples\_split} $\in \{2, 5\}$, and \texttt{max\_features} $\in \{\text{sqrt}, \text{log2}\}$. For Lasso, the regularization path was automatically tuned via \texttt{LassoCV}.
    \item \textbf{Computation:} The bootstrap loop was parallelized using Python's \texttt{multiprocessing} module to ensure efficiency. We fixed $B=500$ to balance computational cost and quantile estimation stability.
\end{itemize}

\subsection{Quantile Regression Forests (QRF)}
\label{app:method_qrf}

To estimate conditional prediction intervals without relying on linear assumptions or explicit residual modeling, we implemented Quantile Regression Forests. QRF generalizes random forests by estimating the full conditional distribution $P(Y|X=x)$ rather than just the conditional mean.

\textbf{Methodology} \quad
The method operates in two phases: structural learning and distribution estimation.
\begin{enumerate}
    \item \textbf{Tree Structure Learning:} A standard random forest of $B$ trees is trained to minimize the Mean Squared Error. For a given input $x$, let $L_b(x)$ denote the set of training indices falling into the same leaf as $x$ in the $b$-th tree.
    \item \textbf{Conditional Distribution:} Instead of averaging the leaf means, QRF aggregates the observed response values from all leaves associated with $x$. The estimated conditional distribution function is given by:
    \begin{equation}
        \hat{F}(y|x) = \frac{1}{B} \sum_{b=1}^B \sum_{i \in L_b(x)} \frac{1}{|L_b(x)|} \mathbbm{1}_{\{Y_i \le y\}}.
    \end{equation}
    In our implementation, this is approximated by collecting the ensemble of target values $\mathcal{Y}_x = \bigcup_{b=1}^B \{Y_i : i \in L_b(x)\}$ and computing empirical quantiles.
    \item \textbf{Interval Construction:} The prediction interval is defined as $[\hat{Q}_{\alpha/2}(\mathcal{Y}_x), \hat{Q}_{1-\alpha/2}(\mathcal{Y}_x)]$, where $\hat{Q}_\tau$ is the $\tau$-quantile of the collected values.
\end{enumerate}

\textbf{Implementation Details}
\begin{itemize}
    \item \textbf{Hyperparameter Optimization:} We tuned the random forest structure using 5-fold GridSearchCV on the training set. The search space included \texttt{n\_estimators} $\in \{200, 500\}$, \texttt{max\_depth} $\in \{5, 10\}$, and \texttt{min\_samples\_leaf} $\in \{5, 10\}$. The latter parameter is critical as it controls the sample size available for local distribution estimation.
    \item \textbf{Algorithm:} We utilized \texttt{scikit-learn}'s \texttt{RandomForestRegressor} to grow the trees and its \texttt{.apply()} method to retrieve leaf indices. The conditional quantiles were computed using \texttt{numpy.quantile} on the aggregated leaf samples.
\end{itemize}

\subsection{Split Conformal Prediction with Adaptive Selection}
\label{app:method_scp}

We implemented the standard Split Conformal Prediction (SCP) method \citep{lei2018distribution} augmented with an adaptive model selection step. This approach guarantees finite-sample marginal coverage regardless of the underlying data distribution, provided the observations are exchangeable.

\textbf{Methodology.} \quad
The procedure involves a three-way data split: a training set $\mathcal{D}_{\text{train}}$ for model learning, a calibration set $\mathcal{D}_{\text{calib}}$ for score computation, and a test set for evaluation.
\begin{enumerate}
    \item \textbf{Adaptive Model Learning:} To accommodate unknown DGPs ranging from linear to highly nonlinear, we perform model selection on $\mathcal{D}_{\text{train}}$. We evaluate two candidate learners: \texttt{Lasso} (linear baseline) and \texttt{RandomForest} (nonlinear baseline). The optimal regressor $\hat{\mu}$ is selected based on the minimum 5-fold Cross-Validation MSE.
    \item \textbf{Calibration:} We compute the non-conformity scores on the hold-out calibration set $\mathcal{D}_{\text{calib}}$ as the absolute prediction errors: $S_i = |y_i - \hat{\mu}(x_i)|$.
    \item \textbf{Interval Construction:} Let $\hat{q}_{1-\alpha}$ be the $(1-\alpha)$-quantile of the calibration scores $\{S_i\}$. The prediction interval for a new point $x$ is constructed as $\mathcal{C}(x) = [\hat{\mu}(x) - \hat{q}_{1-\alpha}, \hat{\mu}(x) + \hat{q}_{1-\alpha}]$. This results in intervals of constant width across the feature space.
\end{enumerate}

\textbf{Implementation Details.} \quad
\begin{itemize}
    \item \textbf{Software:} The conformalization process was implemented using the \texttt{MAPIE} library in Python, utilizing the \texttt{cv="prefit"} mode to strictly separate training and calibration data.
    \item \textbf{Hyperparameters:} The \texttt{RandomForest} candidate was tuned over a grid of \texttt{n\_estimators} $\in \{200, 500\}$ and \texttt{max\_depth} $\in \{5, 10, \text{None}\}$. The \texttt{Lasso} regularization path was automatically optimized via internal cross-validation.
    \item \textbf{Data Splitting:} We used an equal split ratio between training and calibration sets ($\rho_{\text{calib}}=0.5$) to balance estimation quality and calibration stability.
\end{itemize}

\subsection{Locally Adaptive (Studentized) Split Conformal Prediction}
\label{app:method_studentized}

Standard split conformal prediction produces intervals of constant width, which fail to reflect local uncertainty in heteroscedastic data. To address this, we implemented the Locally Adaptive Split Conformal Prediction method, often referred to as Studentized Conformal Prediction.

\textbf{Methodology.} \quad
This approach normalizes the non-conformity score by an estimate of the local conditional dispersion. The procedure consists of two training steps on the proper training set $\mathcal{D}_{\text{train}}$:
\begin{enumerate}
    \item \textbf{Mean Estimator:} A primary regression model $\hat{\mu}$ is trained to predict the conditional mean $\mathbb{E}[Y|X]$.
    \item \textbf{Dispersion Estimator:} We compute the absolute residuals on the training set, $R_i = |Y_i - \hat{\mu}(X_i)|$. A secondary regression model $\hat{\sigma}$ is then trained on $\{(X_i, R_i)\}$ to predict these residuals, effectively estimating the local mean absolute deviation.
    \item \textbf{Studentized Calibration:} On the calibration set $\mathcal{D}_{\text{calib}}$, we compute the normalized non-conformity scores:
    \begin{equation}
        S_i = \frac{|Y_i - \hat{\mu}(X_i)|}{\hat{\sigma}(X_i) + \epsilon},
    \end{equation}
    where $\epsilon=10^{-6}$ is a small constant added for numerical stability.
    \item \textbf{Interval Construction:} Let $\hat{Q}_{1-\alpha}$ be the $(1-\alpha)$-quantile of the scores $\{S_i\}$. The prediction interval for a new point $x$ scales with the estimated difficulty: $\mathcal{C}(x) = \hat{\mu}(x) \pm \hat{Q}_{1-\alpha} \cdot \hat{\sigma}(x)$.
\end{enumerate}

\textbf{Implementation Details}
\begin{itemize}
    \item \textbf{Base Learners:} Both $\hat{\mu}$ and $\hat{\sigma}$ were instantiated as \texttt{RandomForestRegressor} models to capture nonlinear patterns in both the central tendency and the variability.
    \item \textbf{Optimization:} We performed independent 3-fold Cross-Validation Grid Searches for both models to optimize their hyperparameters. The search grid included \texttt{n\_estimators} $\in \{200, 500\}$, \texttt{max\_depth} $\in \{5, 10, \text{None}\}$, and \texttt{min\_samples\_split} $\in \{2, 5\}$.
    \item \textbf{Data Splitting:} Consistent with the standard split conformal approach, we maintained disjoint training and calibration sets.
\end{itemize}

\subsection{Conformalized Quantile Regression (CQR)}
\label{app:method_cqr}

To construct intervals that adapt to both the location and spread of the data distribution, we employed Conformalized Quantile Regression (CQR) \citep{romano2019conformalized}. CQR calibrates initial estimates of the conditional quantiles to guarantee validity.

\textbf{Methodology.} \quad
The method proceeds in two steps:
\begin{enumerate}
    \item \textbf{Quantile Estimation:} We estimate the lower ($\alpha/2$) and upper ($1-\alpha/2$) conditional quantiles using a regression model trained on $\mathcal{D}_{\text{train}}$. To robustly handle diverse DGPs, we implemented an adaptive selection strategy that chooses between a linear \texttt{QuantileRegressor} and a nonlinear \texttt{GradientBoostingRegressor} based on cross-validation performance.
    \item \textbf{Conformal Calibration:} We compute non-conformity scores on a calibration set as $S_i = \max(\hat{q}_{\alpha/2}(X_i) - Y_i, \; Y_i - \hat{q}_{1-\alpha/2}(X_i))$. The final interval is constructed by expanding (or contracting) the initial quantile estimates by the $(1-\alpha)$-quantile of these scores: $\mathcal{C}(x) = [\hat{q}_{\alpha/2}(x) - Q_{1-\alpha}, \hat{q}_{1-\alpha/2}(x) + Q_{1-\alpha}]$.
\end{enumerate}

\textbf{Implementation Details}
\begin{itemize}
    \item \textbf{Base Learners:} The linear candidate employed L1-regularization, tuned over $\lambda \in \{10^{-3}, 10^{-2}, 0.1, 1\}$. The gradient boosting candidate was optimized over \texttt{n\_estimators} $\in \{100, 200, 300\}$, \texttt{max\_depth} $\in \{3, 5, 7\}$, and \texttt{learning\_rate} $\in \{0.01, 0.05, 0.1\}$.
    \item \textbf{Software:} We utilized the \texttt{MapieQuantileRegressor} from the MAPIE library to perform the calibration step.
\end{itemize}

\subsection{Cross-Validation Plus (CV+)}
\label{app:method_cvplus}

To overcome the sample efficiency limitations of split conformal prediction without incurring the high computational cost of full jackknife (leave-one-out) methods, we employed the Cross-Validation Plus (CV+) method \citep{barber2021predictive}.

\textbf{Methodology.}
CV+ constructs prediction intervals using the empirical distribution of leave-one-fold-out residuals.
\begin{enumerate}
    \item \textbf{Base Model Selection:} Similar to the split conformal approach, we first select the optimal base learner $\hat{\mu}$ (Lasso or Random Forest) via 5-fold cross-validation on the training set.
    \item \textbf{K-Fold Training:} The training data is partitioned into $K$ disjoint folds $S_1, \dots, S_K$. For each $k \in \{1, \dots, K\}$, a model $\hat{\mu}_{-k}$ is trained on the data excluding fold $S_k$.
    \item \textbf{Residual Computation:} For each training sample $i$ belonging to fold $S_{k(i)}$, we compute the absolute out-of-fold residual $R_i = |y_i - \hat{\mu}_{-k(i)}(x_i)|$.
    \item \textbf{Interval Construction:} 
    For a test point $x$, the prediction interval is constructed by taking empirical quantiles of the shifted predictions:
    \[
    \mathcal{C}_{\text{CV+}}(x)
    =
    \Big[
    \operatorname{Quantile}_{\alpha}\big\{ \hat{\mu}_{-k(i)}(x) - R_i \big\}_{i=1}^n,\;
    \operatorname{Quantile}_{1-\alpha}\big\{ \hat{\mu}_{-k(i)}(x) + R_i \big\}_{i=1}^n
    \Big].
    \]
\end{enumerate}

\textbf{Implementation Details}
\begin{itemize}
    \item \textbf{Algorithm:} We utilized the \texttt{MapieRegressor} with \texttt{method="plus"} and \texttt{cv=5}. This $K=5$ configuration balances coverage validity with computational feasibility.
    \item \textbf{Hyperparameters:} The base estimators were tuned over the same grid as in the split conformal experiments (\texttt{n\_estimators} $\in \{200, 500\}$ for RF).
\end{itemize}

\subsection{Localized Conformal Prediction (LCP)}
\label{app:method_lcp}

To adapt to local heteroscedasticity without the randomness of data splitting or auxiliary variance modeling, we implemented Localized Conformal Prediction \citep{guan2023localized}.

\textbf{Methodology.} \quad
LCP modifies the standard conformal score integration by re-weighting calibration samples based on their proximity to the test point.
\begin{enumerate}
    \item \textbf{Local Weighting:} Let $\hat{\mu}$ be the base regression model. For a test point $x$, we assign a weight to each calibration residual $S_i = |y_i - \hat{\mu}(X_i)|$ using a Gaussian kernel:
    \begin{equation}
        w_h(x, X_i) = \exp\left(-\frac{\|x - X_i\|_2^2}{h}\right).
    \end{equation}
    \item \textbf{Bandwidth Selection:} The performance of LCP is sensitive to the kernel bandwidth $h$. We implemented the data-driven \texttt{autoTune} procedure described in \citet{guan2023localized}. This routine performs bootstrap resampling on the training data to select the $h$ that minimizes the deviation between empirical coverage and the target level $1-\alpha$.
    \item \textbf{Inference:} The prediction interval is constructed using the weighted quantile of the calibration scores, where the probability mass of each score $S_i$ is proportional to $w_h(x, X_i)$.
\end{enumerate}

\textbf{Implementation Details}
\begin{itemize}
    \item \textbf{Base Learner:} A \texttt{RandomForestRegressor} (200 trees) trained on $\mathcal{D}_{\text{train}}$.
    \item \textbf{Optimization:} To accelerate the computationally intensive re-weighting step for every test point, we vectorized the distance computations using \texttt{scipy.spatial.distance.cdist} and parallelized the quantile search using \texttt{joblib}.
\end{itemize}

\subsection{Randomized Local Conformal Prediction (RLCP)}
\label{app:method_rlcp}

Standard LCP does not guarantee finite-sample marginal coverage in all settings. To address this theoretical limitation while maintaining local adaptivity, we implemented Randomized LCP (RLCP) \citep{hore2025conformal}.

\textbf{Methodology.} \quad
RLCP introduces a randomization mechanism to the weight calculation to restore exchangeability properties.
\begin{enumerate}
    \item \textbf{Noise Injection:} Instead of using the exact distance $\|x - X_i\|$, RLCP computes weights based on perturbed locations. We specifically implemented the Gaussian smoothing kernel variant. For a test point $x$, we generate random noise vectors $\xi \sim \mathcal{N}(0, h^2 I)$.
    \item \textbf{Smoothed Weights:} The weights are computed as $w(x, X_i) = \exp(-\|x + \xi - X_i\|^2 / 2h^2)$. In our implementation, we utilized the \textbf{m-RLCP} algorithm, which averages the threshold over $m$ independent noise realizations to reduce the variance of the interval boundaries.
    \item \textbf{Bandwidth Selection:} Unlike LCP, for RLCP we employed a computational heuristic for bandwidth selection to ensure scalability. We set $h$ to the median distance between each calibration point and its $k$-nearest neighbors in the training set, with $k=\sqrt{n_{\text{train}}}$.
\end{enumerate}

\textbf{Implementation Details}
\begin{itemize}
    \item \textbf{Parameters:} We set the number of noise realizations to $m=10$.
    \item \textbf{Computation:} The neighborhood search for bandwidth selection was implemented using \texttt{sklearn.neighbors.NearestNeighbors}. The noise injection and threshold averaging were parallelized across CPU cores using a dynamic batching strategy.
\end{itemize}

\section{Definitions of Ranking Metrics}
\label{app:ranking_metrics}

In Section \ref{subsec:oracle_recovery}, we employ a suite of metrics to evaluate the fidelity of the CPA-estimated ranking $\hat{\pi}$ relative to the ground-truth Oracle ranking $\pi$. Let the set of $M$ benchmarking methods be indexed by $\{1, \dots, M\}$.

\subsection{Distance-Weighted Kendall's Tau ($\tau_w$)}
Standard rank correlation metrics treat all transpositions equally. In the context of model selection, penalizing the transposition of two models with nearly identical performance is overly harsh, whereas failing to distinguish between a highly accurate model and a severely miscalibrated one constitutes a critical failure. Therefore, we adopt a \textbf{Distance-Weighted Kendall's $\tau$}, which assigns penalties proportional to the actual performance gap.

Let $(i, j)$ denote a pair of algorithms, and $d_i, d_j$ be their ground-truth conditional validity errors. We define the indicator of concordance as $K_{ij} = \operatorname{sgn}(\pi(i) - \pi(j)) \cdot \operatorname{sgn}(\hat{\pi}(i) - \hat{\pi}(j))$. The weighted correlation is computed as:
\begin{equation}
    \tau_w(\pi, \hat{\pi}) = \frac{\sum_{1 \le i < j \le M} w_{ij} K_{ij}}{\sum_{1 \le i < j \le M} w_{ij}},
\end{equation}
where the weight $w_{ij} = |d_i - d_j|^p$ (with $p=1$ in our experiments) ensures that pairs with larger performance disparities exert a greater influence on the correlation score.

\subsection{Normalized Discounted Cumulative Gain (NDCG)}
While Kendall's $\tau_w$ measures general pairwise correlation, NDCG quantifies the \textit{utility} of the selection, strictly rewarding the estimator for placing high-quality models at the top of the list. The metric at cutoff $k$ is defined as:
\begin{equation}
    \text{NDCG}@k = \frac{\text{DCG}@k}{\text{IDCG}@k}, \quad \text{with } \text{DCG}@k = \sum_{p=1}^k \frac{\text{rel}(\hat{\pi}^{-1}(p))}{\log_2(p+1)},
\end{equation}
where $\hat{\pi}^{-1}(p)$ denotes the item at rank $p$ in the estimated list. To prevent numerical instability associated with reciprocal transformations, the relevance score is derived using a linear inversion of the true distance: $\text{rel}(i) = \max_{j}(d_j) - d_i$. This mapping ensures that models with smaller true errors naturally yield proportionally higher utility gains. IDCG is the ideal gain achieved by the perfect Oracle permutation.

\subsection{Hit Rate (Hit@k)}
To directly assess the success rate of identifying the optimal subset of algorithms, we use the Hit Rate. This metric calculates the overlap proportion between the top-$k$ sets identified by the Oracle ($\mathcal{S}_{\text{oracle}}^{(k)}$) and the CPA estimator ($\mathcal{S}_{\text{est}}^{(k)}$):
\begin{equation}
    \text{Hit}@k = \frac{|\mathcal{S}_{\text{oracle}}^{(k)} \cap \mathcal{S}_{\text{est}}^{(k)}|}{k}.
\end{equation}
A Hit@k of 1 implies that the CPA framework successfully retrieved all $k$ best-performing models, regardless of their internal relative order.

\section{Details of Reliability Estimator Configurations}
\label{app:estimator_configs}
\subsection{The Automated Model Selection Procedure (Baseline)}
\label{app:autoML}

The \textbf{Baseline} reliability estimator utilizes a rigorous automated machine learning (AutoML) pipeline to identify the optimal probabilistic classifier for conditional coverage estimation. Given that standard target marginal coverage levels (e.g., $1-\alpha = 0.9$) inherently induce severe class imbalance within the coverage indicators, all candidate estimators are trained utilizing inverse class-frequency weighting. The pipeline consists of three sequential stages:

\begin{enumerate}
    \item \textbf{Hypothesis Space Definition.} We construct a diverse pool of classification families to capture varying degrees of nonlinearity and feature interactions. For each family, we optimize over a pre-defined hyperparameter grid $\Theta$:
    \begin{itemize}
        \item \textbf{Logistic Regression (LR):}
        \begin{itemize}[label=$\cdot$, nosep]
            \item Regularization penalty: $\ell_1, \ell_2$
            \item Inverse regularization strength $C$: $\{0.01, 0.1, 1, 10, 100\}$
            \item Solvers: $\{\text{liblinear}, \text{saga}\}$
        \end{itemize}
        
        \item \textbf{Random Forest (RF):}
        \begin{itemize}[label=$\cdot$, nosep]
            \item Number of estimators $B$: $\{100, 200, 500\}$
            \item Maximum depth $d$: $\{\text{None}, 5, 10\}$
            \item Minimum samples split: $\{2, 5\}$
            \item Max features: $\{\sqrt{p}, \log_2 p\}$
        \end{itemize}
        
        \item \textbf{Gradient Boosting Machine (GBM):}
        \begin{itemize}[label=$\cdot$, nosep]
            \item Number of estimators $B$: $\{100, 200, 300\}$
            \item Learning rate $\eta$: $\{0.01, 0.05, 0.1\}$
            \item Subsampling rate: $\{0.8, 1.0\}$
            \item Minimum samples split: $\{2, 5\}$
        \end{itemize}
        
        \item \textbf{XGBoost (XGB):}
        \begin{itemize}[label=$\cdot$, nosep]
            \item Number of estimators $B$: $\{100, 200\}$
            \item Maximum depth $d$: $\{3, 5, 7\}$
            \item Learning rate $\eta$: $\{0.01, 0.1\}$
            \item Subsample \& Colsample by tree: $\{0.8, 1.0\}$
        \end{itemize}
        
        \item \textbf{K-Nearest Neighbors (KNN):}
        \begin{itemize}[label=$\cdot$, nosep]
            \item Neighbors $k$: $\{3, 5, 10, 30\}$
            \item Weighting: $\{\text{uniform}, \text{distance}\}$
            \item Metric: $\{\text{Minkowski ($p=1, p=2$)}, \text{Euclidean}\}$
        \end{itemize}
    \end{itemize}
    
    \item \textbf{Estimator Selection via Cross-Validation.} Candidate models are evaluated using stratified $K$-fold cross-validation on the hold-out evaluation set $\mathcal{D}_{\text{eval}}$. To maintain sufficient minority class representation at extreme quantiles ($\alpha \le 0.02$), the number of folds $K$ is dynamically adjusted from 5 to 2. The configuration $\hat{f}_{\text{base}}$ that minimizes the cross-validated negative log-likelihood (log-loss) is selected as the optimal base estimator.
    
    \item \textbf{Non-parametric Calibration.} High-capacity classifiers often produce uncalibrated probability estimates. To rectify this without distorting the learned ranking, we apply \textbf{Isotonic Regression} \citep{zadrozny2002transforming} via cross-fitting. This step fits a monotonically non-decreasing piecewise constant function to the raw scores, yielding the final calibrated reliability estimates $\hat{\eta}(x)$.
\end{enumerate}

\subsection{Perturbation Configurations}
\label{app:perturbations}

To rigorously evaluate the robustness of the proposed conformal assessment framework, we introduce seven structural and procedural perturbations to the reliability estimator. These configurations are meticulously designed to simulate a spectrum of sub-optimal estimation behaviors, ranging from absent calibration to severe capacity constraints. The exact specifications, derived directly from our evaluation pipeline, are detailed in Table \ref{tab:estimator_configs_detailed}.

{
\linespread{1}
\begin{table}[h]
    \centering
    \caption{Detailed configurations of reliability estimators used in the robustness analysis. The \textbf{Baseline} employs the full AutoML pipeline coupled with isotonic calibration.}
    \label{tab:estimator_configs_detailed}
    \small
    \renewcommand{\arraystretch}{1.4} 
    \begin{tabular}{l p{8.5cm} l}
    \toprule
    \textbf{ID} & \textbf{Model Architecture / Specification} & \textbf{Calibration} \\
    \midrule
    \textbf{Baseline} & \textbf{AutoML} (Best of LR, RF, GBM, XGB, KNN pool) & Isotonic \\
    \midrule
    \multicolumn{3}{l}{\textit{Sensitivity to Calibration Method}} \\
    \texttt{pert\_1} & Same as Baseline (AutoML) & \textbf{None (Raw Output)} \\
    \texttt{pert\_2} & Same as Baseline (AutoML) & \textbf{Sigmoid (Platt)} \\
    \texttt{pert\_3} & Same as Baseline (AutoML) & \textbf{Histogram Binning} \\
    \midrule
    \multicolumn{3}{l}{\textit{Sensitivity to Structural Misspecification}} \\
    \texttt{pert\_4} & \textbf{Logistic Regression} (Tuned via CV) & Isotonic \\
                     & Grid: Same as Baseline LR grid & \\
    \midrule
    \multicolumn{3}{l}{\textit{Sensitivity to Model Capacity}} \\
    \texttt{pert\_5} & \textbf{Random Forest (Underfit)} & Isotonic \\
                     & Fixed parameters: \texttt{n\_estimators=50}, \texttt{max\_depth=3} & \\
    \texttt{pert\_6} & \textbf{Random Forest (Restricted Depth)} & Isotonic \\
                     & Fixed parameters: \texttt{n\_estimators=100}, \texttt{max\_depth=6} & \\
    \midrule
    \multicolumn{3}{l}{\textit{Alternative Strong Learner}} \\
    \texttt{pert\_7} & \textbf{Gradient Boosting} (Tuned via CV) & Isotonic \\
                     & Grid: Same as Baseline GBM grid & \\
    \bottomrule
    \end{tabular}
\end{table}
}

\textbf{Rationale for Perturbations.} \quad
Configurations \texttt{pert\_1} through \texttt{pert\_3} isolate the impact of the post-hoc calibration phase, contrasting the baseline isotonic approach with raw uncalibrated probabilities, parametric scaling, and discrete histogram binning. Configuration \texttt{pert\_4} restricts the hypothesis space entirely to linear decision boundaries, acting as a structural stress test under nonlinear data generating processes. Configurations \texttt{pert\_5} and \texttt{pert\_6} bypass the cross-validation tuning phase, imposing strict architectural bottlenecks on tree ensembles to simulate under-parameterized regimes. Finally, \texttt{pert\_7} serves as a competent, single-family baseline to evaluate whether the heterogeneous AutoML ensemble search strictly outperforms a well-tuned classical boosting algorithm.

\section{Bias Analysis and Mechanism of Failure}
\label{app:bias_analysis}

In Section \ref{subsec:robustness}, we identified a specific ranking failure associated with the Logistic Regression estimator in the nonlinear Setting B. To elucidate the underlying mechanism, we analyze the distribution of pointwise estimation bias, formally defined as $\text{Bias}(x) = \hat{\eta}(x) - \eta_{\text{oracle}}(x)$. To construct these empirical distributions, we aggregated pointwise estimates across multiple independent trials and performed a stratified random subsampling of 50,000 evaluation points per configuration. This implementation ensures robust kernel density estimation while maintaining computational tractability.

Figure~\ref{fig:bias_violin} presents the resulting bias manifolds for representative configurations. The \textbf{Baseline} estimator (positioned leftmost), alongside other structurally adequate perturbations, consistently exhibits a distribution tightly centered near zero with relatively light tails. This zero-centering is a strong indicator of empirical consistency, confirming that high-capacity estimators can provide nearly unbiased reliability estimates without systematic distortions. In stark contrast, the \textbf{Logistic Regression} estimator in Setting B displays a severe negative location shift and heavy skewness. 

This systematic downward bias explicitly confirms that the ranking failure observed in the main text is strictly driven by \textit{underfitting}: the linear decision boundary is structurally incapable of isolating the complex, localized regions of under-coverage inherent to the nonlinear data manifold. Consequently, the estimator lacks the necessary resolution to resolve localized coverage deficiencies, effectively masking these critical failures through inappropriate global averaging. This fundamental loss of sensitivity degrades the diagnostic fidelity of the CPA audit and misguides the subsequent algorithm selection.

\begin{figure}[h]
    \centering
    \includegraphics[width=0.95\textwidth]{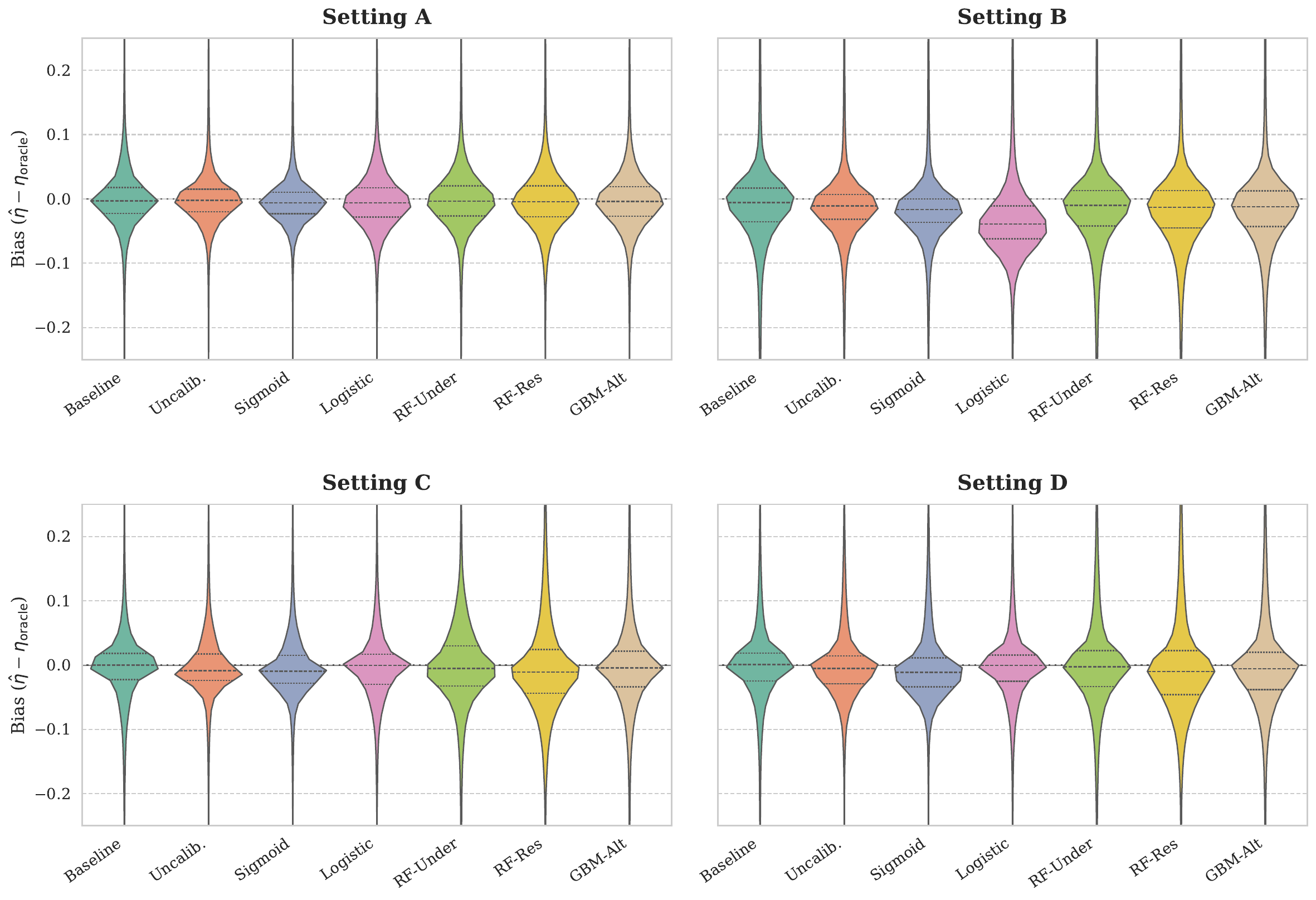}
    \caption{\textbf{Distribution of pointwise estimation bias $\hat{\eta}(x) - \eta_{\text{oracle}}(x)$.} While the Baseline and other high-capacity estimators remain robustly centered near zero, the Logistic Regression estimator exhibits a severe systemic downward bias in the nonlinear Setting B, corroborating the capacity mismatch hypothesis.}
    \label{fig:bias_violin}
\end{figure}

\section{Sensitivity Analysis of Data Allocation}
\label{app:sensitivity_rho}

In this section, we provide a detailed empirical analysis of the trade-off involved in the data allocation between the conformal predictor and the reliability estimator, as discussed in Section \ref{sec:simulation}.

\subsection{The Bias-Variance Trade-off}

A critical decision in deploying the CPA framework is the partitioning of the available labeled data $\mathcal{D}$ ($| \mathcal{D} | = N$) into a training set for the conformal predictor ($n_{\text{train}}$) and an evaluation set for the reliability estimator ($n_{\text{eval}}$). Let $\rho = n_{\text{train}}/N$ denote the split ratio. This allocation is governed by the bias-variance trade-off formalized in Section \ref{sec:theory}:

\begin{itemize}
    \item \textbf{Approximation Error ($\rho \to 0$):} As the training set shrinks, the conformal predictor is trained on insufficient data. Its conditional coverage properties may diverge from its asymptotic behavior on the full dataset.
    \item \textbf{Estimation Error ($\rho \to 1$):} As the evaluation set shrinks, the reliability estimator $\hat{\eta}$ is constrained by data scarcity, preventing it from effectively learning the decision boundary of the coverage indicator. This results in high-variance estimates of conditional validity.
\end{itemize}

\subsection{Empirical Results}

To empirically characterize this trade-off, we evaluated the ranking consistency (Weighted Kendall's $\tau_w$) across all four DGPs varying $\rho \in \{0.1, 0.3, 0.5, 0.7, 0.9\}$. Figure~\ref{fig:split_tradeoff} presents the resulting performance trajectories.

\begin{figure}[t]
     \centering
     \includegraphics[width=0.75\textwidth]{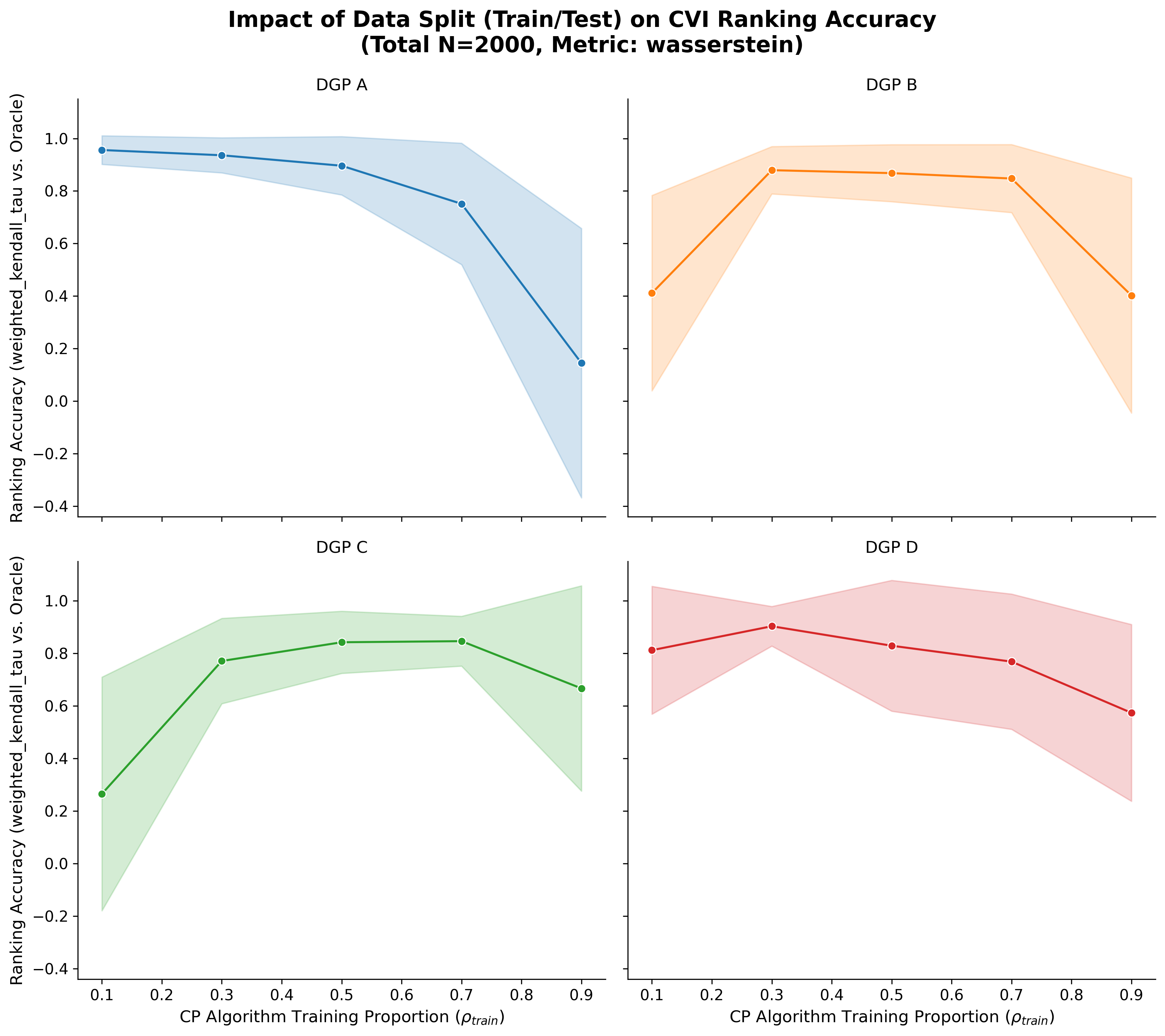}
     \caption{Impact of the data split ratio $\rho = n_{\text{train}}/N$ on assessment accuracy (Weighted Kendall's $\tau_w$). While the complex scenarios (Settings B, C, D) exhibit a characteristic concave profile peaking near $\rho=0.5$, the simple linear scenario (Setting A) achieves optimal performance with minimal training data ($\rho=0.1$).}
     \label{fig:split_tradeoff}
\end{figure}

The results generally exhibit a distinct \textbf{concave profile}, particularly for the complex scenarios (Settings B, C, and D). In these regimes, extreme splits ($\rho=0.1$ or $\rho=0.9$) lead to a degradation in ranking correlation, confirming that limiting the sample size for either the subject model or the assessor compromises the overall evaluation fidelity. 

Performance consistently peaks in the balanced regime ($\rho \approx 0.5$). However, Setting A presents a notable exception, where performance is maximized at $\rho=0.1$. This deviation is attributable to the \textbf{low sample complexity} of the linear, homoscedastic data generating process; the conformal predictors converge rapidly, allowing the majority of samples to be allocated to the reliability estimator to suppress estimation variance in a low-signal environment. For real-world data where the underlying structure is unknown, an equal allocation ($\rho=0.5$) provides a \textbf{robust heuristic} to balance predictive quality and audit precision.

\section{Empirical Investigation of Sample Complexity}
\label{app:sample_complexity}
In this appendix, we address a key practical consideration for deploying the CPA framework: determining the minimum sample size ($N$) required to effectively train the reliability estimator and facilitate robust model selection.

\subsection{Experimental Protocol}

Our investigation focuses on the \textbf{heteroscedastic regression scenario (Setting C)}, representing the most demanding regime for conditional coverage auditing due to the covariate-dependent noise structure. We adopt a rigorous ``Selection-Deployment'' protocol:

\begin{itemize}
    \item \textbf{Data Configuration:} We systematically vary the total budget $N \in [200, 2000]$ across different feature dimensionalities $p \in \{10, 20, 50\}$.
    \item \textbf{CPA Implementation:} Following our heuristic recommendations, we employ a balanced split ($\rho=0.5$). The reliability estimator $\hat{\eta}(x)$ is trained via the AutoML pipeline described in Appendix \ref{app:estimator_configs} to ensure maximally flexible modeling of the coverage boundary.
    \item \textbf{Oracle Benchmark:} Ground-truth rankings are derived by evaluating all candidate methods on the full data distribution using a large independent test set ($N_{\text{test}}=2000$) and the known DGP parameters.
    \item \textbf{Performance Metrics:} Assessment fidelity is quantified via \textbf{Weighted Kendall's $\tau_w$} (ranking alignment), \textbf{Hit@1} (precision of best-model identification), and \textbf{NDCG@1} (selection utility).
\end{itemize}

\subsection{Results and Analysis}

Figure~\ref{fig:sample_complexity} illustrates the empirical performance trajectories of the CPA framework as a function of the total sample size $N$. 

\begin{figure}[htbp]
    \centering
    \includegraphics[width=\textwidth]{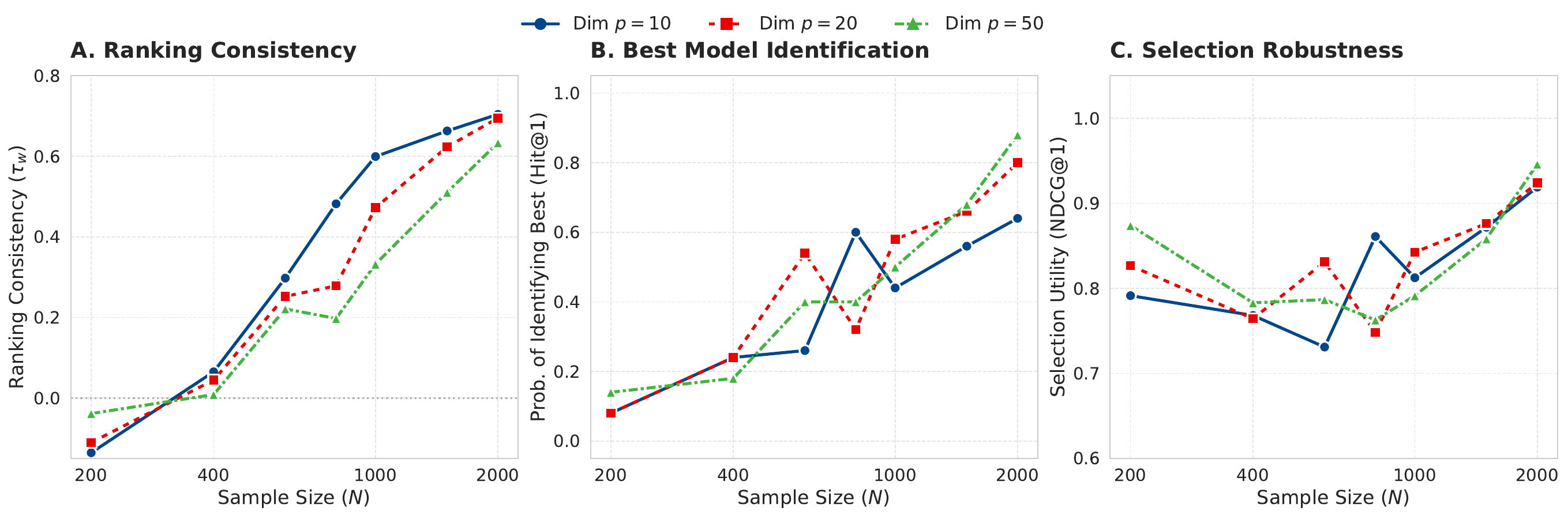}
    \caption{\textbf{Sample complexity analysis of CPA in heteroscedastic settings.} 
    Panel A illustrates the learning curve of ranking consistency ($\tau_w$). 
    Panel B shows the probability of identifying the optimal model (Hit@1). 
    Panel C demonstrates the robustness of the selection utility (NDCG@1). 
    The lines represent the mean performance over 50 replications.}
    \label{fig:sample_complexity}
\end{figure}

\medskip

\noindent \textbf{1. Threshold Behavior and Recommended Sample Size.} 
The primary objective of this simulation is to identify the critical sample size required for statistically meaningful evaluations. Panel A reveals a clear threshold behavior in ranking consistency ($\tau_w$). In the limited-sample regime ($N < 400$), estimation variance dominates the signal. As $N$ surpasses 600, the signal-to-noise ratio improves sufficiently to resolve structural coverage gaps, with $\tau_w$ stabilizing above 0.5 at $N = 1000$. \textbf{Based on these findings, we recommend practitioners allocate a minimum of $n_{\text{eval}} \ge 400$ evaluation samples} (corresponding to $N \ge 800$ under a $\rho=0.5$ split) to ensure a robust and reliable CPA audit.

\medskip

\noindent \textbf{2. Dimensionality Modulates the Threshold.} 
Panel B indicates that feature dimensionality shifts this empirical threshold. Lower-dimensional settings ($p=10$) converge earlier, whereas higher-dimensional spaces ($p=50$) demand larger sample sizes to initiate accurate identification. Nevertheless, once the sample complexity requirement is satisfied ($N \ge 1500$), higher dimensions ultimately yield superior identification accuracy, as the true performance disparities between adaptive and rigid methods become statistically discernible.

\medskip

\noindent \textbf{3. Utility Under Sub-optimal Sample Sizes.} 
Finally, Panel C highlights the framework's robustness under finite-sample constraints. Even below the recommended threshold ($N \approx 600$), where global ranking consistency is still developing, NDCG@1 scores reliably exceed 0.75. This confirms that while the framework may lack the precision to perfectly rank all candidates given limited data, it retains sufficient discriminatory power to consistently isolate high-utility models and discard severe under-performers.

\section{Additional Details for Real-Data Experiments}
\label{appendix:real-data-details}

\subsection{Worst Slab Coverage}
\label{app:wsc_details}
We provide additional details on the \emph{Worst Slab Coverage} (WSC) metric, originally proposed to assess conditional-like coverage properties beyond marginal guarantees. Let $X \in \mathbb{R}^d$ denote the feature vector and $Y$ the response. For a given direction $v \in \mathbb{R}^d$ and real numbers $a < b$, define a \emph{slab}
\[
S_{v,a,b} := \{ x \in \mathbb{R}^d : a \le v^\top x \le b \}.
\]
Slabs provide a flexible yet computationally tractable family of subsets of the feature space, allowing us to probe coverage behavior along low-dimensional projections of $X$. Given a prediction or confidence set $\widehat{C}(\cdot)$ and a threshold $0 < \delta \le 1$, the worst slab coverage along direction $v$ is defined as
\[
\mathrm{WSC}_n(\widehat{C}, v)
:=
\inf_{a < b}
\left\{
P_n\!\left( Y \in \widehat{C}(X) \mid a \le v^\top X \le b \right)
\;:\;
P_n(a \le v^\top X \le b) \ge \delta
\right\},
\]
where $P_n$ denotes the empirical distribution of the observed sample $\{(X_i,Y_i)\}_{i=1}^n$.
Intuitively, $\mathrm{WSC}_n(\widehat{C}, v)$ captures the \emph{worst-case empirical coverage} over all slabs along direction $v$ that contain at least a $\delta$ fraction of the data.

This criterion can be interpreted as an intermediate notion between marginal and full conditional coverage.
While exact conditional coverage is generally unattainable without strong assumptions, WSC detects systematic coverage failures that may occur on structured subsets of the feature space, even when marginal coverage is satisfied.

In practice, WSC is computed using a finite collection of randomly sampled directions.
For each direction, the feature vectors are projected onto a one-dimensional axis and sorted, and the empirical coverage is evaluated over all contiguous intervals in the projected space that contain at least a $\delta$ fraction of the observations.
The minimum coverage over such intervals defines the worst slab coverage for that direction, and the overall WSC is obtained by taking the minimum across all sampled directions.

We consider two variants of this procedure.
The first is an in-sample (biased) version, where the same dataset is used both to identify the adversarial direction and slab and to evaluate the resulting coverage.
The second is an out-of-sample (unbiased) version based on sample splitting: the adversarial direction and slab are selected using a training subset, and the coverage is subsequently evaluated on a held-out test subset restricted to that slab.

In our experiments, we observe that the WSC metric is highly sensitive to the choice of projection directions.
As a consequence, the in-sample and out-of-sample variants can yield substantially different numerical values.
In particular, the sample-splitting (unbiased) version often produces much more conservative estimates, with WSC values substantially higher and frequently very close to the nominal target level.
While this behavior is expected to some extent due to the reduced adaptivity in the slab selection step, it may also mask localized coverage failures that WSC is designed to detect.
For this reason, we primarily rely on the in-sample (biased) WSC as a diagnostic tool in our empirical studies, as it more effectively reveals potential worst-case coverage deficiencies across directions.

\subsection{Detailed Numerical Results for Real Data Experiments}
\label{app:detailed_results}

Table \ref{tab:full_results_appendix} reports the mean values (over 10 repetitions) for Average Interval Length, Marginal Coverage, and Worst-Slab Coverage (WSC) for all candidate methods across the nine datasets.

\begin{table}[H]
\centering
\renewcommand{\arraystretch}{} %
\caption{Real Data Experiment Results (Mean Values)}
\label{tab:full_results_appendix}
\makebox[\textwidth][c]{%
\scalebox{.8}{%
\begin{tabular}{ll|>{\columncolor{gray!25}}r|ccccccc}
\toprule
 & Method & Selected & Bootstrap & CV+ & CP-Residual & CP-Studentized & CQR & LCP & RLCP \\
Dataset & Metric &  &  &  &  &  &  &  &  \\
\midrule
\multirow[t]{3}{*}{Bike} & Length & 161.847 & 230.635 & 217.013 & 234.722 & 182.445 & 156.291 & \textbf{128.791} & 138.572 \\
 & Coverage & 0.898 & 0.906 & 0.905 & 0.896 & 0.897 & 0.898 & 0.898 & 0.906 \\
 & WSC-B & 0.782 & 0.747 & 0.753 & 0.735 & \textbf{0.784} & 0.780 & 0.721 & 0.758 \\
\cline{1-10}
\multirow[t]{3}{*}{Computer} & Length & 7.234 & 7.892 & 7.711 & 7.925 & \textbf{7.234} & 7.877 & 7.314 & 7.267 \\
 & Coverage & 0.895 & 0.900 & 0.898 & 0.895 & 0.895 & 0.900 & 0.897 & 0.910 \\
 & WSC-B & 0.769 & 0.704 & 0.683 & 0.669 & 0.769 & 0.753 & 0.718 & \textbf{0.774} \\
\cline{1-10}
\multirow[t]{3}{*}{Debutanizer} & Length & 0.254 & 0.253 & 0.244 & 0.274 & \textbf{0.235} & 0.303 & 0.263 & 0.262 \\
 & Coverage & 0.911 & 0.910 & 0.909 & 0.903 & 0.905 & 0.906 & 0.900 & 0.916 \\
 & WSC-B & 0.680 & 0.635 & 0.653 & 0.611 & \textbf{0.670} & 0.637 & 0.625 & 0.668 \\
\cline{1-10}
\multirow[t]{3}{*}{Kin8nm} & Length & 0.481 & 0.480 & 0.476 & 0.504 & \textbf{0.473} & 0.566 & 0.476 & 0.500 \\
 & Coverage & 0.899 & 0.904 & 0.905 & 0.904 & 0.898 & 0.901 & 0.898 & 0.901 \\
 & WSC-B & 0.765 & 0.742 & 0.744 & 0.740 & 0.761 & \textbf{0.772} & 0.720 & 0.739 \\
\cline{1-10}
\multirow[t]{3}{*}{Meps\_21} & Length & 23.276 & 34.378 & 34.431 & 34.525 & \textbf{23.276} & 28.222 & 29.574 & 35.332 \\
 & Coverage & 0.902 & 0.902 & 0.904 & 0.902 & 0.902 & 0.955 & 0.904 & 0.899 \\
 & WSC-B & 0.804 & 0.478 & 0.457 & 0.442 & \textbf{0.804} & 0.883 & 0.518 & 0.457 \\
\cline{1-10}
\multirow[t]{3}{*}{Miami\_2016} & Length & 0.446 & 0.484 & 0.483 & 0.498 & \textbf{0.446} & 0.523 & 0.469 & 0.500 \\
 & Coverage & 0.900 & 0.908 & 0.909 & 0.899 & 0.900 & 0.901 & 0.901 & 0.910 \\
 & WSC-B & 0.804 & 0.696 & 0.692 & 0.655 & \textbf{0.804} & 0.782 & 0.728 & 0.753 \\
\cline{1-10}
\multirow[t]{3}{*}{Parkinsons} & Length & 16.992 & 17.877 & 17.866 & 19.300 & 16.946 & 17.379 & 11.335 & \textbf{11.249} \\
 & Coverage & 0.896 & 0.898 & 0.900 & 0.896 & 0.898 & 0.902 & 0.899 & 0.905 \\
 & WSC-B & 0.740 & 0.714 & 0.715 & 0.712 & \textbf{0.743} & 0.738 & 0.725 & 0.734 \\
\cline{1-10}
\multirow[t]{3}{*}{Qsar} & Length & 2.665 & 2.688 & \textbf{2.640} & 2.820 & 2.672 & 3.467 & 2.830 & 2.861 \\
 & Coverage & 0.905 & 0.909 & 0.910 & 0.898 & 0.905 & 0.898 & 0.897 & 0.897 \\
 & WSC-B & 0.761 & \textbf{0.763} & 0.757 & 0.734 & 0.762 & 0.753 & 0.732 & 0.736 \\
\cline{1-10}
\multirow[t]{3}{*}{Temperature} & Length & 2.428 & 2.453 & 2.431 & 2.604 & 2.443 & 3.010 & \textbf{2.341} & 2.716 \\
 & Coverage & 0.904 & 0.902 & 0.906 & 0.901 & 0.902 & 0.908 & 0.903 & 0.904 \\
 & WSC-B & 0.776 & 0.740 & 0.738 & 0.742 & \textbf{0.778} & 0.767 & 0.758 & 0.759 \\
\cline{1-10}
\bottomrule
\end{tabular}}
} 
\par\medskip
\small
\textit{Note:} The CQR method was excluded from the best-value comparison for the Meps\_21 dataset due to its abnormal marginal coverage (0.955).
\end{table}

\subsection{Multivariate Reliability Landscapes for Bike Sharing Demand}
\label{app:bike-multivariate}

We extend the Bike Sharing analysis to explore interacting covariates in Figure~\ref{fig:bike-temp-hour}. The visual alignment between the empirical coverage (left column) and the CPA-estimated reliability surface (middle column) confirms that our framework reconstructs complex failure boundaries. These estimated landscapes provide a granular diagnostic comparison of the three conformal procedures. First, \textbf{CP-Residual (top row)}, which imposes a homoscedastic assumption, exhibits substantial undercoverage (red regions, dropping below 0.6) during rush hours (8--10 and 16--19), particularly on warm days ($>20^\circ\text{C}$). Conversely, it yields overcoverage (dark blue) at night (hours 0--5). Operationally, casual ridership surges on warm evenings, amplifying the conditional variance; a constant-width interval is thus conservative during dormant periods but insufficient during peak volatility. Second, \textbf{CP-Studentized (middle row)} scales intervals with local variance estimates. While undercoverage regions become shallower, systematic deficits persist, indicating that simple variance scaling struggles to fully capture the complex interaction between temperature and time.

Finally, \textbf{CQR (bottom row)} directly fits conditional quantiles, effectively adapting to the underlying heteroscedasticity. The undercoverage zones are largely mitigated, yielding a landscape that consistently hovers near the nominal target across the bivariate feature space. Crucially, these multivariate landscapes transition model assessment from aggregate marginal summaries to localized failure characterization, enabling practitioners to identify specific operational regimes that necessitate targeted interventions. \looseness=-1

\begin{figure}[htbp]
  \centering
  \includegraphics[width=\textwidth]{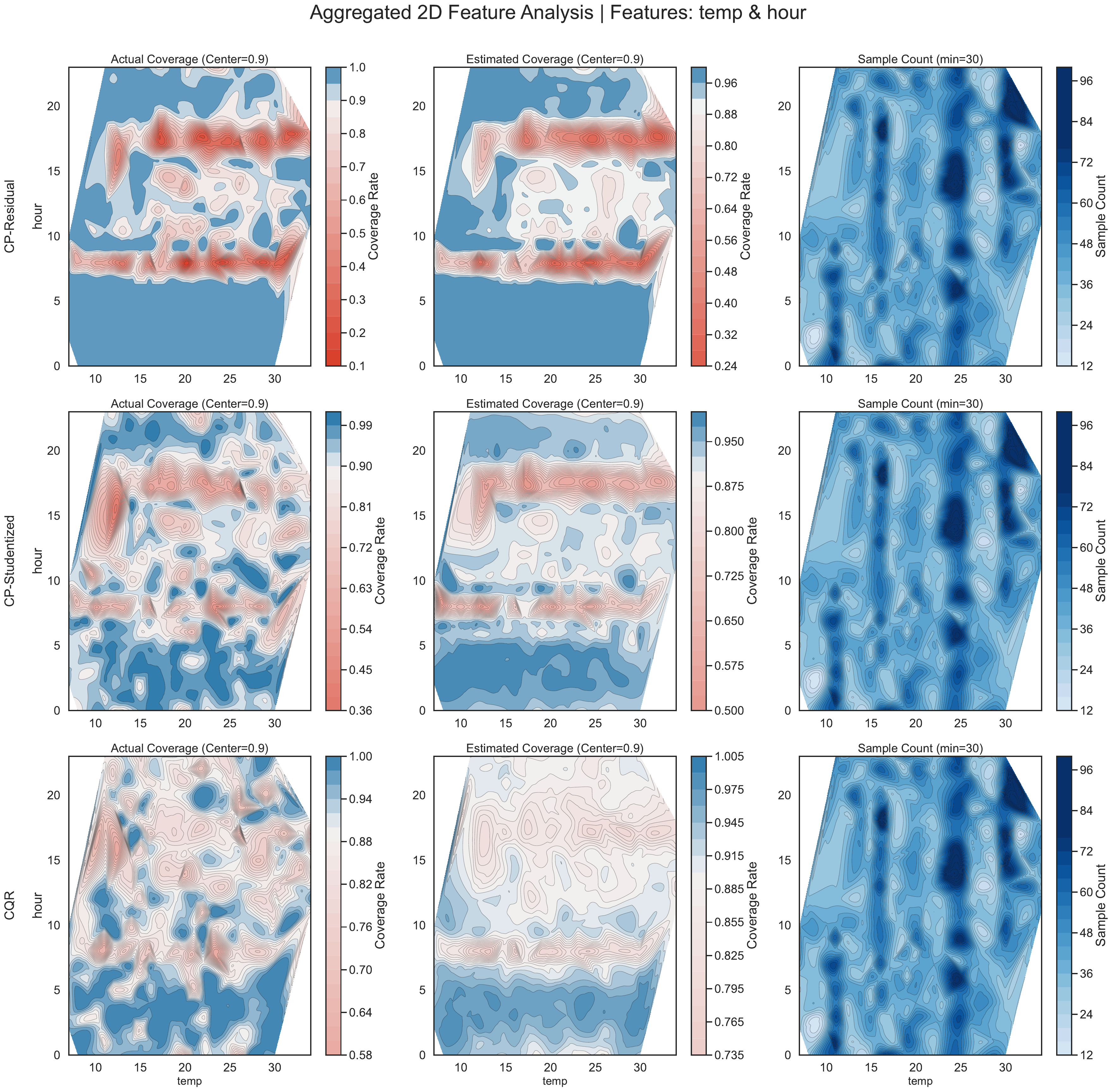}
  \caption{\textbf{Reliability landscape: Temperature vs. Hour.} Rows correspond to CP-Residual (top), CP-Studentized (middle), and CQR (bottom). Columns display the smoothed empirical coverage (left), the CPA-predicted reliability $\hat{\eta}(x)$ (middle), and the sample density (right). Red indicates undercoverage; blue indicates overcoverage.}
  \label{fig:bike-temp-hour}
\end{figure}

\end{document}